\documentclass[reqno,a4paper,11pt]{article}
\pdfoutput=1
\usepackage{xcolor}

\usepackage{graphicx}
\usepackage[textwidth = 430 pt, textheight = 630 pt]{geometry}

\definecolor{MyDarkBlue}{rgb}{0.15,0.25,0.45}
\usepackage{epsfig,rotating}
\usepackage{amsmath,amssymb}
\usepackage{amsfonts}
\usepackage{mathrsfs}
\usepackage{bbm}
\usepackage[normalem]{ulem}
\usepackage[T1]{fontenc}

\usepackage{latexsym}
\usepackage{amsthm}
\usepackage[all,knot]{xy}
\xyoption{arc}

\usepackage[utf8x]{inputenc}

\usepackage{hyperref}
\hypersetup{
hypertexnames=false,
colorlinks=true,
citecolor=MyDarkBlue,
linkcolor=MyDarkBlue,
urlcolor=MyDarkBlue,
pdfauthor={},
pdftitle={},
pdfsubject={hep-th math-ph}
breaklinks=true
}

\usepackage{tikz}
\usepackage{mathtools}

\let\fn\footnote
\renewcommand{\footnote}[1]{\linespread{1.1}\fn{#1}\linespread{1.29}}

\makeatletter\renewcommand{\section}{\@startsection
{section}{1}{\z@}{-3.5ex plus -1ex minus
    -.2ex}{2.3ex plus .2ex}{\bf }}
\makeatletter\renewcommand{\subsection}{\@startsection{subsection}{2}{\z@}{-3.25ex
plus -1ex minus
   -.2ex}{1.5ex plus .2ex}{\bf }}
\makeatletter\renewcommand{\subsubsection}{\@startsection{subsubsection}{3}{-2.45ex}{-3.25ex
plus -1ex minus -.2ex}{1.5ex plus .2ex}{\it }}
\renewcommand{\thesection}{\arabic{section}}
\renewcommand{\thesubsection}{\arabic{section}.\arabic{subsection}}
\renewcommand{\@seccntformat}[1]{\@nameuse{the#1}.~~}

\renewcommand{\theequation}{\thesection.\arabic{equation}}
\makeatletter \@addtoreset{equation}{section}

\usepackage[toc,page]{appendix}

\newtheorem{thm}{Theorem}[section]
\renewcommand{\thethm}{\thesection.\arabic{thm}}
\newtheorem{lemma}[thm]{Lemma}
\newtheorem{definition}[thm]{Definition}
\newtheorem{theorem}[thm]{Theorem}
\newtheorem{proposition}[thm]{Proposition}
\newtheorem{corollary}[thm]{Corollary}
\newtheorem{remark}[thm]{Remark}

\newcommand{\myxymatrix}[1]{\vcenter{\vbox{\xymatrix{#1}}}}

\renewcommand{\appendices}{
\section*{Appendix}\label{appendices}\setcounter{subsection}{0}
\addcontentsline{toc}{section}{Appendix}
\setcounter{equation}{0}
\makeatletter
\renewcommand{\theequation}{\Alph{subsection}.\arabic{equation}}
\renewcommand{\thesubsection}{\Alph{subsection}}
\renewcommand{\thethm}{\Alph{subsection}.\arabic{thm}}
\@addtoreset{equation}{subsection}
\@addtoreset{thm}{subsection}
\makeatother
}

\def\slasha#1{\setbox0=\hbox{$#1$}#1\hskip-\wd0\hbox to\wd0{\hss\sl/\/\hss}}

\def\periodb#1{\setbox0=\hbox{$#1$}#1\hskip-\wd0\hbox to\wd0{-}}

\newcommand{\unit}{\mathbbm{1}}   			% identity map/matrix
\newcommand{\id}{\mathrm{id}}   			% identity map/matrix
\newcommand{\CC}{\mathcal{C}}
\newcommand{\CCC}{\mathscr{C}}

\newcommand{\CH}{\mathcal{H}}

\newcommand{\CL}{\mathcal{L}}
\newcommand{\CM}{\mathcal{M}}

\newcommand{\CQ}{\mathcal{Q}}

\newcommand{\CR}{\mathcal{R}}

\newcommand{\CT}{\mathcal{T}}
\newcommand{\CCT}{\mathscr{T}}

\newcommand{\CV}{\mathcal{V}}

\newcommand{\CCX}{\mathscr{X}}

\newcommand{\CE}{\mathcal{E}}
\newcommand{\frg}{\mathfrak{g}}				% frak-letters
\newcommand{\frU}{\mathfrak{U}}

\newcommand{\frX}{\mathfrak{X}}

\newcommand{\FR}{\mathbbm{R}}     			% field of real numbers
\newcommand{\NN}{\mathbbm{N}}     			% set of natural numbers
\newcommand{\RZ}{\mathbbm{Z}}     			% ring of integers
\newcommand{\dd}{\mathrm{d}}     			% total differential
\newcommand{\dpar}{\partial}     			% partial differential
\newcommand{\embd}{{\hookrightarrow}}     		% embedded
\newcommand*{\longembd}{\ensuremath{\lhook\joinrel\relbar\joinrel\rightarrow}}
\newcommand{\de}{\mathrm{e}}     			% Euler's number
\newcommand{\di}{\mathrm{i}}     			% imaginary unit
\newcommand{\eand}{{\qquad\mbox{and}\qquad}}     		% and etc. in equations
\newcommand{\ewith}{{\qquad\mbox{with}\qquad}}

\newcommand{\der}[1]{\frac{\dpar}{\dpar #1}}   		% partielle ableitung, 1 argument
\newcommand{\derr}[2]{\frac{\dpar #1}{\dpar #2}}   	% partielle ableitung, 2 argumente
\newcommand{\au}{\mathfrak{u}}

\newcommand{\ao}{\mathfrak{o}}

\newcommand{\sU}{\mathsf{U}}     			% groups

\newcommand{\sG}{\mathsf{G}}
\newcommand{\sL}{\mathsf{L}}
\newcommand{\sT}{\mathsf{T}}

\newcommand{\sLie}{\mathsf{Lie}}

\newcommand{\sH}{\mathsf{H}}

\newcommand{\sGL}{\mathsf{GL}}

\newcommand{\sO}{\mathsf{O}}

\newcommand{\acton}{\vartriangleright}     			% span
\renewcommand{\remark}[1]{}     				% remark
\def\tyng(#1){\hbox{\tiny$\yng(#1)$}}			% small Young diagram
\def\tyoung(#1){\hbox{\tiny$\young(#1)$}}			% small Young diagram
\newcommand{\beq}{\begin{eqnarray}}
\newcommand{\eeq}{\end{eqnarray}}

\begin{document}
\begin{titlepage}
\begin{flushright}
 ITP--UH--22/16\\
 EMPG--16--18
\end{flushright}
\vskip 1.0cm
\begin{center}
{\LARGE \bf Extended Riemannian Geometry I:\\[0.2cm] Local Double Field Theory}
\vskip 1.cm
{\Large Andreas Deser$^{a}$ and Christian S\"amann$^b$}
\setcounter{footnote}{0}
\renewcommand{\thefootnote}{\arabic{thefootnote}}
\vskip 1cm
{\em${}^a$ Institut f\"ur Theoretische Physik \\
Appelstr.\ 2\\
30167 Hannover, Germany
}\\
and\\
{\em Istituto Nationale di Fisica Nucleare \\
Via P.~Giuria 1\\
10125 Torino, Italy
}\\[0.5cm]
{\em${}^b$
Department of Mathematics,
Heriot-Watt University\\
Colin Maclaurin Building, Riccarton, Edinburgh EH14 4AS, U.K.}\\
and\\  {\em Maxwell Institute for Mathematical Sciences, Edinburgh,
  U.K.} \\ and \\ {\em Higgs Centre for Theoretical Physics,
  Edinburgh, U.K.}\\[0.5cm]
{Email: {\ttfamily deser@to.infn.it , c.saemann@hw.ac.uk}}
\end{center}
\vskip 1.0cm
\begin{center}
{\bf Abstract}
\end{center}
\begin{quote}
We present an extended version of Riemannian geometry suitable for the description of current formulations of double field theory (DFT). This framework is based on graded manifolds and it yields extended notions of symmetries, dynamical data and constraints. In special cases, we recover general relativity with and without 1-, 2- and 3-form gauge potentials as well as DFT. We believe that our extended Riemannian geometry helps to clarify the role of various constructions in DFT. For example, it leads to a covariant form of the strong section condition. Furthermore, it should provide a useful step towards global and coordinate invariant descriptions of T- and U-duality invariant field theories.
\end{quote}
\end{titlepage}

\begin{center}
 
\end{center}

\tableofcontents

\newpage

\section{Introduction and results}

One of the most prominent features differentiating string theory from field theories of point particles is the symmetry known as T-duality \cite{Giveon:1994fu}. This symmetry interchanges the momentum modes of a string with its winding modes along compact cycles and it is one of the main ingredients in the web of dualities connecting the various different string theories. Considering the success of low-energy space-time descriptions of string theory in terms of ordinary field theories of the massless string modes, it is tempting to ask if there is some manifestly T-dual completion of these descriptions. It is the aim of double field theory (DFT) to provide such a T-duality invariant version of supergravity. The idea of DFT goes back to the early 90ies \cite{Tseytlin:1990va,Siegel:1993th,Siegel:1993xq} and the development of double geometry \cite{Hull:2004in,Hull:2009sg} led to the paper \cite{Hull:2009mi}, which gave DFT its name and seems to have triggered most of the current interest in this area. A detailed review of DFT is found in the overview papers \cite{Zwiebach:2011rg,Berman:2013eva,Aldazabal:2013sca,Hohm:2013bwa}.

In this paper, we shall not concern ourselves with the issue of how reasonable the main objective of DFT is. Rather, we would like to explain some of the mathematical structures involved in currently available formulations. These formulations are local and linear in various senses, and we believe that the mathematical picture we provide points towards global descriptions. In particular, we give an appropriate extension of Hitchin's generalized geometry \cite{Hitchin:2004ut,Hitchin:2005in,Gualtieri:2003dx}.

There have been a number of previous papers clarifying the underlying geometrical structures from various perspectives \cite{Vaisman:2012ke,Vaisman:2012px,Hohm:2012mf,Cederwall:2014kxa,Blumenhagen:2014gva,Deser:2014mxa,Bakas:2016nxt}. We also want to draw attention to the phase space perspective described in \cite{Aschieri:2015roa}. 

In this paper, we choose to follow a different approach. The $B$-field of string theory is well-known to be part of the connective structure of an abelian gerbe \cite{Gawedzki:1987ak,Freed:1999vc}, which is a categorified principal $\sU(1)$-bundle. In this picture, we have two kinds of symmetries: the diffeomorphisms on the base manifolds as well as the gauge transformations on the connective structure of the abelian gerbe. It is natural to expect that together, they form a categorified Lie group, or Lie 2-group for short. At the infinitesimal level, we have a Lie 2-algebra, which is most concisely described by a symplectic N$Q$-manifold. The latter objects are simply symplectic graded manifolds $(\CM,\omega)$ endowed with a vector field $Q$ of degree 1 satisfying $Q^2=0$ and $\CL_Q\omega=0$. They are familiar to physicists from BV-quantization and string field theory \cite{Lada:1992wc} and mathematicians best think of them as symplectic $L_\infty$-algebroids.

Our expectation is confirmed by the fact that the symplectic N$Q$-manifold known as Courant algebroid (a symplectic Lie 2-algebroid) already features prominently in generalized geometry and various other examples appear in mathematical discussions of $T$-duality. The Lie 2-algebra of infinitesimal symmetries can here be described by a derived bracket construction on the Courant algebroid \cite{Roytenberg:1998vn,Roytenberg:1999aa} involving the Poisson structure induced by the symplectic form as well as the Hamiltonian of the homological vector field $Q$. It has long been known that also the Lie bracket of vector fields is a derived bracket $\iota_{[X,Y]}=[\iota_X,[\dd,\iota_Y]]$, cf.\ \cite{Kosmann-Schwarzbach:0312524}. It is therefore natural to ask if even the Lie 2-algebra underlying double field theory can be realized by a derived bracket construction. It was shown in \cite{Deser:2014mxa} that the C-bracket of double field theory,  which is part of the Lie 2-algebra structure, is a derived bracket. Moreover  by applying the easiest form of the strong section condition $\tilde \partial = 0$, it reduces to the Courant bracket of generalized geometry.

To develop the full Lie 2-algebra, however, one quickly realizes that ordinary symplectic N$Q$-manifolds $(\CM,\omega,Q)$ are not sufficient. One rather needs to lift the condition $Q^2=0$ and subsequently restrict the algebra of functions $\CC^\infty(\CM)$ to a suitable subset. We will develop this more general picture of derived brackets in great detail, studying what we call pre-N$Q$-manifolds. These pre-N$Q$-manifolds still carry categorified Lie algebras induced by derived brackets on restricted sets of functions $\CC^\infty(\CM)$, and they are the right extension of Hitchin's generalized geometry to capture the geometry of local double field theory.

We tried to make our presentation self-contained and in particular to provide a very detailed introduction to the language of N$Q$-manifolds, which might not be well-known by people studying double field theory. We therefore begin with reviews of generalized geometry, double field theory and N$Q$-manifolds in sections \ref{sec:review} and \ref{sec:NQ}. Note that unless stated otherwise, we always focus on local structures and spacetimes are contractible. 

Our actual discussion then starts in section \ref{sec:extendedRG} with the definition of pre-N$Q$-manifolds and the identification of an $L_\infty$-algebra structure arising from derived brackets. This yields the notion of extended vector fields generating infinitesimal extended symmetries as well as a restriction of the algebra of functions to a subset. We then introduce extended tensors and define an action of the $L_\infty$-algebra structure on those. In this picture, a natural notion of extended metric arises, together with candidate terms for an invariant action functional.

This rather abstract discussion is then filled with life by studying a number of examples in section \ref{sec:examples}. We show that the pre-N$Q$-manifolds suitable for the description of the symmetries of Einstein-Hilbert-Deligne actions, by which we mean general relativity minimally coupled to $n$-form gauge potentials, are the Vinogradov algebroids $T^*[n]T[1]M$. The fact that the symmetries are recovered appropriately might not be too surprising to people familiar with N$Q$-manifolds; new in this context is the construction of DFT-like action functionals for Einstein-Hilbert-Deligne theories.

Our formalism develops its real strength when applied to double field theory. The relevant pre-N$Q$-manifold here, $\CE_2(M)$, is obtained in a rather nice way as a half-dimensional submanifold from the Vinogradov algebroid $T^*[2]T[1]X$ for the doubled space $X=T^*M$. That is, double field theory can be regarded as a restriction of generalized geometry on a doubled space.

The action of extended symmetries on extended tensors is then indeed the generalized Lie derivative familiar from double field theory. Moreover, the restrictions imposed on extended functions, vectors and tensors are coordinate invariant versions of the strong section condition of double field theory. This is probably the most interesting aspect of our formalism presented in this paper: Since our restrictions arise classically as a condition for obtaining the right symmetry structures, they automatically yield well-defined, covariant and transparent conditions on all types of fields. Also, the usual examples of solutions to the strong section condition yield indeed correct restrictions on extended fields.

It is interesting to note that our restrictions are also slightly weaker for tensor fields than the usual strong section condition of double field theory. The latter is found to be too strong in various contexts \cite{Blumenhagen:2016vpb}, and our slight weakening might provide a resolution to this issue.

Another point that becomes more transparent in our formalism is the problem of patching together the local descriptions of double field theory to a global framework and we will comment in more detail on this point in section \ref{ssec:global}. In \cite{Berman:2014jba}, the integrated action of the generalized Lie derivative on fields was used to patch local doubled fields to global ones. This, however, seems to be consistent only if the curvature of the $B$-field is globally exact \cite{Papadopoulos:2014mxa}, which implies that the underlying gerbe is topologically trivial. From our perspective, this is rather natural as it is already known in generalized geometry that the Courant algebroid needs to be twisted by a closed 3-form $T$ to describe accurately the symmetries of a gerbe with Dixmier-Douady class $T$.

In our formalism, it is clear how the generalized Lie derivative of double field theory should be twisted and we give the relevant definitions and an initial study of such twists in section \ref{ssec:twisted_extended_symmetries}. The resulting symmetries are the ones that should be used to patch together local descriptions of double field theory in the case involving topologically non-trivial gerbes.

Finally, we also comment on the definitions of covariant derivative, torsion and Riemann curvature tensors. The problem here is that the Lie bracket in the various definitions should be replaced by a Lie 2-algebra bracket, which in general breaks linearity of the tensors with respect to multiplication of extended vector fields by functions. After introducing an extended covariant derivative, we can write the Gualtieri torsion of generalized geometry \cite{Gualtieri:2007bq} in a nice form using our language. Recall that the Gualtieri torsion is also the appropriate notion of a torsion tensor in double field theory \cite{Hohm:2012mf}. Our expression for the torsion contains in particular the ordinary torsion tensor as a special case. The same holds for the Riemann tensor of double field theory \cite{Hohm:2012mf}, see also \cite{Jeon:1011.1324,Jeon:2011cn,Hohm:2011si,Jurco:2015xra}. It is noteworthy that in the discussion of the latter, our weakened and invariant strong section condition appears.

There are a number of evident open problems that we plan to attack in future work. First, an extension of our formalism suitable for exceptional field theory is in preparation \cite{Deser:2017aa}. Second, we would like to develop a better understanding of the extension of Riemannian geometry underlying the actions of double and exceptional field theory using our language. Third, one should study our twists of the generalized Lie derivative in more detail, which can be considered as a stepping stone to fourth: developing an understanding of an appropriate global picture for double field theory. Since our approach showed that the local description of classical double field theory is mathematically consistent and reasonable, we have no more reason to doubt the existence of a global description.

\section{Lightning review of double field theory}\label{sec:review}

Let us very briefly sum up the key elements of local generalized geometry and local double field theory, since we shall reproduce them in our formalism from a different starting point.

\subsection{Generalized geometry}\label{ssec:generalised_geometry}

Any target space description of classical strings has to include the massless excitations of the closed string. These consist of the spacetime metric $g$, the Kalb-Ramond $B$-field and the dilaton $\phi$. The former two can be elegantly described in Hitchin's generalized geometry \cite{Hitchin:2004ut,Hitchin:2005in,Gualtieri:2003dx}. Underlying this description is a vector bundle $E_2\cong TM\oplus T^*M$ over some $D$-dimensional manifold $M$, which fits into the short exact sequence
\begin{equation}
 0 \longrightarrow T^*M \longembd E_2\xrightarrow{~\rho~} TM \longrightarrow 0~.
\end{equation}
Sections $\sigma$ of the map $\rho$ are given by rank 2 tensors, which can be split into their symmetric and antisymmetric parts, $\sigma_{\mu\nu}=g_{\mu\nu}+B_{\mu\nu}$. Moreover, sections of $E_2$ describe the infinitesimal symmetries of the fields $g,B,\phi$ and they are encoded in a vector field $X$ capturing infinitesimal diffeomorphisms and a 1-form $\alpha$ describing the gauge symmetry. Note that $E_2$ comes with a natural inner product of signature $(D,D)$, which is given by
\begin{equation}
 (X_1+\alpha_1,X_2+\alpha_2)=\iota_{X_1}\alpha_2+\iota_{X_2}\alpha_1~.
\end{equation}
To write this more concisely, we introduce generalized vectors $X^M=(X^\mu,\alpha_\mu)$ with index $M=1,\ldots,2D$. Then we have
\begin{equation}\label{eq:dft_eta}
 (X_1,X_2)=X_1^MX_2^N\eta_{MN}\ewith \eta_{MN}=\eta^{MN}=\left(\begin{array}{cc} 0 & \unit \\ \unit & 0 \end{array}\right)~.
\end{equation}
Note that the metric $\eta$ reduces the structure group of the bundle $E_2$ to $\sO(D,D)$. The Lie algebra $\ao(D,D)$ consists of matrices of the form
\begin{equation}
 \left(\begin{array}{cc} A & B \\ \beta & -A^T \end{array}\right)~,
\end{equation}
where $A$ is an arbitrary matrix, generating $\sGL(D)\subset \sO(D,D)$, and $\beta\in \wedge^2 \FR^D$ and $B\in \wedge^2 (\FR^D)^*$ generate what are known as $\beta$-transformations and $B$-transformations in generalized geometry \cite{Gualtieri:2003dx,Grana:2008yw}. 

The transformation rule for the metric and the $B$-field can also be written in a homogeneous way, by introducing the {\em generalized metric}
\cite{Shapere:1988zv,Giveon:1988tt,Maharana:1992my,Gualtieri:2003dx}
\begin{equation}\label{eq:metric_H_O(D,D)}
 \CH_{MN}=\left(\begin{array}{cc} g_{\mu\nu}-B_{\mu\kappa}g^{\kappa\lambda}B_{\lambda\nu}  &  B_{\mu\kappa}g^{\kappa\nu}\\ -g^{\mu\kappa}B_{\kappa\nu}  &g^{\mu\nu}  \end{array}\right)~.
\end{equation}
Note that this metric is obtained from the ordinary metric $g$ via a finite $B$-field transformation with adjoint action
\begin{equation}\label{eq:B-field_trafo}
 \CH=\left(\begin{array}{cc} \unit & B \\ 0 & \unit \end{array}\right)\left(\begin{array}{cc} g & 0 \\ 0 & g^{-1} \end{array}\right)\left(\begin{array}{cc} \unit & B \\ 0 & \unit \end{array}\right)^T~.
\end{equation}
This explains the terminology, a $B$-field transformation modifies linearly the value of the $B$-field. Moreover, we have the relation $\CH^{-1}=\eta^{-1}\CH\eta^{-1}$.

Under infinitesimal diffeomorphisms and gauge transformations encoded in the generalized vector $X$, $\CH$ transforms according to
\begin{equation}\label{eq:action_generalized_Lie}
\delta_X \CH_{MN}=X^P\dpar_P\CH_{MN}+(\dpar_M X^P-\dpar^PX_M)\CH_{PN}+(\dpar_N X^P-\dpar^PX_N)\CH_{MP}~,
\end{equation}
where $\dpar_M=(\dpar_\mu,\dpar^\mu=0)$. Underlying this transformation is a {\em generalized Lie derivative} whose action, e.g.\ on tensors $T^M{}_N$, is defined as
\begin{equation}\label{eq:action_extended_Lie}
 \hat \CL_X T^M{}_N:=X^K\dpar_K T^M{}_N+(\dpar^MX_K-\dpar_KX^M)T^K{}_N+(\dpar_NX^K-\dpar^KX_N)T^M{}_K~.
\end{equation}
Note that $\hat\CL_X$ is compatible with the Courant bracket $[-,-]_C$, which is the natural bracket between generalized vector fields:
\begin{equation}
 \hat\CL_X\hat \CL_Y-\hat \CL_Y\hat \CL_X=\hat\CL_{[X,Y]_C}~,
\end{equation}
where
\begin{equation}\label{eq:Courant_bracket}
 \Bigl([X,Y]_C\Bigr)^M :=\,X^K\partial_K Y^M - Y^K\partial_K X^M -\tfrac{1}{2}\Bigl(X^K\partial^M Y_K - Y^K\partial^M X_K\Bigr)\;.
\end{equation}
Formulas \eqref{eq:action_generalized_Lie}, \eqref{eq:action_extended_Lie} and  \eqref{eq:Courant_bracket} clearly hint at a completion of the picture where $\dpar^\mu\neq 0$. Starting from $\dpar_M=(\dpar_\mu,\dpar^\mu=0)$ we can transform to a situation with $\dpar^\mu\neq 0$ by applying a global $\sO(D,D)$-transformations with non-trivial off-diagonal components $\beta$ or $B$. Certain subclasses of these transformations correspond to T-dualities, which leads us to formulas which are invariant under T-duality. This is the aim of double field theory.

\subsection{Double field theory}\label{ssec:DFT}

The moduli space of toroidal compactifications of closed string theory on the $D$-dimensional torus $T^D$ is the {\em Narain moduli space} \cite{Narain:1985jj}
\begin{equation}
 \mathfrak{M}=\sO(D,D,\RZ)~\backslash~\sO(D,D,\FR)~/~\sO(D,\FR)\times \sO(D,\FR)~.
\end{equation}
Here, $\sO(D,D,\FR)$ is the global symmetry group, $\sO(D,D,\RZ)$ are the T-dualities and\linebreak $\sO(D,\FR)\times \sO(D,\FR)$ are the remaining Lorentz transformations. The latter can be promoted to a local symmetry group. Beyond these three symmetry groups, we have the local diffeomorphisms together with the local gauge symmetries of the $B$-field. These are precisely the various symmetry groups appearing in the previous section.

To obtain an $\sO(D,D,\FR)$-invariant formulation, we complement the coordinates $x^\mu$ by an additional set, $x_\mu$, to achieve the total $x^M=(x^\mu,x_\mu)$. Generalized vectors here are special sections of the bundle $E_2$ over the doubled space, and we write
\begin{equation}
 X=X^\mu\left(\der{x^\mu}+\dd x_\mu\right)+X_\mu\left(\der{x_\mu}+\dd x^\mu\right)~.
\end{equation}

String theory now requires that states or fields satisfy the level matching condition 
\begin{equation}
 (L_0-\bar L_0)f(x)=0~,
\end{equation}
which, for the massless subsector $N=\bar N=1$, translates directly into
\begin{equation}
 \der{x^\mu}\der{x_\mu}f(x)=0~.
\end{equation}
This is the {\em weak} or {\em physical section condition}. We expect that there is an algebra structure on the fields satisfying the weak section condition and in particular, multiplication should close. This requires that we also impose the {\em strong section condition}, 
\begin{equation}\label{eq:strong_condition}
 \eta^{MN}\left(\der{x^M}f(x)\right)\left(\der{x^N}g(x)\right)=\left(\der{x^\mu}f(x)\right)\left(\der{x_\mu}g(x)\right)+\left(\der{x^\mu}g(x)\right)\left(\der{x_\mu}f(x)\right)=0
\end{equation}
on any pair of fields $f(x)$ and $g(x)$. A straightforward example of a solution to the strong section is simply to put $\der{x_\mu}=0$ and for consistency, we should also put its dual, $\dd x_\mu$ to zero. This leads back to generalized geometry.

Note that the meaning of the strong condition \eqref{eq:strong_condition} on tensor fields is rather opaque, as it is clearly only covariant on functions. Also, some constructions in double field theory suggest that this condition is too strong and problematic, as mentioned in the introduction. 

The metric $g$ and the Kalb-Ramond field $B$ are again encoded in the generalized metric $\CH$. Let us also consider the dilaton $d$ which is connected via the field redefinition $\de^{-2d}=\sqrt{|g|}\de^{-2\phi}$ to the usual dilaton field $\phi$. The transformation laws for $\CH$ and $d$ read as
\begin{equation}\label{eq:rev:gauge transformations}
\begin{aligned}
\delta_X \CH_{MN}&=X^P\dpar_P\CH_{MN}+(\dpar_M X^P-\dpar^PX_M)\CH_{PN}+(\dpar_N X^P-\dpar^PX_N)\CH_{MP}~,\\
\delta_X (\de^{-2d})&=\dpar_M (X^M\de^{-2d})~.
\end{aligned}
\end{equation}
This transformation is given indeed by the generalized Lie derivative \eqref{eq:action_extended_Lie} and the dilaton transforms as a scalar tensor density. 

The biggest success of double field theory is probably the provision of an action\footnote{A first action for double field theory was already suggested in \cite{Siegel:1993th}.} \cite{Hohm:2010pp}
\begin{equation}
 S_{\rm DFT}=\int \dd^{2D} x~\de^{-2d}~\CR
\end{equation}
based on the Ricci scalar
\begin{equation}
\begin{aligned}
 \CR&=\tfrac18 \CH_{MN}\dpar^M\CH_{KL}\dpar^N \CH^{KL}-\tfrac12 \CH_{MN}\dpar^M\CH_{KL}\dpar^L\CH^{KN}\\
 &~~~~~-2\dpar^Md \dpar^N \CH_{MN}+4\CH_{MN}\dpar^M d\dpar^N d~.
\end{aligned}
\end{equation}
This action is invariant under the generalized diffeomorphisms \eqref{eq:rev:gauge transformations}. Moreover, upon imposing the strong section condition $\der{x_\mu}=0$ and integrating by parts, it reduces to the usual action for the NS sector of supergravity:
\begin{equation}
 S_{\rm NS}=\int \dd^D x~\sqrt{g}~\de^{-2\phi}\left(R+4(\dpar\phi)^2-\tfrac{1}{12} H^2\right)~.
\end{equation}
A more detailed analysis shows that constructing consistent doubled versions of the torsion, Riemann and Ricci tensors is much more involved.

Let us close with a few remarks on exceptional field theory, in which the T-duality invariant target space formulation of the massless subsector of string theory is replaced by a U-duality invariant target space formulation of a subsector of M-theory. Instead of winding modes, one here has to account for the various wrapping modes of M2- and M5-branes around non-trivial cycles of the target space. 

There is an obvious analogue of generalized geometry, based on truncations of the bundle $TM\oplus \wedge^2 T^*M\oplus \wedge^5 T^*M$. Sections of this bundle are capable of describing infinitesimal diffeomorphisms as well as gauge symmetries of 3- and 6-form potentials. The $\sO(D,D,\FR)$ symmetry is replaced by the exceptional group $E_n$, where $n$ is the dimension of the torus $T^n$ in the compactification. There are evident analogues of the generalized metric as well as the generalized Lie derivative and the Courant bracket.

Moreover, a U-duality invariant extension of these objects exists, known as {\em exceptional field theory}, again with a Ricci scalar and an action principle. An extension of our formalism suitable for U-duality will be presented in upcoming work \cite{Deser:2017aa}.

\section{N\texorpdfstring{$Q$}{Q}-manifolds}\label{sec:NQ}

In the following section, we set up our conventions for symplectic N$Q$-manifolds, which will be the mathematical structure underlying most of our considerations. We also explain the relation between the symplectic N$Q$-manifolds known as exact Courant algebroids and $\sU(1)$-bundle gerbes.

\subsection{N\texorpdfstring{$Q$}{Q}-manifolds and higher Lie algebras}\label{ssec:NQ-manifolds}

Recall that an {\em N-manifold} $\CM$ is an $\NN_0$-graded manifold, i.e.\ a non-negatively graded manifold. We usually denote the body, i.e.\ the degree 0 part, of $\CM$ by $\CM_0$. We shall be interested exclusively in N-manifolds arising from $\NN$-graded vector bundles over $\CM_0$. Such N-manifolds were called {\em split} in \cite{Sheng:1103.5920}, following the nomenclature for supermanifolds. Note that Batchelor's theorem \cite{JSTOR:1998201} for supermanifolds can be extended to N-manifolds, and therefore any smooth N-manifold is diffeomorphic to a split N-manifold \cite{Bonavolonta:2012fh}, see also \cite{Roytenberg:0203110} for special cases.

An {\em N$Q$-manifold} is an N-manifold endowed with a homological vector field $Q$. That is, $Q$ is of degree 1 and satisfies $Q^2=0$. If the N$Q$-manifold $\CM$ is concentrated\footnote{i.e.\ nontrivial} in degrees $0,\ldots, n$, then we call $\CM$ a Lie $n$-algebroid. Lie $n$-algebroids $\CM$ with trivial body $\CM_0=*$ form Lie $n$-algebras or $n$-term $L_\infty$-algebras, given in terms of their Chevalley-Eilenberg description.

For example, consider an N$Q$-manifold concentrated in degree 1: $\CM=\frg[1]$, where in general $[p]$ indicates a shift of the degree of the fibers or relevant linear spaces by $p$. If we parametrize the vector space $\frg[1]$ by coordinates $\xi^\alpha$ of degree 1, the vector field $Q$ is necessarily of the form $Q=\frac12 c^\alpha_{\beta\gamma}\xi^\beta\xi^\gamma\der{\xi^\alpha}$. The condition $Q^2=0$ is then equivalent to the Jacobi identity for the structure constants $c^\alpha_{\beta\gamma}$. Similarly, one obtains higher Lie $n$-algebras. 

Since the case of Lie 2-algebras will be particularly important later on, let us present it in somewhat more detail. Here, the N$Q$-manifold is concentrated in degrees 1 and 2, and consists of two vector spaces $V[2]$ and $W[1]$ fibered over a point $*$:
\begin{equation}
 \CM= * \leftarrow W[1] \leftarrow V[2]\leftarrow * \leftarrow \ldots ~.
\end{equation}
In terms of coordinates $v^a$ and $w^\alpha$ on $\CM$ of degrees 2 and 1, respectively, the most general homological vector field reads as
\begin{equation}\label{eq:Q_for_Lie_2_algebras}
 Q = - c^\alpha_a v^a \der{w^\alpha} - \frac12 c^\alpha_{\beta\gamma} w^\beta w^\gamma \der{w^\alpha} 
- c^a_{\alpha b} w^\alpha v^b \der{v^a} + \frac{1}{3!} c^a_{\alpha\beta\gamma} w^\alpha
w^\beta w^\gamma \der{v^a}~.
\end{equation}
Recall from above that in the Chevalley-Eilenberg description, ordinary Lie algebras $\frg$ appear as N$Q$-manifolds $\frg[1]$ concentrated in degree one. Similarly, the N$Q$-manifold concentrated in degrees 1 and 2 encodes the Lie 2-algebra $\sL=W[0]\oplus V[1]$. In terms of a {\em graded} basis $\tau_\alpha$ and $t_a$ carrying degrees 0 and 1, respectively, the structure constants in \eqref{eq:Q_for_Lie_2_algebras} encode brackets\footnote{In principle, we also complete these brackets by putting all other possible $\mu_i(...)$ taking $i$ basis elements as arguments to zero.}
\begin{equation}
\begin{gathered}
 \mu_1(t_a)=c^\alpha_a \tau_\alpha~,\\
 \mu_2(\tau_\alpha,\tau_\beta)=c^\gamma_{\alpha\beta}\tau_\gamma~,~~~\mu_2(\tau_\alpha,t_b)=c^a_{\alpha b}t_a~,\\
 \mu_3(\tau_\alpha,\tau_\beta,\tau_\gamma)=c^a_{\alpha\beta\gamma}t_a~.
\end{gathered}
\end{equation}
The equation $Q^2=0$ now translates into the {\em higher} or {\em homotopy Jacobi identities}, which for a Lie 2-algebra read as 
\begin{subequations}\label{eq:homotopy_relations}
\begin{equation}
\begin{aligned}
 \mu_1(\mu_2(w,v))&=\mu_2(w,\mu_1(v))~,~~~\mu_2(\mu_1(v_1),v_2)=\mu_2(v_1,\mu_1(v_2))~,\\
 \mu_1(\mu_3(w_1,w_2,w_3))&=\mu_2(w_1,\mu_2(w_2,w_3))+\mu_2(w_2,\mu_2(w_3,w_1))+\mu_2(w_3,\mu_2(w_1,w_2))~,\\
 \mu_3(\mu_1(v),w_1,w_2)&=\mu_2(v,\mu_2(w_1,w_2))+\mu_2(w_2,\mu_2(v,w_1))+\mu_2(w_1,\mu_2(w_2,v)) 
\end{aligned}
\end{equation}
and
\begin{equation}\label{eq:homotopy_relation4b}
\begin{aligned}
 \mu_2(\mu_3(w_1,&w_2,w_3),w_4)-\mu_2(\mu_3(w_4,w_1,w_2),w_3)+\mu_2(\mu_3(w_3,w_4,w_1),w_2)\\
 & -\mu_2(\mu_3(w_2,w_3,w_4),w_1)=\\
 &\mu_3(\mu_2(w_1,w_2),w_3,w_4)-\mu_3(\mu_2(w_2,w_3),w_4,w_1)+\mu_3(\mu_2(w_3,w_4),w_1,w_2)\\
 &-\mu_3(\mu_2(w_4,w_1),w_2,w_3)
 -\mu_3(\mu_2(w_1,w_3),w_2,w_4)-\mu_3(\mu_2(w_2,w_4),w_1,w_3)~,
\end{aligned}
\end{equation}
\end{subequations}
where $v,v_i\in V[1]$ and $w,w_i\in W[0]$. Finally, let us mention that this form of a Lie 2-algebra is usually called {\em semistrict} Lie 2-algebra. This is sufficiently general for all our purposes; the most general notion of a weak Lie 2-algebra was given in \cite{Roytenberg:0712.3461}. Further details on semistrict Lie 2-algebras can be found in \cite{Baez:2003aa} and references therein.

An important example of a Lie 2-algebra which will be very similar in form to the ones we will encounter later is the String Lie 2-algebra of a compact simple Lie group $\sG$. Let $\frg:=\sLie(\sG)$ be the Lie algebra of $\sG$, then the String Lie 2-algebra has underlying graded vector space $\FR[1]\oplus \frg$ endowed with higher products
\begin{equation}
 \mu_1(r)=0~,~~~\mu_2(X_1,X_2)=[X_1,X_2]~,~~~\mu_2(X,r)=0~,~~~\mu_3(X_1,X_2,X_3)=(X_1,[X_2,X_3])~,
\end{equation}
where $X_i\in \frg$, $r\in \FR[1]$ and $(-,[-,-])$ is the generator of $H^3(\sG,\RZ)$ involving the Killing form on $\sG$.

One can now extend our above two examples of Lie algebras or N$Q$-manifolds concentrated in degree 1 and Lie 2-algebras or N$Q$-manifolds concentrated in degree 2 to {\em Lie $n$-algebras} or {\em $n$-term $L_\infty$-algebras}, which are N$Q$-manifolds concentrated in degrees $1$ to $n$. That is, an $n$-term $L_\infty$-algebra is given by a graded vector space $\sL=\oplus_{i=0}^{n-1}\sL_i$ forming a complex
\begin{equation}\label{eq:L_infity_action}
\xymatrixcolsep{4pc}
 \myxymatrix{
\ldots \ar@{->}[r]^{\mu_1} & \sL_2 \ar@{->}[r]^{\mu_1}& \sL_1 \ar@{->}[r]^{\mu_1} & \sL_0 
    }~,
\end{equation}
which is endowed with totally antisymmetric maps $\mu_i:\sL^{\wedge i}\rightarrow \sL$ of degree $i-2$. These higher products satisfy higher Jacobi identities generalizing~\eqref{eq:homotopy_relations}, which are equivalent to the condition $Q^2=0$ in the N$Q$-manifold picture. Finally, $n$-term $L_\infty$-algebras are special cases of {\em $L_\infty$-algebras} for which $\sL=\oplus_{i=0}^{\infty}\sL_i$.

Note that there are various notions of Lie $n$-algebras, and in particular multiple definitions of Lie 2-algebras, which differ in generality but which are usually equivalent in a suitable sense. For example, semistrict Lie 2-algebras have been shown to be categorically equivalent to 2-term $L_\infty$-algebras in~\cite{Baez:2003aa}. To simplify our nomenclature, we will use the terms Lie $n$-algebras and $n$-term $L_\infty$-algebras interchangeably in the following.

A {\em morphism of $L_\infty$-algebras} is now simply a morphism of N$Q$-manifolds. We shall be mostly interested in {\em strict} such morphisms $\varphi:\sL\rightarrow \sL'$, which consist of maps of graded vector spaces of degree 0 such that
\begin{equation}
 \mu'_i\circ \varphi^{\otimes i}=\varphi\circ \mu_i~.
\end{equation}

Rarely discussed in the literature but still very useful for our analysis are the notions of action and semidirect product for $L_\infty$-algebras. A very detailed account\footnote{with certain restrictions to finite dimensional $L_\infty$-algebras} is found in \cite{Mehta:2012ppa}, where the rather evident generalizations of the notions action, extension and semidirect product of Lie algebras to $L_\infty$-algebras were explored.

First, recall that an action of a Lie algebra $\frg$ on a smooth manifold $M$ is a Lie algebra homomorphism from $\frg$ to the vector fields $\frX(M)$. Next, note that given a graded manifold $\CM$, the vector fields $\frX(\CM)$ form a graded Lie algebra, which is a particular case of an $L_\infty$-algebra. Thus, an {\em action of an $L_\infty$-algebra} $\sL$ on a manifold $\CM$ is a morphism of $L_\infty$-algebras from $\sL$ to $\frX(\CM)$. We shall be exclusively interested in actions corresponding to strict morphisms of $L_\infty$-algebras, given by a chain map $\delta$ of the form
\begin{equation}\label{eq:L_infity_action}
\xymatrixcolsep{4pc}
 \myxymatrix{
\ldots \ar@{->}[r]^{\mu_1} \ar@{->}[d]^{\delta} & \sL_2 \ar@{->}[r]^{\mu_1}\ar@{->}[d]^{\delta}& \sL_1 \ar@{->}[r]^{\mu_1}\ar@{->}[d]^{\delta}& \sL_0 \ar@{->}[d]^{\delta}\\
    \ldots \ar@{->}[r]^{\id} & {*} \ar@{->}[r]^{\id} & {*}\ar@{->}[r]^{0} & \frX(\CM)\\
    }
\end{equation}

{\em Semidirect products of $L_\infty$-algebras} can be similarly defined, using analogies with Lie algebras \cite{Mehta:2012ppa}. These have an underlying action $\rho$ of an $L_\infty$-algebra $(\sL,\mu)$ on another $L_\infty$-algebra $(\sL',\mu')$ such that $\tilde \sL=\sL\ltimes \sL'$ forms an $L_\infty$-algebra with brackets 
\begin{equation}
 \tilde \mu_2(X_1+w_1,X_2+w_2)=\mu_2(X_1,X_2)+\rho(X_1)w_2-\rho(X_2)w_1+\mu'_2(w_1,w_2)~,
\end{equation}
where $X_{1,2}\in \sL$ and $w_{1,2}\in \sL'$.

\subsection{Higher symplectic Lie algebroids and their associated Lie \texorpdfstring{$n$}{n}-algebras}

An important but simple example of a Lie algebroid is the grade-shifted tangent space $T[1]M$. In terms of local coordinates $x^\mu$ and $\xi^\mu$ on the base and the fiber, we can define a vector field $Q=\xi^\mu\der{x^\mu}$. The functions on $T[1]M$ can be identified with differential forms on $M$ with $Q$ being the de Rham differential.

A {\em symplectic N$Q$-manifold of degree $n$} is an N$Q$-manifold endowed with a symplectic form $\omega$ of $\NN_0$-degree $n$, such that it is compatible with $Q$: $\CL_Q \omega=0$. We write $|\omega|=n$ for its degree. Note that the degree of $n$ implies together with the non-degeneracy of $\omega$ that the underlying N-manifold is concentrated in degrees $0,\ldots,n$. Given such a symplectic N$Q$-manifold $(\CM,Q,\omega)$, we can introduce a corresponding Poisson structure $\{-,-\}$ as follows. The Hamiltonian vector field $X_f$ of a function is defined implicitly by\footnote{Recall the Koszul rule of adding a sign when interchanging two odd elements, e.g.\ $f \dd g=(-1)^{|f|\,|g|}(\dd g)f$.}
\begin{equation}
 \iota_{X_f}\omega=\dd f
\end{equation}
and has grading $|f|-n$. The corresponding Poisson bracket reads as
\begin{equation}
 \{f,g\}:=X_f g=\iota_{X_f}\dd g=\iota_{X_f}\iota_{X_g}\omega~.
\end{equation}
As one readily verifies, this Poisson bracket is of grading $|\{f,g\}|=|f|+|g|-n$, it is graded antisymmetric,
\begin{subequations}\label{eq:properties_poisson}
\begin{equation}
 \{f,g\}=-(-1)^{(|f|+n)(|g|+n)}\{g,f\}~,
\end{equation}
and satisfies the graded Leibniz rule
\begin{equation}\label{eq:Leibniz-rule-n}
 \{f,gh\}=\{f,g\}h+(-1)^{(n-|f|)|g|}g\{f,h\}
\end{equation}
as well as the graded Jacobi identity
\begin{equation}
 \{f,\{g,h\}\}=\{\{f,g\},h\}+(-1)^{(|f|+n)(|g|+n)}\{g,\{f,h\}\}
\end{equation}
or
\begin{equation}
 \{\{f,g\},h\}=\{f,\{g,h\}\}+(-1)^{(|h|+n)(|g|+n)}\{\{f,h\},g\}
\end{equation}
\end{subequations}
for all $f,g,h\in \CC^\infty(\CM)$. 

Because of $\CL_Q \omega=0$, $Q$ is Hamiltonian, cf.\ \cite[Lemma 2.2]{Roytenberg:0203110}, and we denote its Hamiltonian function by $\CQ$:
\begin{equation}
 Q f=X_{\CQ}f=\{\CQ,f\}~.
\end{equation}
Note that
\begin{equation}
\begin{aligned}
 Q\{f,g\}&=\{\CQ,\{f,g\}\}=\{\{\CQ,f\},g\}+(-1)^{(n+1+n)(|f|+n)}\{f,\{\CQ,g\}\}\\
 &=\{Qf,g\}+(-1)^{|f|+n}\{f,Qg\}~.
\end{aligned}
\end{equation}

A trivial example of a symplectic N$Q$-manifold is an ordinary symplectic manifold $(M,\omega)$ with $Q=0$. Thus, symplectic manifolds are symplectic Lie 0-algebroids. A more interesting example is the grade-shifted cotangent bundle $T^*[1]M$. Locally, $T^*[1]M$ is described by coordinates $x^\mu$ of degree 0 and coordinates $\xi_\mu$ of degree 1. A homological vector field $Q$ is necessarily of the form $Q=\pi^{\mu\nu}\xi_\mu \der{x^\nu}$. Compatibility with the natural symplectic structure $\omega=\dd x^\mu\wedge \dd \xi_\mu$ means that $\pi^{\mu\nu}\der{x^\mu}\otimes \der{x^\nu}\in \wedge^2 \frX(\CM_0)$ and $Q^2$ implies that the bivector $\pi$ yields a Poisson structure on $M$. In this sense, Poisson manifolds are symplectic Lie 1-algebroids.

An important feature of symplectic Lie $n$-algebroids $(\CM,\{-,-\},\CQ)$ is that they come with an {\em associated Lie $n$-algebra} via a derived bracket construction, see \cite{Roytenberg:1998vn,Fiorenza:0601312,Getzler:1010.5859,Ritter:2015ffa}. This Lie $n$-algebra has underlying $\NN$-graded vector space
\begin{equation}\label{eq:ordinary_L_infty_complex}
\begin{aligned}
\sL(\CM)~:=&&~\CC^\infty_0(\CM)&\rightarrow &\CC^\infty_1(\CM)&\rightarrow &\ldots &\rightarrow &\CC^\infty_{n-2}(\CM)&\rightarrow &\CC^\infty_{n-1}(\CM)\\
=&&~\sL_{n-1}(\CM)&\rightarrow &\sL_{n-2}(\CM)&\rightarrow &\ldots &\rightarrow &\sL_1(\CM)&\rightarrow &\sL_0(\CM)~~~,
\end{aligned}
\end{equation}
where $\CC^\infty(\CM)=\CC^\infty_0(\CM)\oplus \CC^\infty_1(\CM)\oplus \CC^\infty_2(\CM)\oplus \ldots $ is the decomposition of $\CC^\infty(\CM)$ into parts $\CC^\infty_i(\CM)$ of homogeneous grading $i$ with $\CC^\infty_0(\CM)=\CC^\infty(\CM_0)$. Note that the degree of elements of $\CC^\infty_i(\CM)$ in $\sL(\CM)$ is different from their $\NN$-degree $i$. We shall therefore make an explicit distinction, calling the former $\sL$-degree and the latter $\NN$-degree. In particular elements of $\CC^\infty_i(\CM)$ have $\sL$-degree $(n-1)-i$. In general, a subscript $i$ in the expressions $\sL_i(\CM)$ and $\CC^\infty_i(\CM)$ will always refer to the $\sL$-degree and $\NN$-degree, respectively.

An explicit formula for the higher products for even $n$ can be found e.g.\ in \cite{Getzler:1010.5859}. The formulas for the lowest products $\mu_i$ for arbitrary $n$ are the following totally antisymmetrized derived brackets:
\begin{equation}\label{eq:L_infty_brackets}
\begin{aligned}
\mu_1(\ell)&=\left\{\begin{array}{ll}
0 & \mbox{if}~\ell \in \CC^\infty_{n-1}(\CM)=\sL_0(\CM)~,\\
Q\ell & \mbox{else}~,\\
\end{array}\right.\\
\mu_2(\ell_1,\ell_2)&=\tfrac12\big(\{\delta\ell_1,\ell_2\}\pm\{\delta\ell_2,\ell_1\}\big)~,\\
\mu_3(\ell_1,\ell_2,\ell_3)&=-\tfrac{1}{12}\big(\{\{\delta\ell_1,\ell_2\},\ell_3\}\pm \ldots\big)~,
\end{aligned}
\end{equation}
where
\begin{equation}\label{def:delta}
 \delta(\ell)=\left\{\begin{array}{ll}
Q\ell & \ell\in \CC^\infty_{n-1}(\CM)=\sL_0(\CM)~,\\
0 & \mbox{else}~,\\
\end{array}\right.\\
\end{equation}
and the last sum runs over all $\sL$-graded permutations. The signs in \eqref{eq:L_infty_brackets} are the obvious ones to conform with the symmetries of the $\mu_k$.

As a final remark, note that if we had not antisymmetrized the arguments in the derived bracket construction for $n=2$, we would have obtained a {\em hemistrict} Lie $2$-algebra \cite{Baez:2008bu}. The connection to {\em semistrict} Lie $2$-algebras is as follows. In a categorification of a Lie algebra to a weak Lie 2-algebra \cite{Roytenberg:0712.3461}, we lift two identities to isomorphisms, the alternator $\mathsf{Alt}:[x,y]\mapsto -[y,x]$ and the Jacobiator $\mathsf{Jac}:[x,[y,z]]\mapsto[[x,y],z]+[y,[x,z]]$. In the semistrict case, $\mathsf{Alt}$ is trivial while $\mathsf{Jac}$ is not and in the hemistrict case $\mathsf{Jac}$ is trivial while $\mathsf{Alt}$ is not. In the cases we are interested in, both hemistrict and semistrict Lie $2$-algebras turn out to be equivalent. An analogous statement should certainly hold also for higher $n$.

\subsection{Vinogradov Lie \texorpdfstring{$n$}{n}-algebroids}\label{ssec:Vinogradov}

A very important hierarchy of symplectic N$Q$-manifolds are the Vinogradov Lie $n$-alge\-broids\footnote{The essential structure underlying this algebroid was first studied in \cite{MR1074539}.} 
\begin{equation}
\CV_n(M):=T^*[n]T[1]M~,  
\end{equation}
cf.\ \cite{Gruetzmann:2014ica,Ritter:2015ffa}. Local coordinates $(x^\mu,\xi^\mu,\zeta_\mu,p_\mu)$ of $\NN$-degrees $0,1,n-1,n$ give a symplectic N$Q$-manifold with
\begin{equation}
 \omega=\dd x^\mu\wedge\dd p_\mu +\dd \xi^\mu\wedge \dd \zeta_\mu\eand \CQ=\xi^\mu p_\mu~.
\end{equation}
Note that $\CQ$ is of $\NN$-degree $n+1$, which guarantees that its Hamiltonian vector field is of degree $1$. Explicitly, the Poisson bracket reads as
\begin{equation}
 \{f,g\}:=(-1)^{n^2}f\overleftarrow{\der{p_\mu}}\overrightarrow{\der{x^\mu}} g-f\overleftarrow{\der{x^\mu}}\overrightarrow{\der{p_\mu}} g+(-1)^{(n-1)^2}f\overleftarrow{\der{\zeta_\mu}}\overrightarrow{\der{\xi^\mu}} g-(-1)^{n}f\overleftarrow{\der{\xi^\mu}}\overrightarrow{\der{\zeta_\mu}} g
\end{equation}
and the homological vector field $Q$ is given by
\begin{equation}
 Q=\xi^\mu\der{x^\mu}+p_\mu\der{\zeta_\mu}
\end{equation}
and trivially satisfies $Q^2=0$ or, equivalently, $\{\CQ,\CQ\}=0$. 

Functions $X\in \CC^\infty_{n-1}(\CV_n(M))$ of $\NN$-degree $n-1$ are of the form
\begin{equation}
 X=X^\mu\zeta_\mu+\tfrac{1}{(n-1)!}X_{\mu_1\ldots \mu_{n-1}}\xi^{\mu_1}\cdots \xi^{\mu_{n-1}}
\end{equation}
and correspond to global sections of the total space of the vector bundle 
\begin{equation}
E_n:=TM\oplus \wedge^{n-1} T^*M~. 
\end{equation}
In particular, the case $n=2$ reproduces all structures of the {\em exact Courant algebroid} \cite{Roytenberg:0203110}, which underlies generalized geometry, cf.\ section \ref{ssec:generalised_geometry}. The case $n=3$ is relevant in exceptional field theory and sections of $E_3$ can accommodate the wrapping modes of M2-branes.

For $n>1$, functions of $\NN$-degree $n-1$ will be necessarily constant or linear in $\zeta$. If we demand this also for $n=1$, then the resulting functions are sections of $TM\oplus \FR$ with derived bracket 
\begin{equation}\label{eq:Courant-bracket-n=1}
\begin{aligned}
 \mu_2(X+f,Y+g)&=\tfrac12\big(\{\{\CQ,X^\mu\zeta_\mu+f\},Y^\nu\zeta_\nu+g\}-\{\{\CQ,Y^\nu\zeta_\nu+g\},X^\mu\zeta_\mu+f\}\big)\\
 &=[X,Y]+\CL_Xg-\CL_Yf~,
\end{aligned}
\end{equation}
where $X,Y\in \frX(M)$ and $f,g\in \CC^\infty(M)$. This Lie algebra describes the local gauge transformations of a metric and a connection one-form: The infinitesimal diffeomorphisms form a Lie algebra, which, together with the abelian Lie algebra of gauge transformations, forms a semidirect product of Lie algebras.

We shall be particularly interested in the case $n=2$, capturing the local gauge transformations for a metric and the Kalb-Ramond field $B$, which is part of the connective structure of an abelian bundle gerbe. Here, we find a Lie 2-algebra $\CC^\infty(M)\xrightarrow{~\mu_1~}\frX(M)\oplus\Omega^1(M)$ with higher products
\begin{equation}\label{eq:ass_Courant_algebra}
\begin{aligned}
 \mu_1(f)&=\dd f~,\\
 \mu_2(X+\alpha,Y+\beta)&=\tfrac12\big(\{\{\CQ,X^\mu\zeta_\mu+\alpha_\mu\xi^\mu\},Y^\nu\zeta_\nu+\beta_\nu\xi^\nu\}- (X+\alpha)\leftrightarrow(Y+\beta)\big)\\
 &=[X,Y]+\CL_X\beta-\CL_Y\alpha-\tfrac12\dd\big(\iota_X\beta-\iota_Y\alpha)~,\\
 \mu_2(X+\alpha,f)&=\tfrac12\{\{\CQ,X+\alpha\},f\}=\tfrac12\iota_X\dd f~,\\
 \mu_3(X+\alpha,Y+\beta,Z+\gamma)&=\tfrac{1}{3!}\big(\{\{\{\CQ,X+\alpha\},Y+\beta\},Z+\gamma\}+\ldots\big)\\
 &=\tfrac{1}{3!}\big(\{X+\alpha,\mu_2(Y+\beta,Z+\gamma)\}+\mbox{cycl.}\big)\\
&=\tfrac{1}{3!}\big(\iota_X\iota_Y\dd \gamma+\tfrac32\iota_X\dd\iota_Y\gamma\pm\mbox{perm.}\big) 
\end{aligned}
\end{equation}
where $X,Y\in \frX(M)$ and $\alpha,\beta\in \Omega^1(M)$. The bilinear operation $\mu_2$ is also known as {\em Courant bracket}. It is the antisymmetrization of the {\em Dorfman bracket} $\nu_2$, which is recovered from the Courant bracket as
\begin{equation}\label{eq:Courant_to_Dorfman}
 \nu_2(X+\alpha,Y+\beta)=\mu_2(X+\alpha,Y+\beta)+\tfrac12Q\{X+\alpha,Y+\beta\}~.
\end{equation}
Courant and Dorfman bracket are the bilinear operations underlying semistrict and hemi\-strict Lie 2-algebras. These two Lie 2-algebras are equivalent in the sense of \cite{Roytenberg:0712.3461}, which is suggested by equation \eqref{eq:Courant_to_Dorfman}.

\subsection{Twisting Vinogradov algebroids}\label{ssec:twisted_Vinogradov}

Vinogradov Lie $n$-algebroids $\CV_n(M)$ may be twisted by a closed $n+1$-form $T $ \cite{Severa:2001qm}, which amounts to the shift
\begin{equation}\label{eq:twisted_Hamiltonian_Q}
 \CQ_T =\xi^\mu p_\mu+\tfrac{1}{(n+1)!}T _{\mu_1\ldots \mu_{n+1}}\xi^{\mu_1}\ldots \xi^{\mu_{n+1}}
\end{equation}
or
\begin{equation}
\begin{aligned}
 Q_T =\xi^\mu\der{x^\mu}+p_\mu\der{\zeta_\mu}&-\frac{1}{(n+1)!}\der{x^\nu}T _{\mu_1\ldots \mu_{n+1}}\xi^{\mu_1}\ldots \xi^{\mu_{n+1}}\der{p_\nu}\\
 &+\frac{1}{n!}T _{\nu\mu_1\ldots \mu_{n}}\xi^{\mu_1}\ldots \xi^{\mu_{n}}\der{\zeta_\nu}~.
\end{aligned}
\end{equation}
This $Q$ is automatically compatible with the symplectic structure, since $\CL_{Q_T }\omega =\dd \iota_{Q_{T }}\omega =\dd^2\CQ_T =0$. Note that $\{\CQ_T ,\CQ_T \}=Q^2_T =0$ is equivalent to $\dd T =0$. The closed form $T $ is known as the {\em \v Severa class} of the Vinogradov algebroid \cite{Severa:1998ab}, see also \cite{Severa:2001qm} and \cite{Bressler:2002ur}. We will explain its relation to the characteristic class of a gerbe in the following section.

First, however, let us analyze the most general twist element, restricting ourselves to the case $\CV_2(M)$. The most general Hamiltonian function of degree 3 reads as
\begin{equation}
\begin{aligned}
 \CQ_{S,T}=\xi^\mu p_\mu&+S_\mu{}^\nu\xi^\mu p_\nu+S^{\mu\nu}\zeta_\mu p_\nu\\
 &+\tfrac{1}{3!}T_{\mu\nu\kappa}\xi^\mu\xi^\nu\xi^\kappa+\tfrac{1}{2}T_{\mu\nu}{}^\kappa\xi^\mu\xi^\nu\zeta_\kappa+\tfrac{1}{2}T_{\mu}{}^{\nu\kappa}\xi^\mu\zeta_\nu\zeta_\kappa+\tfrac{1}{3!}T^{\mu\nu\kappa}\zeta_\mu\zeta_\nu\zeta_\kappa~,
\end{aligned}
\end{equation}
where the coefficients $S_{\ldots}^{\ldots}$ and $T_{\ldots}^{\ldots}$ are functions on $M$. We readily compute
\begin{equation}\label{eq:Qsq_twist_GG}
 \begin{aligned}
  \{&\CQ_{S,T},\CQ_{S,T}\}=\\
  &p_{\mu} p_{\nu} (2S_{\kappa }{}^{\mu} S^{\kappa \nu}+2 S^{\mu \nu})\\
&+p_{\mu} \xi^{\nu} \xi^{\kappa} \left(S^{\lambda \mu} T_{\nu \kappa \lambda}+S_{\lambda }{}^{\mu} T_{\nu \kappa }{}^{\lambda}+T_{\nu \kappa }{}^{\mu}+2 (\dpar_{\nu}+S_{\nu }{}^{\lambda} \dpar_{\lambda})S_{\kappa }{}^{\mu}\right)\\
&+p_{\mu} \zeta _{\nu} \xi^{\kappa} \left(2 S^{\lambda \mu} T_{\kappa \lambda }{}^{\nu}+2 T_{\kappa }{}^{\mu \nu}-2 S_{\lambda }{}^{\mu} T_{\kappa }{}^{\nu \lambda}-2 (\dpar_{\kappa}+S_{\kappa }{}^{\lambda} \dpar_{\lambda}) S^{\nu \mu}+2 S^{\nu \lambda} \dpar_{\lambda} S_{\kappa }{}^{\mu}\right)\\
&+p_{\mu} \zeta _{\nu} \zeta _{\kappa} \left(S^{\lambda \mu} T_{\lambda }{}^{\nu \kappa}+T^{\mu \nu \kappa}+S_{\lambda }{}^{\mu} T^{\nu \kappa \lambda}+2S^{\nu \lambda} \dpar_{\lambda} S^{\kappa \mu}\right)\\
&+\xi^{\mu} \xi^{\nu} \xi^{\kappa} \xi^{\lambda} \left(\tfrac{1}{2} T_{\mu \nu \rho} T_{\kappa \lambda }{}^{\rho}-\tfrac{1}{3} (\dpar_{\lambda}+S_{\lambda }{}^{\rho} \dpar_{\rho})T_{\mu \nu \kappa}\right)\\
&+\zeta _{\mu} \xi^{\nu} \xi^{\kappa} \xi^{\lambda} \left(T_{\lambda \rho }{}^{\mu} T_{\nu \kappa }{}^{\rho}-T_{\nu \kappa \rho} T_{\lambda }{}^{\mu \rho}-(\dpar_{\lambda}+S_{\lambda }{}^{\rho} \dpar_{\rho}) T_{\nu \kappa }{}^{\mu}+\tfrac{1}{3} S^{\mu \rho} \dpar_{\rho} T_{\nu \kappa \lambda}\right)\\
&+\zeta _{\mu} \zeta _{\nu} \xi^{\kappa} \xi^{\lambda} \left(2T_{\lambda \rho }{}^{\nu} T_{\kappa }{}^{\mu \rho}+\tfrac{1}{2} T_{\kappa \lambda }{}^{\rho} T_{\rho }{}^{\mu \nu}+\tfrac{1}{2} T_{\kappa \lambda \rho} T^{\mu \nu \rho}-(\dpar_{\lambda}+S_{\lambda }{}^{\rho} \dpar_{\rho})T_{\kappa }{}^{\mu \nu}-S^{\nu \rho} \dpar_{\rho} T_{\kappa \lambda }{}^{\mu}\right)\\
&+\zeta _{\mu} \zeta _{\nu} \zeta _{\kappa} \xi^{\lambda} \left(-T_{\lambda }{}^{\kappa \rho} T_{\rho }{}^{\mu \nu}+T_{\lambda \rho }{}^{\kappa} T^{\mu \nu \rho}-\tfrac{1}{3} (\dpar_{\lambda}+S_{\lambda }{}^{\rho} \dpar_{\rho}) T^{\mu \nu \kappa}+S^{\kappa \rho} \dpar_{\rho} T_{\lambda }{}^{\mu \nu}\right)\\
&+\zeta _{\mu} \zeta _{\nu} \zeta _{\kappa} \zeta _{\lambda} \left(\tfrac{1}{2} T_{\rho }{}^{\mu \nu} T^{\kappa \lambda \rho}+\tfrac{1}{3} S^{\mu \rho} \dpar_{\rho} T^{\nu \kappa \lambda}\right)
 \end{aligned}
\end{equation}
Note that the twist deformation $S_\mu{}^\nu=-\delta_\mu^\nu$ would remove the constant term in $\CQ_{S,T}$, and therefore we will not consider it. In the remaining cases, we can verify the following statement by a short computation, cf.\ also \cite{Roytenberg:0112152} for a similar statement.
\begin{proposition}
 The antisymmetric part of the deformation parameter $S^{\mu\nu}$ can be set to zero by a symplectomorphism, which can be written as
 \begin{equation}\label{eq:form_beta_trafo}
  z\mapsto -\big\{z,\tau\}\ewith \tau=\tfrac12 \tau^{\mu\nu}\zeta_\mu\zeta_\nu:=\tfrac12(\delta_\mu^\kappa-S_\mu{}^\kappa)^{-1} S^{\nu\kappa}\zeta_\mu\zeta_\kappa
 \end{equation}
 for $z\in (x^\mu,\xi^\mu,\zeta_\mu,p_\mu)$. Explicitly, this transformation reads as
 \begin{equation}\label{eq:coord_trafo_3.30}
  \begin{aligned}
   x^\mu\rightarrow x^\mu~,~~~\zeta_\mu\rightarrow \zeta_\mu~,~~~\xi^\mu\rightarrow \xi^\mu-\tau^{\mu\nu}\zeta_\nu\eand p_\mu\rightarrow p_\mu-\frac12 \derr{\tau^{\kappa\lambda}}{x^\mu}\zeta_\kappa\zeta_\lambda~.
  \end{aligned}
 \end{equation}
 If $S_\kappa{}^\mu\neq 0$, then $S^{\mu\nu}$ may have a symmetric part and still satisfy $\{\CQ_{S,T},\CQ_{S,T}\}=0$. Even this part can be gauged away in this way.\footnote{To write also this transformation in the form \eqref{eq:form_beta_trafo}, one should extend the Poisson brackets to the free tensor algebra of functions on $\CV_2(M)$, along the lines of section \ref{ssec:extended_tensors}.}
\end{proposition}
\noindent The coordinate transformation \eqref{eq:coord_trafo_3.30} is precisely the $\beta$-transformation of generalized geometry, cf.\ section \ref{ssec:generalised_geometry}.

Next, we can readily classify all admissible infinitesimal twist elements, see also \cite{Roytenberg:0203110} for a more abstract proof.
\begin{theorem}
 The non-trivial infinitesimal twist elements of $\CQ$ consist of closed 3-forms $\tfrac{1}{3!}T_{\mu\nu\kappa}\xi^\mu\xi^\kappa\xi^\lambda$.
\end{theorem}
\begin{proof}
 Consider the equation $\{\CQ_{S,T},\CQ_{S,T}\}=0$ for infinitesimal twist elements $S,T$. From the linearized version of \eqref{eq:Qsq_twist_GG}, it follows that
 \begin{equation}
 \begin{aligned}
 S^{\mu\nu}+S^{\nu\mu}&=0~,~~~&T_{\nu\kappa}{}^\mu&=-2\dpar_\nu S_{\kappa}{}^\mu~,~~~T_{\kappa}{}^{\mu\nu}=\dpar_\kappa S^{\nu\mu}~,\\
 \dpar_{[\mu}T_{\nu\kappa\lambda]}&=0~,~~~&T^{\mu\nu\kappa}&=0~,
 \end{aligned}
 \end{equation}
 and $T_{\nu\kappa\lambda}$ are thus the components of a closed 3-form.  Note that the deformations
 \begin{equation}
  d_1=S_\mu{}^\nu\xi^\mu p_\nu+\tfrac{1}{2}T_{\mu\nu}{}^\kappa\xi^\mu\xi^\nu\zeta_\kappa\eand
  d_2=S^{\mu\nu}\zeta_\mu p_\nu+\tfrac{1}{2}T_{\mu}{}^{\nu\kappa}\xi^\mu\zeta_\nu\zeta_\kappa
 \end{equation}
 are then $Q$-exact with
 \begin{equation}
  d_1=Q(S_\mu{}^\nu\xi^\mu\zeta_\nu)\eand d_2=Q(S^{\mu\nu}\zeta_\mu\zeta_\nu)
 \end{equation}
 and therefore trivial, since $Q$-exact terms arise from symplectomorphisms. 
\end{proof}
\noindent Thus, we see that the \v Severa class characterizes all infinitesimal deformations of the Vinogradov algebroid.

\subsection{Exact Courant algebroids and gerbes}\label{ssec:Ex_Courant_and_Gerbes}

It is well known \cite{Gawedzki:1987ak,Freed:1999vc} that the $B$-field of string theory should really be regarded as part of a Deligne 3-cocycle from a global perspective. Equivalently, it is the curving 2-form of a $\sU(1)$-bundle gerbe with connective structure over the target space. 

Recall that a principal $\sU(1)$-bundle is topologically characterized by its first Chern class, whose image in de Rham cohomology is an integer curvature 2-form. Similarly, its higher analogue, a {\em $\sU(1)$-gerbe}, is topologically characterized by its Dixmier-Douady class, whose image in de Rham cohomology is an integer curvature 3-form $H$ which locally equals $\dd B$ for some 2-form potential $B$.

In general, a gerbe is some geometric realization of an element in singular cohomology $H^3(M,\RZ)$. The most accessible such realization is probably that in terms of Murray's bundle gerbes \cite{Murray:9407015}. Let us very briefly review these in the following, a useful and detailed introduction is found in \cite{Murray:2007ps}. Assume that we have a suitable surjective map $\pi:Y\rightarrow M$ such that $Y$ covers our manifold $M$. This covering space $Y$ can be chosen to consist of local patches of an ordinary good covering, but $Y$ does not have to be locally diffeomorphic to $M$. Then we can form the fibered product
\begin{equation}
 Y^{[2]}=Y\times_M Y:=\{(y_1,y_2)\in Y\times Y~|~\pi(y_1)=\pi(y_2)\}~,
\end{equation}
which generalizes double overlaps of patches. More generally, we define
\begin{equation}    
  Y^{[n]}=Y\times_M \ldots \times_M Y:=\{(y_1,\ldots,y_n)\in Y\times \ldots\times Y~|~\pi(y_1)=\ldots=\pi(y_n)\}~,
\end{equation}
which generalize double, triple and $n$-fold overlaps of local patches. A {\em $\sU(1)$-bundle gerbe} is a principal $\sU(1)$-bundle $P$ over $Y^{[2]}$, together with an isomorphism of principal bundles $\mu_{123}:P_{12}\otimes P_{23}\rightarrow P_{13}$, called {\em bundle gerbe multiplication}, where $P_{ij}$ are the three possible pullbacks of $P$ along the three projections $Y^{[3]}\rightarrow Y^{[2]}$:
\begin{equation}
    \myxymatrix{
	{\begin{array}{c}
	 \mu_{123}:P_{12}\otimes P_{23}\stackrel{\cong}{\rightarrow}P_{13}\\
	 P_{12,13,23}	 
	\end{array}}\ar@{->}[d]& P \ar@{->}[d] & \\
	Y^{[3]}\ar@<-3pt>[r] \ar@<0pt>[r]\ar@<3pt>[r]& Y^{[2]} \ar@<1.5pt>[r] \ar@<-1.5pt>[r] & Y \ar@{->}[d]^{\pi} \\
	 & & M
}
\end{equation}
This bundle gerbe multiplication, or rather its pullback to $Y^{[4]}$ has to satisfy
\begin{equation}
 \mu_{134}\circ \mu_{123}=\mu_{124}\circ \mu_{234}~.
\end{equation}
Altogether, a $\sU(1)$-bundle gerbe is given by the triple $(P,\pi,\mu)$.

As an example, consider the trivial bundle $P=M\times \sU(1)$, which can be regarded as the trivial bundle gerbe for $Y=Y^{[2]}=Y^{[n]}=M$. Since $P\otimes P\cong P$ on $Y^{[3]}=M$, the bundle gerbe multiplication $\mu_{123}$ is simply the trivial isomorphism. This example is relevant for T-duality for trivial torus fibrations. Note that non-trivial principal $\sU(1)$-bundles $P$ are {\em not} bundle gerbes because for them, there is no isomorphism between $P\otimes P$ and $P$ over $M$.

We can now endow the principal $\sU(1)$-bundle $P$ over $Y$ in a bundle gerbe $(P,\pi,\mu)$ with a connection $\nabla$ compatible with the bundle gerbe multiplication. Its curvature $F_\nabla$ on $Y$ can be shown \cite{Murray:2007ps} to be necessarily the difference of the pullbacks of a 2-form $B$ on $Y$ along the two possible projections $Y^{[2]}\rightarrow Y$. Finally, the 3-form $\dd B$ on $Y$ is necessarily the pullback of a global 3-form $H$ on $M$. This 3-form $H$ is integral, just as the 2-form representing the first Chern class of a principal $\sU(1)$-bundle. That is, any integral over closed 3-dimensional manifolds in $M$ is an integer. We call the data $(\nabla,B)$ on a $\sU(1)$-bundle gerbe a {\em connective structure}.

Inversely, consider an integral 3-form $H\in H^3_{\rm dR}(M)$ representing the Dixmier-Douady of a bundle gerbe. For simplicity, we restrict ourselves to an ordinary cover $\frU=(U_i)$. By the Poincar\'e lemma, there are potential 2-forms $B_i$ on local patches $U_i$ with $H|_{U_i}=\dd B_i$. These are glued together by gauge transformations on overlaps $U_i\cap U_j$ parameterized by 1-forms $\Lambda_{ij}$ according to $B_i-B_j=\dd \Lambda_{ij}$, which in turn give rise to functions $f_{ijk}$ with $\dd f_{ijk}=\Lambda_{ij}-\Lambda_{ik}+\Lambda_{jk}$. The $f_{ijk}$ form the \v Cech cocycle describing the bundle gerbe multiplication $\mu_{ijk}=\de^{\di f_{ijk}}$, analogously to transition functions of a principal $\sU(1)$-bundle.

There are now two obvious local symmetries of a $\sU(1)$-bundle gerbe: the gauge symmetry of the 2-form $B$ as well as diffeomorphisms acting on the base space $M$. At infinitesimal level, these are captured by a 1-form and a vector field. Since a gerbe is a categorified principal bundle with categorified structure group, these infinitesimal symmetries do not form an ordinary Lie algebra but rather a Lie 2-algebra, and this Lie 2-algebra is precisely the Lie 2-algebra \eqref{eq:ass_Courant_algebra} constructed from the twisted Hamiltonian $\CQ_{S,T}$. Note that it is {\em not} a semidirect product of $L_\infty$-algebras, because if we write
\begin{equation}
 \mu_2(X,\beta)=\rho(X)\beta~,
\end{equation}
then $\rho(X)$ is {\em not} an action by a vector field, i.e.\ an element of $\frX(M)$, but contains second order derivatives,
\begin{equation}
 \dd \iota_X\beta=\xi^\mu \der{x^\mu} \left(X^\nu \der{\xi^\nu}\beta_\kappa \xi^\kappa\right)~.
\end{equation}

In the Lie 2-algebra \eqref{eq:ass_Courant_algebra}, elements $X+\alpha$ of $\frX(M)\oplus\Omega^1(M)$ parametrize infinitesimal diffeomorphisms and gauge transformations, while elements $f$ of $\CC^\infty(M)$ yield gauge transformations between gauge transformations:
\begin{equation}
 \myxymatrix{
\bullet
  \ar@/^4ex/[rr]^{X+\alpha}="g1" 
  \ar@/_4ex/[rr]_{X+\alpha+\dd f}="g3"
  \ar@{=>}^{f} "g1"+<0ex,-1.8ex>;"g3"+<0ex,1.8ex> 
   && \bullet
}
\end{equation}

More generally, and as suggested by \v Severa \cite{Severa:1998ab}, the Courant algebroid with \v Severa class $H\in H^3(M,\RZ)$ is to be seen as the Atiyah algebroid of the gerbe with Dixmier-Douady class $H$, see also \cite{Rogers:2010sc}. The infinitesimal symmetries\footnote{i.e.\ diffeomorphisms and gauge transformations of the connective structure} of a $\sU(1)$-gerbe with characteristic class $H$ over $M$ are described locally over patches $U_i$ by the exact Courant algebroid $\CV_2(U_i)$ with \v Severa class $H$. The 1-forms $\Lambda_{ij}$ introduced above, which describe the connection on the gerbe can now be used to glue together the local descriptions $TU_i\oplus T^*U_i$ over overlaps of patches $U_i\cap U_j$ by mapping
\begin{equation}
 X+\alpha\mapsto X+\alpha+\iota_X\dd \Lambda_{ij}~.
\end{equation}
The result is the {\em generalized tangent bundle}, cf.\ \cite{Hitchin:2005in}, and this is one of the first replacements one should make in globalizing Vinogradov algebroids in the presence of non-trivial $H$-flux. Note that in this paper, all our considerations will remain local.

These above observations for $n=2$ readily extend to higher $n$. Functions on $\CV_n(M)$ of degree $n-1$ parametrize local gauge transformations of a metric and a connection $n$-form on a $\sU(1)$-bundle $n-1$-gerbe. Such local gauge transformations form a Lie $n$-algebra, which is encoded in the $L_\infty$-algebra associated to the Vinogradov Lie $n$-algebroids $\CV_n(M)$. A more comprehensive analysis of all this is found e.g.\ in \cite{Fiorenza:1304.6292}.

In summary, we can say that the Vinogradov Lie $n$-algebroids $\CV_n(M)$ twisted by a closed $n+1$ form $T$ are linearizations of $(n-1)$-gerbes with characteristic class $T$ in the sense that they capture their infinitesimal gauge and diffeomorphism symmetry Lie $n$-algebra.

\section{Extended Riemannian geometry and pre-N\texorpdfstring{$Q$}{Q}-manifolds}\label{sec:extendedRG}

\subsection{Motivation}

To study differential geometry of Riemannian manifolds, we require a definition of tensors including a preferred symmetric tensor specifying the metric. Moreover, there is a Lie algebra of infinitesimal diffeomorphisms which is generated by vector fields and acts on tensors via the Lie derivative. This Lie algebra is simply the Lie algebra of vector fields. With it and the Lie derivative, we readily define the exterior derivative. Next, Cartan's formula links the Lie derivative to a composition of exterior derivative and interior product. The interior product is naturally extended to a more general tensor contraction, which should be defined, too. Finally, we require some notion of curvature tensors, from which we can derive an action principle for geodesics, if desired.

Most of the above mentioned structures are readily described in terms of structures on N$Q$-manifolds. Consider, e.g., the Vinogradov Lie $2$-algebroid $\CV_2(M)$ for some $D$-dimensional manifold $M$ with coordinates $(x^\mu,p_\mu,\xi^\mu,\zeta_\mu)$, $\mu=1,\ldots, D$, on the total space. Recall from section \ref{ssec:Vinogradov} that $\CV_2(M)$ is a symplectic N$Q$-manifold of $\NN$-degree $n=2$ and comes with a symplectic form $\omega$ and a homological vector field $Q$ with Hamiltonian $\CQ$.

Forms and polyvector fields are now simply functions on $\CV_2(M)$: 
\begin{equation}\label{eq:identification}
 \begin{aligned}
  \frac{1}{k!}X^{\mu_1\ldots\mu_k}(x)\der{x^{\mu_1}}\wedge\ldots \wedge\der{x^{\mu_k}} &\ \leftrightarrow\  \frac{1}{k!}X^{\mu_1\ldots\mu_k}(x)\zeta_{\mu_1}\ldots \zeta_{\mu_k}~,\\
  \frac{1}{k!}\alpha_{\mu_1\ldots\mu_k}(x)\dd{x^{\mu_1}}\wedge\ldots \wedge\dd{x^{\mu_k}} &\ \leftrightarrow\  \frac{1}{k!}\alpha_{\mu_1\ldots\mu_k}(x)\xi^{\mu_1}\ldots \xi^{\mu_k}~.
 \end{aligned}
\end{equation}
The Lie derivative is given by 
\begin{equation}
 \CL_X \alpha=\{QX,\alpha\}~,
\end{equation}
and the Lie bracket on vector fields is its antisymmetrization
\begin{equation}
 [X,Y]=\tfrac12\big(\{QX,Y\}-\{QY,X\}\big)=\{QX,Y\}~.
\end{equation}
Note that plugging polyvector fields into this bracket, one recovers the Schouten bracket. The exterior derivative is then simply 
\begin{equation}
 \dd \alpha=Q\alpha
\end{equation}
and, compatible with Cartan's formula, we have
\begin{equation}
 \iota_X \alpha=\{X,\alpha\}=(-1)^{p+1}\{\alpha,X\}~.
\end{equation}

It is therefore natural to expect that one can formulate differential and Riemannian geometry in terms of expressions on N$Q$-manifolds, which then extends to more general situations. It will turn out that we can even go slightly beyond N$Q$-manifolds, and the resulting framework will reproduce the formulas of double field theory. 

In this section, we axiomatize our definitions. In a subsequent section, we work out various examples of our framework, which we call {\em extended Riemannian geometry}.

\subsection{Extended diffeomorphisms}\label{ssec:extended_diffeos}

For our purposes, we will require a generalization of symplectic N$Q$-manifolds.
\begin{definition}
A \uline{symplectic pre-N$Q$-manifold of $\NN$-degree $n$} is a symplectic N-manifold $(\CM,\omega)$ of $\NN$-degree $n$ with a vector field $Q$ of $\NN$-degree 1 satisfying $\CL_Q\omega=0$.
\end{definition}
\noindent The fact that $Q$ generates symplectomorphisms on $\CM$ implies again that it is Hamiltonian. In the following, let $\{-,-\}$ be the Poisson bracket induced by $\omega$ and let $\CQ$ be the Hamiltonian of $Q$. 

To describe symmetries in extended Riemannian geometry, we require a homogeneously graded set of extended vector fields $\CCX(\CM)\subset \CC^\infty(\CM)$ (not to be confused with the vector fields on $\CM$, $\frX(\CM)$). These should form a semistrict Lie $n$-algebra with the higher brackets obtained as the derived brackets introduced in \eqref{eq:L_infty_brackets}. 

In particular, there is the binary bracket
\begin{equation}
\begin{aligned}
 &\mu_2: \wedge^2 \CCX(\CM)\rightarrow \CCX(\CM)~,\\
 &\mu_2(X,Y):=\tfrac12\left(\{QX,Y\}-\{QY,X\}\right)~,
\end{aligned}
\end{equation}
which requires the elements of $\CCX(\CM)$ to be of $\NN$-degree $n-1$. We immediately conclude the following simple statement:
\begin{lemma}\label{lem:PB_X_f_trivial}
 Since the Poisson bracket is of degree $-n$, we have $\{X,f\}=0$ for all $X\in\CCX(\CM)$ and $f$ of $\NN$-degree 0.
\end{lemma}

Note that $\mu_2$ is related to the ordinary derived bracket $\nu_2$.
\begin{lemma}\label{lem:Courant-Dorfman}
 The map $\mu_2$ and the map
  \begin{equation}\label{eq:def:nu2}
  \begin{aligned}
  &\nu_2: \otimes^2 \CCX(\CM)\rightarrow \CCX(\CM)~,\\
  &\nu_2(X,Y):=\{QX,Y\}~,
  \end{aligned}
  \end{equation}
 satisfy the following equations:
 \begin{equation}
  \mu_2(X,Y)=\tfrac12(\nu_2(X,Y)-\nu_2(Y,X))\eand \nu_2(X,Y)=\mu_2(X,Y)+\tfrac12 Q\{X,Y\}~.
 \end{equation}
\end{lemma}

As stated above, $\mu_2$ should be the binary bracket on the degree 0 part of a Lie $n$-algebra. One might additionally want $\nu_2$ to form a Leibniz algebra\footnote{A Leibniz algebra of degree $n$ is a graded vector space endowed with a bilinear operation satisfying the Jacobi identity of degree $n$, which differs from the Leibniz rule of degree $n$ given in \eqref{eq:Leibniz-rule-n}, cf.\ \cite{Kosmann-Schwarzbach:0312524}.}. We start with examining the latter constraint.
\begin{lemma}\label{lem:nu_2:Leibniz}
 The map $\nu_2$ as defined in \eqref{eq:def:nu2} defines a Leibniz algebra (of $\sL$-degree $0$) on $\CCX(\CM)$, i.e.\ 
 \begin{equation}
  \nu_2(X,\nu_2(Y,Z))=\nu_2(\nu_2(X,Y),Z)+\nu_2(Y,\nu_2(X,Z))~,
 \end{equation}
 if and only if 
  \begin{equation}
    \{\{Q^2X,Y\},Z\}=0
  \end{equation}
 for all $X,Y,Z\in \CCX(\CM)$. 
\end{lemma}
\begin{proof}
 We compute
 \begin{equation}
  \begin{aligned}
    &\nu_2(X,\nu_2(Y,Z))-\nu_2(\nu_2(X,Y),Z)-\nu_2(Y,\nu_2(X,Z))\\
    &=\{QX,\{QY,Z\}\}-\{Q\{QX,Y\},Z\}-\{QY,\{QX,Z\}\}\\
    &=\{\{QX,QY\},Z\}+\{QY,\{QX,Z\}\}-\{\{Q^2X,Y\},Z\}-\{\{QX,QY\},Z\}+\\
    &\hspace{10.5cm}-\{QY,\{QX,Z\}\}\\
    &=-\{\{Q^2X,Y\},Z\}~.
  \end{aligned}
 \end{equation}
\end{proof}
Note that this result is interesting in its own right. Derived brackets of this kind have been abstractly discussed in \cite{kosmann1996poisson}, where they were called {\em generalized Loday-Gerstenhaber brackets}. We are not aware of more than a few examples of such brackets, and our construction of $\nu_2$ on pre-N$Q$-manifolds using a subset $\CCX(\CM)$ of the full set of functions provides a large class of such examples.

For our discussion, however, lemma \ref{lem:nu_2:Leibniz} merely serves as a stepping stone to the following statement.
\begin{theorem}\label{thm:algebra_structure_of_brackets}
 The antisymmetrized derived bracket $\mu_2$ satisfies the homotopy Jacobi relations 
 \begin{equation}
  \mu_2(X,\mu_2(Y,Z))+\mu_2(Y,\mu_2(Z,X))+\mu_2(Z,\mu_2(X,Y ))=Q(\mu_3(X,Y,Z))~,
 \end{equation}
  \begin{equation}
  \mu_3(X,Y,Z)=3\{\{Q X,Y\},Z\}_{[X,Y,Z]}~,
 \end{equation}
 if and only if
  \begin{equation}\label{eq:Q2XYZ}
    \{\{Q^2X,Y\},Z\}_{[X,Y,Z]}=0
  \end{equation}
  holds for all $X,Y,Z\in \CCX(\CM)$. Here, and in the following we use the notation 
 \begin{equation}
 \begin{aligned}
  &F(X,Y,Z)_{[X,Y,Z]}:=\\
  &\hspace{1cm}\tfrac{1}{3!}\big(F(X,Y,Z)-F(Y,X,Z)+F(Y,Z,X)-F(Z,Y,X)+F(Z,X,Y)-F(X,Z,Y)\big)
 \end{aligned}
 \end{equation}
 for the weighted, totally antisymmetrized sum of some multilinear map $F$.
\end{theorem}
\begin{proof}
 This statement follows from the proof of lemma \ref{lem:nu_2:Leibniz} together with the observation that 
 \begin{equation}
  \begin{aligned}
   \mu_2(X,\mu_2(Y,Z))=\nu_2(X,\nu_2(Y,Z))-\{X,Q^2\{Y,Z\}\}+Q\{\{QY,Z\},X\}-Q\{X,Q\{Y,Z\}\}~,
  \end{aligned}
 \end{equation}
 which is a consequence of lemma \ref{lem:Courant-Dorfman}, and $\{Y,Z\}-\{Z,Y\}=0$.
\end{proof}

Clearly, the condition for $\mu_2$ satisfying the homotopy Jacobi relation is weaker than that of $\nu_2$ forming a Leibniz algebra. Again, we are merely interested in the Lie $n$-algebra of symmetries, so we will constrain the set $\CCX(\CM)$ by \eqref{eq:Q2XYZ}. 

Let us now extend theorem \ref{thm:algebra_structure_of_brackets} to a full Lie $n$-algebra structure on a subset $\sL(\CM)$ of $\CC^\infty(\CM)$ with derived brackets \eqref{eq:L_infty_brackets}. We start with the chain complex 
\begin{equation}
 \sL(\CM):= \sL_{n-1}(\CM) \xrightarrow{~Q~} \sL_{n-2}(\CM) \xrightarrow{~Q~} \ldots \xrightarrow{~Q~} \sL_{1}(\CM) \xrightarrow{~Q~} \sL_{0}(\CM) \xrightarrow{~0~} 0~, 
\end{equation}
cf.\ \eqref{eq:ordinary_L_infty_complex}. Note that the homotopy Jacobi relation $\mu_1^2:=Q^2=0$ is nontrivial only when applied to elements of $\sL_k(\CM)$ with $k>1$. We can therefore lift $Q^2=0$ on elements of $\sL_1(\CM)$ and $\sL_0(\CM)$, as done in theorem \ref{thm:algebra_structure_of_brackets}. Moreover, it is clear from the proof of theorem \ref{thm:algebra_structure_of_brackets} that in verifying the higher homotopy relations \eqref{eq:homotopy_relations}, $Q^2$ will never appear outside all  Poisson brackets, but always in the form $\{\ldots\{Q^2-,-\},\ldots\}$. Therefore, the condition $Q^2=0$ is unnecessarily strict and can be relaxed to conditions like \eqref{eq:Q2XYZ}. As a consequence, however, we cannot expect an $L_\infty$-algebra structure on all of $\CC^\infty(\CM)$, as in the case of Lie $n$-algebroids. Instead, we have to choose a suitable subset.

\begin{definition}
 Given a pre-N$Q$-manifold $\CM$, an \uline{$L_\infty$-algebra structure} on $\CM$ is a subset $\sL(\CM)$ of the functions $\CC^\infty(\CM)$ such that the derived brackets \eqref{eq:L_infty_brackets} close on $\sL(\CM)$ and form an $L_\infty$-algebra.
\end{definition}
\noindent In special cases, this notion of $L_\infty$-algebra structure corresponds to a polarization or, as we will see later, to the strong section condition of double field theory up to a slight weakening.

We are particularly interested in the case of Lie 2-algebras for which we have the following theorem.
\begin{theorem}\label{thm:Lie_2_subset}
 Let $\CM$ be a pre-N$Q$-manifold of degree 2. Consider a subset $\sL(\CM)$ of $\CC^\infty(\CM)$ concentrated in $\sL$-degrees $0$ and $1$, i.e.\ $\sL(\CM)=\sL_1(\CM)\oplus \sL_0(\CM)$, on which the derived brackets \eqref{eq:L_infty_brackets} and the Poisson bracket close. Then $\sL(\CM)$ is an $L_\infty$-algebra if and only if
 \begin{equation}\label{eq:conditions_thm_Lie2_subset}
  \begin{aligned}
    \{Q^2f,g\}+\{Q^2g,f\}&=0~,\\
    \{Q^2X,f\}+\{Q^2f,X\}&=0~,\\
    \{\{Q^2X,Y\},Z\}_{[X,Y,Z]}&=0
  \end{aligned}
 \end{equation}
 for all $f,g\in \sL_1(\CM)$ and $X,Y,Z\in \sL_0(\CM)$.
\end{theorem}
\begin{proof}
 The first and second conditions are equivalent to the homotopy Jacobi identities\linebreak $\mu_2(\mu_1(f),g)=\mu_2(f,\mu_1(g))$ and $\mu_1(\mu_2(X,f))=\mu_2(X,\mu_1(f))$, respectively. The third condition is equivalent to
 \begin{equation}
  \mu_1(\mu_3(X,Y,Z))=\mu_2(X,\mu_2(Y,Z))+\mu_2(Y,\mu_2(Z,X))+\mu_2(Z,\mu_2(X,Y))
 \end{equation}
 by theorem \ref{thm:algebra_structure_of_brackets}. The homotopy Jacobi identity
 \begin{equation}
  \mu_3(\mu_1(f),X,Y)=\mu_2(f,\mu_2(X,Y))+\mu_2(Y,\mu_2(f,X))+\mu_2(X,\mu_2(Y,f))
 \end{equation}
 yields the condition
 \begin{equation}
  \{\{Q^2 f,X\},Y\}-\{\{Q^2f,Y\},X\}-\{\{Q^2 X,Y\},f\}+\{\{Q^2Y,X\},f\}=0~,
 \end{equation}
 which is automatically satisfied if the brackets close on $\sL(\CM)$ and the second condition holds. Here, we used that $\{X,f\}=0$ from lemma \ref{lem:PB_X_f_trivial} implies
 \begin{equation}
  \{Q X,f\}-\{X,Qf\}=\{Qf,X\}+\{f,QX\}=0~.
 \end{equation}
 Using the same relation, we find by direct computation that also the identity \eqref{eq:homotopy_relation4b} holds up to terms of the form $\{\{Q^2X,Y\},Z\}_{[X,Y,Z]}$.
\end{proof}

\subsection{Extended tensors}\label{ssec:extended_tensors}

The $L_\infty$-structure $\sL(\CM)$ on the pre-N$Q$-manifold $\CM$ will play the role of extended infinitesimal diffeomorphisms and gauge transformations, which should act on extended tensors. Note that $\CC^\infty(\CM)$ already encodes totally antisymmetric products of vector fields, which will be identified with extended tensors. In order to capture extended Riemannian geometry, and in particular the metric, we will have to allow for arbitrary extended tensors. This is done by replacing the graded symmetric tensor algebra $\CC^\infty(\CM)$ generated by the coordinate functions of positive grading on $\CM$ by the free (associative) tensor algebra $T(\CM)$ generated by the coordinate functions of positive grading. 
\begin{definition}
 By $T(\CM)$, we mean the tensor algebra of $\CC^\infty(\CM)$ which is seen as a graded module over the commutative ring $\CC^\infty(\CM_0)$.
\end{definition}

We would like our gauge $L_\infty$-algebra $\sL(\CM)$ to act on elements in $T(\CM)$, and the natural candidate is a generalization of the derived bracket. For this, we extend the Poisson bracket via the Leibniz rule as follows.
\begin{definition}
 We implicitly define an extension of the Poisson bracket on $\CM$,\linebreak $\{-,-\}:\sL(\CM)\times T(\CM)\rightarrow T(\CM)$, by 
\begin{equation}\label{eq:ext_Poisson}
 \{f,g\otimes h\}:=\{f,g\}\otimes h+(-1)^{(n-|f|)|g|}g\otimes \{f,h\}~.
\end{equation}
\end{definition}
\noindent Note that this equation is consistent and fixes the extension uniquely, as is readily seen by computing the expression $\{f,g_1\otimes g_2\otimes g_3\}$, $g_{1,2,3}\in \CC^\infty(\CM)$, in the two possible ways. We now have the following lemma:
\begin{lemma}
 The extended Poisson bracket \eqref{eq:ext_Poisson} satisfies the Jacobi identity
 \begin{equation}
  \{f,\{g,t\}\}=\{\{f,g\},t\}+(-1)^{(|f|+n)(|g|+n)}\{g,\{f,t\}\}
 \end{equation}
 for all $f,g\in \CC^\infty(\CM)$ and $t\in T(\CM)$.
\end{lemma}
\begin{proof}
 By direct verification of the identity for $t=t_1\otimes t_2$.
\end{proof}

The extended Poisson bracket allows us now to define a natural action of the Lie $n$-algebra $\sL(\CM)$ on a subset $\sT(\CM)$ of $T(\CM)$ via 
\begin{equation}\label{eq:L_infty-rep}
f\acton t:=\{\delta f,t\}~,
\end{equation}
where $f\in \sL(\CM)$, $t\in \sT(\CM)\subset T(\CM)$ and $\delta$ was defined in section \ref{ssec:NQ-manifolds} as $\delta(t)=Qt$ for $t\in \CC^\infty_{n-1}(\CM)=\sL_0(\CM)$ and $\delta(t)=0$ else. The subset $\sT(\CM)$ is then defined by the requirement that \eqref{eq:L_infty-rep} is indeed an action of an $L_\infty$-algebra as defined in section \ref{ssec:NQ-manifolds}. 

\begin{theorem}\label{thm:restrictions_T}
 The map \eqref{eq:L_infty-rep} defines an action of $\sL(\CM)$ on $\sT(\CM)$ if $\sT(\CM)$ is a subset of $T(\CM)$ whose elements $t$ satisfy
  \begin{equation}\label{eq:tensor-conditions}
   \big\{\{Q^2X,Y\}-\{Q^2Y,X\},t\big\}=0\eand \Big\{\big\{\{Q^2X,Y\}-\{Q^2Y,X\},QZ\big\},t\Big\}=0
  \end{equation}
  for all $X,Y,Z\in \sL_0(\CM)$. That is, $\acton$ defines a morphism of $L_\infty$-algebras from $\sL(\CM)$ to $\frX(\sT(\CM))$.
\end{theorem}
\begin{proof}
 The first equation amounts to
 \begin{equation}
  X\acton(Y\acton t)-Y\acton(X\acton t)=\mu_2(X,Y)\acton t~, 
 \end{equation}
 for all $X,Y\in \sL_0(\CM)$, where 
 \begin{equation}
  \mu_2(X,Y)=\tfrac12(\{QX,Y\}-\{QY,X\})~.
 \end{equation}
 The second equation guarantees that $X\acton$ is indeed an element of $\frX(\sT(\CM))$, that is $X\acton t\in \sT(\CM)$.  Since $\delta$ vanishes on $\sL_i(\CM)$ for $i>0$, there is nothing else to check.
\end{proof}

Note that for $t\in\sL_0(\CM)$, the first condition in \eqref{eq:tensor-conditions} is an antisymmetrization of that of lemma \ref{lem:nu_2:Leibniz}, but stronger than that of theorem \ref{thm:algebra_structure_of_brackets}. Ideally, we would like elements of $\sL(\CM)$ to be simultaneously elements of $\sT(\CM)$, just as vectors are also tensors. Also, it is not clear to us whether the most general set $\sT(\CM)$ is really relevant or interesting. We therefore give the following definition.
\begin{definition}
 An \uline{extended tangent bundle structure} on a symplectic pre-N$Q$-manifold $\CM$ of degree $n$ is an $L_\infty$-algebra structure $\sL(\CM)$, which is simultaneously carrying an action of $\sL(\CM)$. An \uline{extended vector field} is an element of $\CCX(\CM)=\sL_0(\CM)$ of $\sL$-degree 0. It is a section of the \uline{extended tangent bundle} $\CCT \CM$, which is a vector bundle over $M=\CM_0$. 
 
 An \uline{extended tensor of type (p,q)} is now simply a map from $\CCX^*(\CM)^{\otimes p}\otimes\CCX(\CM)^{\otimes q}\rightarrow \CC^\infty(\CM_0)$, which is multilinear over each point of $\CM_0$. Here, $*$ indicates the dual of $\CCX(\CM)$ as a module over $\CC^\infty(\CM_0)$.
 
 An \uline{extended function} is an element of $\CC^\infty(\CM_0)=\sL_{n-1}(\CM)$ of $\sL$-degree $n-1$.
\end{definition}
\noindent Note that any $L_\infty$-algebra structure $\sL(\CM)$ gives rise to an extended tangent bundle structure, as long as $Q^2=0$.

\subsection{Extended tensor densities}\label{ssec:tensor_densitites}

For writing down action functionals, we shall also require an extended scalar tensor density and its transformations under extended symmetries. For example, the naive local top form $\dd x^1\wedge \ldots \wedge \dd x^D$ on a $D$-dimensional Lorentzian manifold is not invariant, but needs to be multiplied by a scalar tensor density, yielding
\begin{equation}\label{eq:vol_GR}
 \sqrt{|g|}\,\dd x^1\wedge \ldots \wedge \dd x^D~~~\mbox{or}~~~\sqrt{|g|}\,\de^{-2\phi}\,\dd x^1\wedge \ldots \wedge \dd x^D~,
\end{equation}
where $\phi$ is the ordinary dilaton field. We follow the convention of double field theory and combine both the tensor density and a potential exponential of a dilaton into an extended dilaton $d$ and always write
\begin{equation}
 \Omega:=\de^{-2d}\,\dd x^1\wedge \ldots \wedge \dd x^D
\end{equation}
for the volume form. Note that $d=-\tfrac12 \log \sqrt{|g|}$ in the case of ordinary geometry and trivial dilaton.

In a general extended geometry, the transformation law of $\de^{-2d}$ is derived from the fact that a scalar function $R$ together with $\Omega$ should give rise to the invariant action
\begin{equation}\label{eq:invariance_action}
 S=\int \Omega R~,
\end{equation}
such that $\delta S$ is an integral over a total derivative. That is,
\begin{equation}
 (X\acton \de^{-2d})\Omega R+\de^{-2d}(X\acton \Omega) R+\de^{-2d}\Omega (X\acton R)=\dpar_M(X^M\de^{-2d}\Omega R)~,
\end{equation}
which reproduces \eqref{eq:vol_GR} in the case of ordinary differential geometry. We shall work out several further examples in sections \ref{sec:examples} and \ref{sec:dft}.

\subsection{Extended metric and action}\label{ssec:metric_and_action}

So far, we have described extended infinitesimal symmetries as Lie $n$-algebras and we have defined their action on extended tensor fields. It remains to provide a definition of an extended metric. For this, we generalize an idea found e.g.\ in \cite{Hull:2007zu}.

Let $\sG$ and $\varrho$ be the structure group and the relevant representation on the extended tangent bundle $\CCT(\CM)$ on our pre-N$Q$-manifold $\CM$ and let $\sH$ be a maximal compact subgroup. An extended metric is simply a reduction of the bundle $\CCT(\CM)$ to another one, $E$, with structure group $\sH$ and restriction of the representation $\varrho$. 

Explicitly, let $\frU=(U_i)$ be a cover of $M$ and let $g_{ij}:U_i\cap U_j\rightarrow \varrho(\sG)$ be maps encoding a general cocycle with which defines $\CCT(\CM)$ subordinate to $\frU$. Then the reduction is given by a coboundary consisting of maps $\gamma_i:U_i\rightarrow \varrho(\sG)$ between $g_{ij}$ and another cocycle $h_{ij}$ with values in $\varrho(\sH)$, defining an isomorphic vector bundle:
\begin{equation}\label{eq:coboundary}
 h_{ij}=\gamma_i g_{ij}\gamma_j^{-1}~.
\end{equation}
This equation  implies that
\begin{equation}
 \gamma_i=h_{ij}\gamma_j g_{ij}^{-1}~,
\end{equation}
and therefore $\gamma_i$ encodes a bundle morphism from $\CCT(\CM)$ to another bundle $E$ with structure group $\sH$ and representation $\varrho|_{\sH}$.

Note, however, that also a coboundary given by $h_i\gamma_i$ with $h_i:U_i\rightarrow \varrho|_{\sH}$ defines such a bundle:
\begin{equation}\label{eq:remaining_equivalence}
 h'_{ij}=h_i\gamma_i g_{ij} (h_j\gamma_j)^{-1}~.
\end{equation}
Let us now assume that $\CCT(\CM)$ carries an $\sH$-invariant inner product and denote the corresponding adjoints by $-^*$. Then $h^*_i=h^{-1}_i$, but $\gamma_i^*\neq \gamma_i^{-1}$ in general and we have a simple way of factoring out the remaining equivalence \eqref{eq:remaining_equivalence} by considering
\begin{equation}
 \CH:=\gamma^*\gamma~.
\end{equation}
We will usually refer to $\CH$ as the extended metric.

Action principles can be defined from the extended metric and its derivatives as functions which are invariant under the symmetry group $\sH$ and the extended symmetries induced by the extended Lie derivative. Such action functionals have to be constructed individually for each extended geometry as a sum of $\sH$-invariant terms which is invariant under extended diffeomorphisms. It turns out that in all the cases we shall discuss, the terms considered in the literature on double and exceptional field theory, cf.\ e.g.\ \cite{Berman:2011pe}, are sufficient. That is, we will consider actions of the form 
\begin{equation}
\begin{aligned}
 S=\int_M &\de^{-2d}\,\dd x^1\wedge \ldots \wedge \dd x^D \Big(c_0 \CH_{MN}\dpar^M\CH_{KL}\dpar^N \CH^{KL}+c_1 \CH_{MN}\dpar^M\CH_{KL}\dpar^L\CH^{KN}\\
 &~~+c_2\CH^{MN}(\CH^{KL}\dpar_M\CH_{KL})(\CH^{RS}\dpar_N \CH_{RS})+c_3\CH^{MN}\CH^{PQ}(\CH^{RS}\dpar_P\CH_{RS})(\dpar_M\CH_{NQ})\\
 &~~+c_4 \dpar^Md \dpar^N \CH_{MN}+ c_5\CH_{MN}\dpar^M d\dpar^N d\Big)
\end{aligned}
\end{equation}
for some real constants $c_0,\ldots,c_5$, where $d$ is the extended dilaton introduced in section \ref{ssec:tensor_densitites}.

It is an obvious question whether one can derive this action from an extended Riemann tensor. We will discuss the problems associated with extended Riemann and torsion tensors in detail in section \ref{ssec:Riemann_tensor}.

\section{Examples}\label{sec:examples}

\subsection{Example: Riemannian geometry}

To describe the symmetries of ordinary Riemannian geometry on a manifold $M$, we choose $\CM=\CV_1(M)$ with coordinates $x^\mu,\zeta_\mu$ of $\NN$-degree 0 and $\xi^\mu,p_\mu$ of $\NN$-degree 1, cf.\ section \ref{ssec:Vinogradov}. We trivially have $Q^2=0$, so we are free to choose any consistent subset of $\CC^\infty(\CM)$ as our $L_\infty$-structure. The appropriate $\sL(\CM)$ here are the functions linear in $\zeta_\mu$ because these parametrize vector fields $X=X^\mu\zeta_\mu$, and their Lie $n$-algebra structure is simply the Lie algebra of vector fields:
\begin{equation}
 \mu_2(X,Y)=\tfrac12\big(\{QX,Y\}-\{QY,X\}\big)=X^\mu \dpar_\mu Y^\nu\zeta_\nu-Y^\mu\dpar_\mu X^\nu\zeta_\nu=[X,Y]~.
\end{equation}
Since $Q^2=0$, the tensor fields $\sT(\CM)$ are all elements of $T(\CM)$, and the action of $\sL(\CM)$ on $\sT(\CM)$ is just the usual transformation law of tensors under infinitesimal diffeomorphisms: $X\acton t:=\CL_X t$ or, in components,
\begin{equation}
\begin{aligned}
 X\acton t^\mu{}_\nu \zeta_\mu\otimes \xi^\nu&:=\{\delta X,t^\mu{}_\nu \zeta_\mu\otimes \xi^\nu\}\\
 &=\big(X^\mu(\dpar_\mu t^\nu{}_\kappa)-(\dpar_\mu X^\nu) t^\mu{}_\kappa+(\dpar_\kappa X^\mu) t^\nu{}_\mu\big)\zeta_\nu\otimes \xi^\kappa~.
\end{aligned}
\end{equation}
The reduction from $\sG$ to $\sH$ defining the extended metric is trivial in this case, since $\sG=\sH$, and the coboundaries $\gamma_i$ are simply chosen to be ordinary vielbeins. This yields the ordinary metric as extended metric, and it transforms as expected:
\begin{equation}
  X\acton g_{\mu\nu}\xi^\mu\odot \xi^\nu=(X^\kappa\dpar_\kappa g_{\mu\nu}+\dpar_\mu X^\kappa g_{\kappa\nu}+\dpar_\nu X^\kappa g_{\mu\kappa})\xi^\mu\odot \xi^\nu~.
\end{equation}

By its definition in section \ref{ssec:tensor_densitites}, the exponential of the extended dilaton $\de^{-2d}$ has to transform in such a way that the Lagrangian in \eqref{eq:invariance_action} transforms as a total derivative under the integral. As it is well known, this is the case if $\de^{-2d}$ transforms as an ordinary tensor tensity of weight 1, as e.g.\ $\sqrt{g}$. 

The relevant action here, which is constructed from the terms listed in section \ref{ssec:metric_and_action}, reads as
\begin{equation}
\begin{aligned}
 S=\int_M \de^{-2d}\,\dd x^1\wedge \ldots \wedge \dd x^D \Big(& \tfrac14 \CH^{MN}\dpar_M\CH^{KL}\dpar_N \CH_{KL}-\tfrac12 \CH^{MN}\dpar_M\CH^{KL}\dpar_L\CH_{KN}\\
 &~~-2 \dpar_Md \dpar_N \CH^{MN}+ 4\CH^{MN}\dpar_M d\dpar_N d\Big)~.
\end{aligned}
\end{equation}
Instead of checking the invariance of this action under the usual diffeomorphisms, we directly verify that it reproduces the Einstein-Hilbert action coupled to a dilaton. To see this, we identify $\de^{-2d}=\sqrt{|g|}\de^{-2\phi}$, which implies
\begin{equation}
 d=\phi-\tfrac12 \log \sqrt{|g|}~.
\end{equation}
We also recall that 
\begin{equation}
 \dpar_\mu\log \sqrt{|g|}=\Gamma_\mu:=\Gamma^{\kappa}{}_{\mu\kappa}=\tfrac12 g^{\kappa\lambda}\dpar_\mu g_{\kappa\lambda}
\end{equation}
and note that 
\begin{equation}
 \tfrac14 \dpar^\mu g^{\kappa\lambda}\dpar_\mu g_{\kappa\lambda}-\tfrac12 \dpar_\mu g_{\kappa\lambda}\dpar^\lambda g^{\kappa\mu}=\dpar_\nu g^{\nu\mu}\Gamma_\mu-\dpar_\kappa g^{\mu\nu}\Gamma^\kappa{}_{\mu\nu}+g^{\mu\nu}(\Gamma^\kappa{}_{\mu\nu}\Gamma_\kappa-\Gamma^\lambda{}_{\mu\kappa}\Gamma^{\kappa}{}_{\nu\lambda})+\Gamma^\mu \Gamma_\mu~,
\end{equation}
which is readily verified. Using the latter equation, we rewrite the action
\begin{equation}
\begin{aligned}
  S&=\int \dd^D x~\de^{-2d}\Big(\tfrac14\dpar^\mu g^{\kappa\lambda}\dpar_\mu g_{\kappa\lambda}-\tfrac12\dpar_\mu g_{\kappa\lambda}\dpar^\lambda g^{\kappa\mu}+2\dpar^\mu d \dpar^\nu g_{\mu\nu}+4\dpar_\mu d\dpar^\mu d\Big)
\end{aligned}
\end{equation}
as
\begin{equation}
 \begin{aligned}
  S&=\int \dd^D x~\de^{-2d}\Big(2\dpar_\kappa d g^{\kappa\mu}\dpar^\nu g_{\mu\nu}+\dpar_\mu g^{\mu\nu} \Gamma_\nu-\dpar_\kappa g^{\mu\nu}\Gamma^{\kappa}{}_{\mu\nu}\\
  &\hspace{4cm}+g^{\mu\nu}(\Gamma^\kappa{}_{\mu\nu}\Gamma_\kappa-\Gamma^{\lambda}{}_{\mu\kappa}\Gamma^{\kappa}{}_{\nu\lambda})+g^{\mu\nu}\Gamma_\mu\Gamma_\nu+4\dpar_\mu d\dpar^\mu d\Big)~.
 \end{aligned}
\end{equation}
Since $g^{\kappa\mu}\dpar^\nu g_{\mu\nu}=g^{\mu\nu}\Gamma^\kappa{}_{\mu\nu}+\Gamma^\kappa$, this equals
\begin{equation}
 \begin{aligned}
  S&=\int \dd^D x~\de^{-2d}\Big(2\dpar_\kappa d (g^{\mu\nu}\Gamma^\kappa{}_{\mu\nu}-\Gamma^\kappa)+4\dpar^\mu d\Gamma_\mu+\dpar_\mu g^{\mu\nu} \Gamma_\nu-\dpar_\kappa g^{\mu\nu}\Gamma^{\kappa}{}_{\mu\nu}\\
  &\hspace{4cm}+g^{\mu\nu}(\Gamma^\kappa{}_{\mu\nu}\Gamma_\kappa-\Gamma^{\lambda}{}_{\mu\kappa}\Gamma^{\kappa}{}_{\nu\lambda})+g^{\mu\nu}\Gamma_\mu\Gamma_\nu+4\dpar_\mu d\dpar^\mu d\Big)~,
 \end{aligned}
\end{equation}
and partial integration of the first term finally yields
\begin{equation}
 \begin{aligned}
  S&=\int \dd^D x\sqrt{|g|}\de^{-2\phi}\Big(g^{\mu\nu}(\dpar_\kappa\Gamma^\kappa{}_{\mu\nu}-\dpar_\nu\Gamma^\kappa{}_{\mu\kappa}+\Gamma^\kappa{}_{\mu\nu}\Gamma_\kappa-\Gamma^{\lambda}{}_{\mu\kappa}\Gamma^{\kappa}{}_{\nu\lambda})+\\
  &\hspace{4cm}+g^{\mu\nu}\Gamma_\mu\Gamma_\nu+4\Gamma_\mu\dpar^\mu d+4\dpar_\mu d\dpar^\mu d\Big)\\
  &=\int \dd^D x\sqrt{|g|}\de^{-2\phi}(R+4\dpar_\mu \phi\dpar^\mu\phi)~.
 \end{aligned}
\end{equation}

\subsection{Example: Riemannian geometry with principal \texorpdfstring{$\sU(1)$}{U(1)}-bundle}

To describe an additional connection on a (trivial) principal $\sU(1)$-bundle over $M$, we can again start from $\CM=\CV_1(M)$, but with a more general $L_\infty$-structure. We choose $\sL(\CM)$ to be the constant and linear functions in the $\zeta^\mu$, which yields $\au(1)$-valued functions together with the vector fields on $M$. Correspondingly, the extended tangent bundle is $TM\oplus \FR$. The resulting Lie $n$-algebra is an ordinary Lie algebra, namely the semidirect product of the diffeomorphisms with the gauge transformations:
\begin{equation}
\begin{aligned}
 \mu_2(f+X,g+Y)&=X^\mu \dpar_\mu (g+Y^\nu\zeta_\nu)-Y^\mu\dpar_\mu (f+X^\nu\zeta_\nu)\\
 &=[X,Y]+\CL_Xg-\CL_Yf
\end{aligned}
\end{equation}
for $f,g,X,Y\in \sL(\CM)\subset\CC^\infty(\CM_0)$ with $f,g$ and $X,Y$ constant and linear in the $\zeta^\mu$, respectively, cf.\ \eqref{eq:Courant-bracket-n=1}.

Let us introduce indices $m=(\mu,\circ)$ on extended tangent vectors, where $\circ$ stands for the component in the new, additional direction. The action of extended symmetries on extended tensors is readily obtained from its definition
\begin{equation}
 (X+f)\acton t=\{\delta (X^\mu\zeta_\mu+f),t\}
\end{equation}
for $X\in \CCX(\CV_1(M))$ and $t\in \sT(\CV_1(M))$.

The extended metric then reads as
\begin{equation}
 \CH_{mn}=\left(\begin{array}{cc}
                 g_{\mu\nu} +A_\mu A_\nu& A_\nu \\ A_\mu & 1
                \end{array}\right)
\end{equation}
with inverse
\begin{equation}
 \CH^{mn}=\left(\begin{array}{cc}
                 g^{\mu\nu} & -g^{\mu\nu}A_\nu \\ -A_\mu g^{\mu\nu} & 1+g^{\mu\nu}A_\mu A_\nu
                \end{array}\right)~.
\end{equation}
These formulas already appeared in \cite{Maharana:1992my} in a different context. Note that $\CH_{mn}$ has trivial kernel, because
\begin{equation}
 \CH=\left(\begin{array}{cc}
                 \unit_{TM} & A\\0 & 1
                \end{array}\right)\left(\begin{array}{cc}
                 g & 0\\0 & 1
                \end{array}\right)\left(\begin{array}{cc}
                 \unit_{TM} & A\\0 & 1
                \end{array}\right)^T
\end{equation}
and each of these matrices has trivial kernel. We could say that $\CH$ arises from $A$-field transformations from the block-diagonal metric, analogous to equation \eqref{eq:B-field_trafo}.

The action of an extended vector $(f+X)\acton \CH$ induces the usual diffeomorphisms on the metric $g_{\mu\nu}$ as well as gauge transformations and diffeomorphisms on the gauge potential $A_\mu$.

The extended Lie derivative transforms the top form $\dd x^1\wedge \ldots \wedge \dd x^D$ as in the previous example, and so the extended dilaton transforms again identical to an ordinary tensor density of weight 1.

The invariant action here is the same as in the previous example:
\begin{equation}
\begin{aligned}
 S=\int_M \de^{-2d}\,\dd x^1\wedge \ldots \wedge \dd x^D \Big(& \tfrac14 \CH^{MN}\dpar_M\CH^{KL}\dpar_N \CH_{KL}-\tfrac12 \CH^{MN}\dpar_M\CH^{KL}\dpar_L\CH_{KN}\\
 &~~-2 \dpar_Md \dpar_N \CH^{MN}+ 4\CH^{MN}\dpar_M d\dpar_N d\Big)~.
\end{aligned}
\end{equation}
Identifying again $\de^{-2d}=\sqrt{|g|}\de^{-2\phi}$, we verify this action's invariance by rewriting it in familiar form:
\begin{equation}
 \begin{aligned}
  S=\int_M \de^{-2d}\,\dd^D x~\Big(&\tfrac14 \CH^{\mu \nu}\dpar_\mu\CH^{\kappa\lambda}\dpar_\nu \CH_{\kappa\lambda}+\tfrac12 \CH^{\mu \nu}\dpar_\mu\CH^{\kappa\circ}\dpar_\nu \CH_{\kappa\circ}+\tfrac14 \CH^{\mu \nu}\dpar_\mu\CH^{\circ\circ}\dpar_\nu \CH_{\circ\circ}\\
 &~~-\tfrac12 \CH^{\mu \nu}\dpar_\mu\CH^{\kappa\lambda}\dpar_\lambda\CH_{\kappa\nu}-\tfrac12 \CH^{\mu \circ}\dpar_\mu\CH^{\kappa\lambda}\dpar_\lambda\CH_{\kappa\circ}-\tfrac12 \CH^{\mu \nu}\dpar_\mu\CH^{\circ\lambda}\dpar_\lambda\CH_{\circ\nu}\\
 &~~-\tfrac12 \CH^{\mu \circ}\dpar_\mu\CH^{\circ\lambda}\dpar_\lambda\CH_{\circ\circ}-2\dpar_\mu d \dpar_\nu \CH^{\mu\nu}+4\CH^{\mu\nu}\dpar_\mu d\dpar_\nu d\Big)\\
 =\int_M \de^{-2d}\,\dd^D x~\Big(&\tfrac14 g^{\mu \nu}\dpar_\mu g^{\kappa\lambda}\dpar_\nu (g_{\kappa\lambda}+A_\kappa A_\lambda)-\tfrac12 g^{\mu \nu}\dpar_\mu A^{\kappa}\dpar_\nu A_\kappa\\
 &~~-\tfrac12 g^{\mu \nu}\dpar_\mu g^{\kappa\lambda}\dpar_\lambda (g_{\kappa\nu}+A_\kappa A_\nu)+\tfrac12 A^\mu\dpar_\mu g^{\kappa\lambda}\dpar_\lambda A_\kappa+\tfrac12 g^{\mu \nu}\dpar_\mu A^\lambda\dpar_\lambda A_\nu\\
 &~~-2\dpar_\mu d \dpar_\nu g^{\mu\nu}+4g^{\mu\nu}\dpar_\mu d\dpar_\nu d\Big)~.
\end{aligned}
\end{equation}
Note that for $A_\mu=0$, this reduces to the action of the previous section. Thus, it remains to study the part $S_A$ of the action containing the gauge potential, which read as
\begin{equation}
 \begin{aligned}
  S_A=\int_M \de^{-2d}\,\dd^D x~\Big(&\tfrac14 g^{\mu \nu}\dpar_\mu g^{\kappa\lambda}\dpar_\nu (A_\kappa A_\lambda)-\tfrac12 g^{\mu \nu}\dpar_\mu A^{\kappa}\dpar_\nu A_\kappa\\
 &~~-\tfrac12 g^{\mu \nu}\dpar_\mu g^{\kappa\lambda}\dpar_\lambda (A_\kappa A_\nu)+\tfrac12 A^\mu\dpar_\mu g^{\kappa\lambda}\dpar_\lambda A_\kappa+\tfrac12 g^{\mu \nu}\dpar_\mu A^\lambda\dpar_\lambda A_\nu\Big)~.
\end{aligned}
\end{equation}
Since
\begin{equation}
\begin{aligned}
 -\tfrac14 F^2&=-\tfrac14 g^{\mu\kappa}g^{\nu\lambda}F_{\mu\nu}F_{\kappa\lambda}\\
 &=-\tfrac12 \dpar_\mu A^{\kappa}\dpar^\nu A_\kappa+\tfrac12 A_\lambda \dpar_\mu g^{\kappa\lambda}\dpar^\mu A_\kappa+\tfrac12 \dpar^\mu A^\lambda\dpar_\lambda A_\mu-\tfrac12 A_\kappa\dpar^\mu g^{\lambda\kappa}\dpar_\lambda A_\mu~,
\end{aligned}
\end{equation}
we obtain the desired expression:
\begin{equation}
 \begin{aligned}
  S_A=\int_M \de^{-2d}\,\dd^D x~\Big(&-\tfrac14 F^2\Big)~.
\end{aligned}
\end{equation}

\subsection{Example: Generalized geometry}

It is well known that the appropriate geometric structure underlying generalized geometry on a manifold $M$ is the symplectic Lie 2-algebroid $\CM:=\CV_2(M)$ of $\NN$-degree 2. Recall that this is the Vinogradov algebroid that underlies the exact Courant algebroid $TM\oplus T^*M$.

The antisymmetrized derived bracket is simply the Courant bracket, see equations \eqref{eq:ass_Courant_algebra}. The resulting Lie 2-algebra is the semidirect product of the Lie algebra of diffeomorphism, regarded trivially as a Lie 2-algebra, and the abelian Lie 2-algebra of gauge transformations of the curving 2-form of a trivial abelian gerbe. The transformation law for tensors reads as 
\begin{equation}
\begin{aligned}
 (\alpha+X)\acton t^\mu{}_\nu~\zeta_\mu\otimes \xi^\nu:=~&\{\delta \alpha+\delta X,t^\mu{}_\nu~\zeta_\mu\otimes \xi^\nu\}\\
 =~&\big(X^\mu(\dpar_\mu t^\nu{}_\kappa)-(\dpar_\mu X^\nu) t^\mu{}_\kappa+(\dpar_\kappa X^\mu) t^\nu{}_\mu\big)~\zeta_\nu\otimes \xi^\kappa+\\
 &~+(\dpar_\mu \alpha_\nu -\dpar_\nu \alpha_\mu)t^\nu{}_\lambda~\xi^\mu\otimes \xi^\lambda~.
\end{aligned}
\end{equation}
The extended metric here and how it arises from a diagonal metric via $B$-field transformation was already discussed in section \ref{ssec:generalised_geometry}. We have
\begin{equation}\label{eq:2metric_H_O(D,D)}
 \CH_{MN}=\left(\begin{array}{cc} g_{\mu\nu}-B_{\mu\kappa}g^{\kappa\lambda}B_{\lambda\nu}  &  B_{\mu\kappa}g^{\kappa\nu}\\ -g^{\mu\kappa}B_{\kappa\nu}  &g^{\mu\nu}  \end{array}\right)
\end{equation}
with inverse
\begin{equation}
 \CH^{MN}=\left(\begin{array}{cc} g^{\mu\nu} &  -g^{\mu\kappa}B_{\kappa\nu}  \\ B_{\mu\kappa}g^{\kappa\nu}& g_{\mu\nu}-B_{\mu\kappa}g^{\kappa\lambda}B_{\lambda\nu}  \end{array}\right)~.
\end{equation}
The action of an extended vector $X\acton \CH$ yields the action of diffeomorphisms on the metric as well as the combined action of diffeomorphisms and gauge transformations on the 2-form potential $B$. As in the two cases before, the extended Lie derivative transforms the top form $\dd x^1\wedge \ldots \wedge \dd x^D$ as usual, and therefore the extended dilaton transforms as an ordinary tensor density of weight 1.

The invariant action for generalized geometry is well known to be that of double field theory:
\begin{equation}
\begin{aligned}
 S=\int_M \de^{-2d}\,\dd x^1\wedge \ldots \wedge \dd x^D \Big(& \tfrac18 \CH^{MN}\dpar_M\CH^{KL}\dpar_N \CH_{KL}-\tfrac12 \CH^{MN}\dpar_M\CH^{KL}\dpar_L\CH_{KN}\\
 &~~-2 \dpar_Md \dpar_N \CH^{MN}+ 4\CH^{MN}\dpar_M d\dpar_N d\Big)~.
\end{aligned}
\end{equation}
Again, we bring this action into familiar form, focusing on the Lagrangian:
\begin{equation}
\begin{aligned}
 &\tfrac18 \dpar_\mu \CH^{KL}\dpar^\mu \CH_{KL}-\tfrac12 \CH^{\mu N}\dpar_\mu\CH^{K\lambda}\dpar_\lambda\CH_{KN}-2 \dpar_\mu d \dpar_\nu g^{\mu\nu}+ 4\dpar_\mu d\dpar^\mu d\\
 &=\tfrac18 \dpar_\mu g^{\kappa\lambda}\dpar^\mu (g_{\kappa\lambda}-B_{\kappa\rho}g^{\rho\sigma}B_{\sigma\lambda})-\tfrac18 \dpar_\mu (g^{\kappa\rho}B_{\rho \lambda})\dpar^\mu (B_{\kappa\sigma}g^{\sigma\lambda})-\tfrac18 \dpar_\mu (B_{\kappa\rho}g^{\rho\lambda})\dpar^\mu (g^{\kappa\sigma}B_{\sigma\lambda})\\
 &~~+\tfrac18 \dpar_\mu (g_{\kappa\lambda}-B_{\kappa\rho}g^{\rho\sigma}B_{\sigma\lambda}) \dpar^\mu g^{\kappa\lambda}-\tfrac12 g^{\mu \nu}\dpar_\mu g^{\kappa\lambda}\dpar_\lambda(g_{\kappa\nu}-B_{\kappa\sigma}g^{\sigma\rho}B_{\sigma\nu})\\
 &~~+\tfrac12 g^{\mu \nu}\dpar_\mu(B_{\kappa\sigma}g^{\sigma\lambda})\dpar_\lambda(g^{\kappa\rho}B_{\rho\nu})-\tfrac12 B_{\mu\nu}\dpar^\mu g^{\kappa\lambda}\dpar_\lambda(B_{\kappa\rho}g^{\rho\nu})+\tfrac12 B_{\mu\nu}\dpar^\mu(B_{\kappa\sigma}g^{\sigma\lambda})\dpar_\lambda g^{\kappa\nu}~.
\end{aligned}
\end{equation}
The terms independent of $B$ combine to the form of the Einstein-Hilbert Lagrangian plus the dilaton contributions which we encountered before,
\begin{equation}
 \tfrac14 \dpar^\mu g^{\kappa\lambda}\dpar_\mu g_{\kappa\lambda}-\tfrac12 \dpar_\mu g_{\kappa\lambda}\dpar^\lambda g^{\kappa\mu}~.
\end{equation}
The terms dependent on $B$ reduce, as expected, to $H^2$ for $H=\dd B$ and we have
\begin{equation}
 S=\int_M \dd^D x \sqrt{|g|}\de^{-2\phi}(R+4\dpar_\mu\phi\dpar^\mu \phi-\tfrac{1}{12}H^2)~,
\end{equation}
see e.g.\ \cite{Hohm:2010jy} for the details.

\subsection{Example: Higher generalized geometry}\label{ssec:3-form-GG}

So far, we have discussed general relativity with 1- and 2-form gauge potentials. This sequence is readily continued to 3-form gauge potentials. The latter case is also physically interesting, as it corresponds to exceptional field theory after imposing a section condition. Therefore, and for completeness sake, let us also list the ingredients here.

From our above discussion, we are clearly led to consider the symplectic Lie 3-algebroid $\CV_3(M)=T^*[3]T[1]M$ of $\NN$-degree 3. We use the usual coordinates $(x^\mu,\xi^\mu,\zeta_\mu,p_\mu)$ of degrees 0,1,2 and 3, respectively. Extended vector fields have degree 2 and are of the form
\begin{equation}
X+\alpha=X^\mu\zeta_\mu+\tfrac12 \alpha_{\mu\nu}\xi^\mu\xi^\nu~. 
\end{equation}
They thus contain a vector field $X^\mu\der{x^\mu}$ as well as a 2-form $\tfrac12 \alpha_{\mu\nu}\dd x^\mu\wedge \dd x^\nu$. This is the right data for diffeomorphisms and gauge transformations of a 3-form potential $C$. Since $Q^2=0$, we get a canonical $L_\infty$-structure on $\CV_3(M)$ with 2-bracket
\begin{equation}
 \mu_2(X+\alpha,Y+\beta)=[X,Y]+\CL_X\beta-\CL_Y\alpha-\tfrac12\dd\big(\iota_X\beta-\iota_Y\alpha)~,
\end{equation}
which has the same form as the Courant bracket \eqref{eq:ass_Courant_algebra}. Note, however, that $\alpha$ and $\beta$ are 2-forms here. The tensor transformation law is again obtained from 
\begin{equation}
 (X+\alpha)\acton t:=\{\delta X+\delta \alpha,t\}~,
\end{equation}
where $t$ carries indices $M=({}^\mu,{}_{\nu\kappa})$. One can, however, extend this action to a subset of more general functions on $\CV_3(M)$.

As for all previous cases, the extended metric is obtained by a gauge field transformation acting on the diagonal metric. The relevant indices here are $M=(\mu,[\nu\kappa])$ such that $\CCT \CM$ is of dimension $D+\tfrac12 D(D-1)$. Explicitly, we have
\begin{equation}
 \CH=\left(\begin{array}{cc}
                 \unit_{TM} & a C\\0 & \unit_{\wedge^2 T^*M}
                \end{array}\right)\left(\begin{array}{cc}
                 g & 0\\0 &  g^{-1}g^{-1}
                \end{array}\right)\left(\begin{array}{cc}
                 \unit_{TM} & aC\\0 & \unit_{\wedge^2 T^*M}
                \end{array}\right)^T~,
\end{equation}
where we inserted a constant $a\in\FR$ allowing for a field rescaling of $C$, and therefore we have the component expressions
\begin{equation}
 \CH_{MN}=\CH_{(\mu,[\rho\sigma]),(\nu,[\kappa\lambda])}=\left(\begin{array}{cc}
                 g_{\mu\nu}+a^2C_{\mu\alpha\beta}g^{\alpha\gamma}g^{\beta\delta}C_{\gamma\delta\nu} &  a C_{\mu\alpha\beta}g^{\alpha\kappa}g^{\beta\lambda}\\
                 a g^{\rho\alpha}g^{\sigma\beta}C_{\alpha\beta \nu} &  g^{\rho\kappa} g^{\sigma\lambda}
                \end{array}\right)
\end{equation}
with inverse
\begin{equation}
 \CH^{MN}=\CH^{(\mu,[\rho\sigma]),(\nu,[\kappa\lambda])}=\left(\begin{array}{cc}
                 g^{\mu\nu} & -a g^{\mu\alpha}C_{\alpha\kappa\lambda} \\
                 -a C_{\rho\sigma\alpha}g^{\alpha\nu} & g_{\rho \kappa}g_{\sigma\lambda}+a^2C_{\rho\sigma\alpha}g^{\alpha\beta}C_{\beta\kappa\lambda}
                \end{array}\right)~.
\end{equation}

The relevant action here requires additional terms to those of the previous examples:
\begin{equation}
\begin{aligned}
 S=\int_M&\de^{-2d}\,\dd x^1\wedge \ldots \wedge \dd x^D \Big( \frac{1}{4(1+2D)}\CH^{MN}\dpar_M\CH^{KL}\dpar_N \CH_{KL}-\tfrac12\CH^{MN}\dpar_M\CH^{KL}\dpar_L\CH_{KN}\\
 &\hspace{3cm}+\frac{1}{2(2D-1)^2(1+2D)}\CH^{MN}(\CH^{KL}\dpar_M\CH_{KL})(\CH^{RS}\dpar_N \CH_{RS})\\
 &\hspace{3cm}-2 \dpar_Md \dpar_N \CH^{MN}+ 4\CH^{MN}\dpar_M d\dpar_N d\Big)~.
\end{aligned}
\end{equation}
Again, we verify its gauge invariance by reducing it to a familiar form. Since the computations get rather involved, we used a computer algebra program to complete this task. The result is
\begin{equation}
   S=\int \dd^D x~\de^{-2d}\Big(\tfrac14\dpar^\mu g^{\kappa\lambda}\dpar_\mu g_{\kappa\lambda}-\tfrac12\dpar_\mu g_{\kappa\lambda}\dpar^\lambda g^{\kappa\mu}+2\dpar^\mu d \dpar^\nu g_{\mu\nu}+4\dpar_\mu d\dpar^\mu d+\tfrac{1}{48}G^2\Big)
\end{equation}
with the choice $a=3\sqrt{2}$, which is the usual Einstein-Hilbert action of the first example coupled to the abelian 3-form potential $C$ with curvature $G=\dd C$.

\section{Double field theory as extended Riemannian geometry}\label{sec:dft}

So far, we only employed ordinary symplectic Lie $n$-algebroids in our examples, and there was no need to lift $Q^2=0$, generalizing to pre-N$Q$-manifolds which are not N$Q$-manifolds. As we shall see now, the description of double field theory requires this lift. 

\subsection{Restriction of doubled generalized geometry}\label{ssec:E2M_from_restriction}

For a description of double field theory, we have to double spacetime from $M$ to $M\times \hat M$. Since we currently have to restrict ourselves to a local description, the bundle $T^* M$ will have to do. On this bundle, we want to describe the gauge transformations of a trivial abelian gerbe, just as in generalized geometry. This suggests to use $\CV_2(T^*M)=T^*[2]T[1](T^*M)$ with coordinates $(x^M,p_M,\xi^M,\zeta_M)=(x^\mu,x_\mu,\ldots,\zeta_\mu,\zeta^\mu)$ as our starting point. This Lie 2-algebroid comes with the usual symplectic structure and Hamiltonian for a homological vector field $Q$,
\begin{equation}
 \omega=\dd x^M\wedge \dd p_M+\dd \xi^M\wedge \dd \zeta_M~,~~~\CQ=\sqrt{2}\xi^M p_M~,
\end{equation}
where we rescaled $\CQ$ for convenience.

The underlying symmetry group is now $\sGL(2D)$, which we have to restrict to the symmetry group $\sO(D,D)$ of double field theory. This is done by introducing the metric $\eta_{MN}$ of split signature $(1,\ldots,1,-1,\ldots,-1)$, cf.\ \eqref{eq:dft_eta}. Using this metric, we can restrict the tangent space coordinates to the diagonal, using coordinates
\begin{equation}
 \theta^M=\frac{1}{\sqrt{2}}(\xi^M+\eta^{MN}\zeta_N)\eand \beta^M=\frac{1}{\sqrt{2}}(\xi^M-\eta^{MN}\zeta_N)~.
\end{equation}
Eliminating the dependence on $\beta^M$ of all our objects\footnote{Interestingly, the equation $\beta^M=0$ can be regarded as part of the equations of motion of a topological string as discussed in \cite[Section 7]{Bouwknegt:2011vn}.} leaves us with the pre-N$Q$-manifold $\CE_2(M):=(T^*[2]\oplus T[1])T^*M$ with coordinates $(x^M,\theta^M,p_M)$ of degrees $0$, $1$ and $2$, respectively. The reduction also leads to the following symplectic structure on $\CE_2(M)$:
\begin{equation}
 \omega=\dd x^M\wedge \dd p_M +\tfrac12\eta_{MN}\dd \theta^M\wedge \dd \theta^N~,~~~\CQ=\theta^M p_M~.
\end{equation}
The Poisson bracket reads as
\begin{equation}
 \{f,g\}:=f\overleftarrow{\der{p_M}}\overrightarrow{\der{x^M}} g-f\overleftarrow{\der{x^M}}\overrightarrow{\der{p_M}} g-f\overleftarrow{\der{\theta^M}}\eta^{MN}\overrightarrow{\der{\theta^N}} g
\end{equation}\label{eq:E2_Poisson}
for $f,g\in\CC^\infty(\CE_2(M))$ and the Hamiltonian vector field of $\CQ$ is
\begin{equation}
 Q=\theta^M\der{x^M}+p_M\eta^{MN}\der{\theta^N}~.
\end{equation}
Note that $Q^2=p_M\eta^{MN}\der{x^N}\neq 0$, and $\CE_2(M)$ is not an N$Q$-manifold. It is, however, a symplectic pre-N$Q$-manifold of $\NN$-degree $2$.

We thus need to choose an $L_\infty$-structure in the form of a subset $\sL(\CE_2(M))\subset \CC^\infty(M)$ satisfying the conditions of theorem \ref{thm:Lie_2_subset}. A short computation yields the following.
\begin{proposition}
 For elements $f,g$ and $X=X_M\theta^M,Y=Y_M\theta^M,Z=Z_M\theta^M$ of $\sL(\CE_2(M))$ of $\sL$-degrees 1 and 0, respectively, we have the following relations.
 \begin{equation}\label{eq:restrictions_LM}
  \begin{gathered}
    \{Q^2 f,g\}+\{Q^2 g,f\}=2\left(\der{x^M}f\right)\eta^{MN}\left(\der{x^N}g\right)=0~,\\
    \{Q^2 X,f\}+\{Q^2 f,X\}=2\left(\der{x^M}X\right)\eta^{MN}\left(\der{x^N}f\right)=0~,\\
    \{\{Q^2X,Y\},Z\}_{[X,Y,Z]}=2\theta^L\big((\dpar^MX_L)(\dpar_MY^K)Z_K\big)_{[X,Y,Z]}=0~.
  \end{gathered}
 \end{equation}
\end{proposition}

Note that the strong section condition of double field theory is sufficient, but not necessary to satisfy \eqref{eq:restrictions_LM}. Thus, choosing a specific $L_\infty$-algebra structure $\sL(\CE_2(M))$ amounts to a choice of solution to \eqref{eq:restrictions_LM} and therefore corresponds essentially to ``solving the strong section condition'' in DFT parlance. Contrary to the strong section condition postulated in DFT, however, the left-hand expressions in equations \eqref{eq:restrictions_LM} are completely independent of a choice of coordinates.

Note also that the Poisson bracket \eqref{eq:E2_Poisson} yields the natural pairing of two extended vector fields $X,Y$:
\begin{equation}
 \{X,Y\}=X_M \eta^{MN} Y_N~.
\end{equation}
We use again the shorthand notations $\der{x^M}=\dpar_M$,  $X_M:=\eta_{MN} X^N$ and $x_N=\eta_{MN} x^N$ introduced in section \ref{sec:review}.

As we shall see in section~\ref{ssec:examples_L_infty}, an obvious $L_\infty$-structure on $\CE_2(M)$ for $M=\FR^n$ recovers the corresponding Vinogradov Lie 2-algebroid encoding generalized geometry. This implies that the C-bracket and all induced structures of double field theory reduce to the Courant bracket and the corresponding structures in generalized geometry.

\subsection{Symmetries}\label{ssec:Symmetries_of_DFT}

The extended vector fields $\CCX(\CE_2(M))=\sL_0(\CE_2(M))$ encode infinitesimal diffeomorphism and gauge transformations, while elements of $\sL_1(\CE_2(M))\subset \CC^\infty(M)$ describe morphisms between these, as discussed in section \ref{ssec:Vinogradov}. Let us now give the explicit Lie 2-algebra structure.
\begin{proposition}
 An $L_\infty$-algebra structure on the pre-N$Q$-manifold $\CE_2(M)$ is a Lie 2-algebra of symmetries given by a graded vector space 
 \begin{equation}
  \sL(\CE_2(M))=\sL_0(\CE_2(M))\oplus\sL_1(\CE_2(M))\subset \CC^\infty_1(\CE_2(M))\oplus \CC^\infty_0(\CE_2(M))
 \end{equation}
 together with higher products
\begin{equation}
 \begin{aligned}
  \mu_1(f)&=Qf=\theta^M\dpar_M f~,\\
  \mu_2(X,Y)&=-\mu_2(Y,X)=\tfrac12\big(\{QX,Y\}-\{QY,X\}\big)\\
  &=X^M\dpar_M Y-Y^M\dpar_M X+\tfrac12\theta^M(Y^K\dpar_MX_K-X^K\dpar_MY_K)~,\\
  \mu_2(X,f)&=-\mu_2(f,X)=\tfrac12\{QX,f\}=\tfrac12 X^M\dpar_M f~,\\
  \mu_3(X,Y,Z)&=\tfrac13\big(\{\mu_2(X,Y),Z\}+\{\mu_2(Y,Z),X\}+\{\mu_2(Z,X),Y\}\big)\\
  &=X^MZ^N\dpar_MY_N-Y^MZ^N\dpar_MX_N+Y^MX^N\dpar_MZ_N\\
  &\hspace{1cm}-Z^MX^N\dpar_MY_N+Z^MY^N\dpar_MX_N-X^MY^N\dpar_MZ_N~,
 \end{aligned}
\end{equation}
for $f\in \sL_1(\CE_2(M))$ and $X,Y,Z\in \sL_0(\CE_2(M))$, where we wrote $X=X_M\theta^M$, etc.
\end{proposition}

This Lie 2-algebra now has a clear relationship to the symmetry structures in double field theory, and we readily conclude the following statements.
\begin{theorem}
 Given an $L_\infty$-structure on $\CE_2(M)$, the product $\mu_2(X,Y)$ is simply the C-bracket of double field theory. The D-bracket is given by
 \begin{equation}
  \nu_2(X,Y)=\mu_2(X,Y)+\tfrac12Q\{X,Y\}=\{QX,Y\}~.
 \end{equation}
 Moreover, the action $X\acton t$ of vector fields $X\in\sL_0(\CE_2(M))$ on tensors $t\in \sT(\CE_2(M))$ is indeed the action of the extended Lie derivative on tensors as given in \eqref{eq:action_extended_Lie}. In particular, for a rank 2-tensor $t_{MN}\theta^M\otimes \theta^N$, we have 
\begin{equation}
\begin{aligned}
 &X\acton t_{MN} \theta^M\otimes \theta^N:=\{\delta X,t_{MN} \theta^M\otimes \theta^N\}\\
 &\hspace{1cm}=X^N\dpar_N t_{KL} \theta^K\otimes \theta^L+(\dpar_M X^N-\dpar^N X_M)t_{NK} \theta^M\otimes \theta^K+\\
 &\hspace{5cm}+(\dpar_M X^N-\dpar^N X_M)t_{KN} \theta^K\otimes \theta^M~.
\end{aligned}
\end{equation}
\end{theorem}
\noindent Note that since we are working with a pre-N$Q$-manifold, the restrictions of theorem \ref{thm:restrictions_T} apply to extended tensors.

The transformation property of the dilaton $\de^{-2d}$ is fixed by the invariance of the action:
\begin{equation}
 S=\int \dd x^1\wedge \ldots \wedge \dd x^D\wedge \dd x_1\wedge \ldots \wedge \dd x_D~\de^{-2d}~\CR~,
\end{equation}
where $\CR$ is an appropriate Ricci scalar. Note that contrary to ordinary differential geometry, the naive top form is invariant:
\begin{equation}
 X\acton \dd x^1\wedge \ldots \wedge \dd x^D\wedge \dd x_1\wedge \ldots \wedge \dd x_D= X\acton \xi^1\ldots \xi^D\xi_1\ldots \xi_{D}=0~.
\end{equation}
Since $X\acton(\de^{-2d}\CR)=(X\acton \de^{-2d})\CR+\de^{-2d}(X\acton \CR)$, we have to demand that
\begin{equation}
 X\acton \de^{-2d}=\dpar_\mu(X^\mu\de^{-2d})
\end{equation}
in order to obtain a total derivative as a transformation of the action.

\subsection{Examples of \texorpdfstring{$L_\infty$}{L-infinity}-structures}\label{ssec:examples_L_infty}

Recall that an $L_\infty$-structure on $\CE_2(M)$ is a subset $\sL(\CE_2(M))\subset \CC^\infty(\CE_2(M))$ satisfying equations \eqref{eq:restrictions_LM}. Moreover, an $L_\infty$-structure is essentially a weaker replacement of a solution to the strong section condition in double field theory.

The most obvious $L_\infty$-structure is obtained after splitting coordinates $x^M=(x^\mu,x_\mu)$, $\theta^M=(\theta^\mu,\theta_\mu)$ and $p_M=(p_\mu,p^\mu)$ and restricting to functions independent of $x_\mu$. We can regard this as restricting ourselves to the subspace of $\CE_2(M)$ given by $x_\mu=0$. Restricted to this subspace, the original symplectic structure $\omega$ is singular, unless we restrict further to $p^\mu=0$. The resulting subspace is simply the symplectic N$Q$-manifold $\CV_2(M)$ underlying generalized geometry. In this way, double field theory reduces to generalized geometry.

Obviously, we can apply $\sO(D,D)$-rotations to this $L_\infty$-structure to obtain new variants. In particular, the ``dual'' restriction to functions independent of $x^\mu$ and $p_\mu$ works similarly well, and produces an isomorphic solution. Also mixed versions exist: For example, putting
\begin{equation}
 \der{x_\mu}+\pi^{\mu\nu}(x)\der{x^\nu}=0
\end{equation}
for any antisymmetric tensor field $\pi\in\Gamma(\wedge^2 TM)$ provides a solution to \eqref{eq:restrictions_LM}, as one easily verifies. All of these $L_\infty$-structures are also solutions to the strong section condition.

Clearly, it would be very interesting to study the existence of further, more general $L_\infty$-structures, in particular of those which do not satisfy the strong section condition.

\subsection{Twisting extended symmetries}\label{ssec:twisted_extended_symmetries}

In section \ref{ssec:Ex_Courant_and_Gerbes} we explained how in generalized geometry, an exact Courant algebroid with \v Severa class $H$ describes the infinitesimal symmetries of an $\sU(1)$-bundle gerbe with 3-form curvature $H$. Because a solution to the section condition reduces double field theory to generalized geometry, it is clear that in the presence of non-trivial background fluxes, also the symmetries of double field theory require twisting. Moreover, such a twist should play an important role in developing a global description of double field theory, and we will return to this point in section \ref{ssec:global}.

The idea is to mimic the twist of a Courant algebroid and introduce a Hamiltonian
\begin{equation}\label{eq:E_2_twist}
 \CQ_{S,T}=\theta^Mp_M+S_{MN}p^M\theta^N+\tfrac{1}{3!}T_{MNK}\theta^M\theta^N\theta^K~,
\end{equation}
where $T:=\tfrac{1}{3!}T_{MNK}\theta^M\theta^N\theta^K\in \CC^\infty_3(\CE_2(M))$ is some extended 3-form on $\CE_2(M)$ and $S:=S_{MN}p^M\theta^N\in \CC^\infty_3(\CE_2(M))$ is another extended function of degree 3. This is in fact the most general deformation of $\CQ$, as only elements of $\NN$-degree 3 are admissible. The corresponding Hamiltonian vector field with respect to the Poisson bracket \eqref{eq:E2_Poisson} reads as
\begin{equation}
\begin{aligned}
 Q_{S,T}&=\theta^M\der{x^M}+p_M\der{\theta_M}-\frac{1}{3!}\der{x^M}T_{NKL}\theta^N\theta^K\theta^L\der{p_M}+\tfrac12 T_{MNK}\theta^N\theta^K\der{\theta_M}\\
    &~~~~+S_{MN}\theta^N\der{x_M}-\der{x^K}S_{MN}p^M\theta^N\der{p_K}+S_{MN}p^M\der{\theta_N}~,
\end{aligned}
\end{equation}
which is by construction a symplectomorphism on $\CE_2(M)$. We now have
\begin{equation}\label{eq:E_2_Qsq}
\begin{aligned}
\{\CQ_{S,T},\CQ_{S,T}\}&=\theta^M\theta^N\theta^K\theta^L\left(\frac{1}{3}\der{x^M}T_{NKL}+\frac13S_{PM}\der{x_P}T_{NKL}+\frac14T_{MNP}T^P{}_{KL}\right)\\
&~~~~+\theta^M\theta^Np^K\left(2\der{x^M}S_{KN}+S_{PM}\der{x_P}S_{KN}+T_{KMN}+T_{MNL}S_K{}^{L}\right)\\
&~~~~+p^Mp^N\left(\eta_{MN}+2S_{MN}+S_{MK}S_N{}^K\right)~,
\end{aligned}
\end{equation}
where $T^P{}_{KL}:=\eta^{PQ}T_{QKL}$ etc. Again, we should not require that $Q^2_{S,T}=0$, which is equivalent to $\{\CQ_{S,T},\CQ_{S,T}\}=0$, but merely demand that the conditions of theorem \ref{thm:Lie_2_subset} are satisfied. This leads to a twisted $L_\infty$-algebra structure, containing twisted C- and D-brackets. Altogether, we make the following definition.
\begin{definition}
 Given an $L_\infty$-algebra structure $\sL(\CE_2(M))$ on $\CE_2(M)$, we call $\CQ_{S,T}$ in \eqref{eq:E_2_twist} or $Q_{S,T}$ a \uline{twist of $\sL(\CE_2(M))$} if \eqref{eq:conditions_thm_Lie2_subset} is satisfied for elements of $\sL(\CE_2(M))$ if $Q$ is replaced by $Q_{S,T}$.
 
 Given a twist of $\sL(\CE_2(M))$, we define the \uline{twisted D- and C-brackets} by
\begin{equation}
\begin{aligned}
 \nu^{S,T}_2(X,Y)&:=\{Q_{S,T}X,Y\}~,\\
 \mu_2^{S,T}(X,Y)&:=\tfrac12\big(\{Q_{S,T}X,Y\}-\{Q_{S,T}Y,X\}\big)~.
\end{aligned}
\end{equation}
\end{definition}

A detailed study of such twists is beyond the scope of this paper and left to future work. Let us merely present a discussion analogue to that of section \ref{ssec:twisted_Vinogradov} and look at infinitesimal twists. 

We start by considering \eqref{eq:E_2_Qsq}. The term quadratic in $p^Mp^N$ implies that $S_{MN}$ is antisymmetric if it is infinitesimal. By construction, the coordinate transformations
\begin{equation}
 z\mapsto -\{z,\tfrac12 \tau^{MN}\theta_M\theta_N\}
\end{equation}
for $z=(x^M,\theta^M,p_M)$ are symplectomorphisms for infinitesimal $\tau^{MN}$. Explicitly, we have
\begin{equation}
 x^M\rightarrow x^M~,~~~\theta^M\rightarrow \theta^M-\tau^{MN}\theta_N~,~~~p_M=p_M-\tfrac12\dpar_M\tau^{KL}\theta_K\theta_L~,
\end{equation}
and we can use these transformations to put the contribution $S_{MN}p^M\theta^N$ to an infinitesimal twist to zero. Assuming that $\dpar_{[M}T_{NKL]}=0$, we are then left with the conditions 
 \begin{equation}\label{eq:twist_condition_1}
    \{\Xi f,g\}+\{\Xi g,f\}=0~,~~~\{\Xi X,f\}+\{\Xi f,X\}=0~,~~~\{\{\Xi X,Y\},Z\}_{[X,Y,Z]}=0
 \end{equation}
with
\begin{equation}
\begin{aligned}
 \Xi&=\{\theta^M\theta^Np^KT_{KMN},-\}\\
 &=\theta^M\theta^NT_{KMN}\der{x_N}-\theta^M\theta^Np^K\left(\der{x^L}T_{KMN}\right)\der{p_L}+2\theta^M T_{KMN}p^K\der{\theta_N}~.
\end{aligned}
\end{equation}
Equations \eqref{eq:twist_condition_1} then reduce to
\begin{equation}
\begin{aligned}
&\theta^MX^NT_{KMN}\der{x_K}f=0~,\\
&(Y^MX^NT_{KMN}\der{x_K}Z+Z^MX^NT_{KMN}\der{x_K}Y)_{[X,Y,Z]}=0~,
\end{aligned}
\end{equation}
and we arrive at the following theorem.
\begin{theorem}
 Consider an $L_\infty$-structure $\sL(\CE_2(M))$. Then $\CQ=\theta^Mp_M+\tfrac{1}{3!}T_{MNK}\theta^M\theta^N\theta^K$ with $T_{MNK}$ infinitesimal is a twist of $\sL(\CE_2(M))$ if
 \begin{equation}
 \dpar_{[M}T_{NKL]}=0\eand T_{MNK}\der{x_K} F=0
 \end{equation}
 for any $F\in \sL(\CE_2(M))$.
\end{theorem}
As a corollary, we directly obtain the usual twists of Courant algebroids in generalized geometry.
\begin{corollary}
 For the $L_\infty$-structure $\sL(\CE_2(M))$ given by 
  \begin{equation}
    \sL(\CE_2(M))=\left\{F\in \CC^\infty(\CE_2(M))~|~\der{x_\mu}F=0\right\}~,
 \end{equation}
  closed 3-forms are infinitesimal twists.
\end{corollary}

\section{Outlook}

In this last section, we present some partial but interesting results on extended torsion and Riemann tensors as well as some comments on the global picture.

\subsection{Extended torsion and Riemann tensors}\label{ssec:Riemann_tensor}

The extension of torsion and Riemann curvature tensors on a general pre-N$Q$-manifold is an unsolved mathematical problem. It is not even clear, whether it is a good problem to pose. In the following, we present extensions for the case of the Vinogradov algebroids $\CV_1(M)$, $\CV_2(M)$ and the pre-N$Q$-manifold $\CE_2(M)$, which are motivated by standard Riemannian geometry and DFT, respectively. 

As a first step, we have to introduce a notion of covariant derivative on a symplectic pre-N$Q$-manifold $\CM$. Let $n$ be the degree of $\CM$. Just as all our symmetries were described in terms of Hamiltonians acting via the Poisson bracket, we ask that a covariant derivative in the direction $X$ is also a Hamiltonian function.
\begin{definition}
  \label{def:excovder}
 An \uline{extended covariant derivative} $\nabla$ on a pre-N$Q$-manifold $\CM$ is a linear map from $\CCX(\CM)$ to $\CC^\infty(\CM)$ such that the image $\nabla_X$ for $X\in \CCX(\CM)$ gives rise to a map
 \begin{equation}
  \{\nabla_X,-\}:\CCX(\CM)\rightarrow \CCX(\CM)~,
 \end{equation}
 which readily generalizes to extended tensors and satisfies
  \begin{equation}
  \{\nabla_{fX},Y\}=f\{\nabla_X,Y\}\eand\{\nabla_X,fY\}=\{Q X,f\}Y+f\{\nabla_X,Y \}
  \end{equation}
  for all $f\in \CCC^\infty(\CM)$ and extended tensors $Y$. 
\end{definition}
We now restrict to the three cases $\CV_1(M)$, $\CV_2(M)$ and $\CE_2(M)$. To simplify out notation, we use $\hat\CCX(\CM)$ to denote extended vectors and covectors. For the Vinogradov algebroids, $\hat \CCX(\CM)$ are simply the functions linear in the coordinates $\zeta_\mu$ or $\xi^\mu$, cf.\ section \ref{ssec:Vinogradov}, and for $\CE_2(M)$, $\hat \CCX(\CM)$ are the functions linear in the coordinates $\theta^M$, cf.\ section \ref{ssec:E2M_from_restriction}. That is, $\hat \CCX(\CM)=\CCX(\CM)$ for $\CV_2(M)$ and $\CE_2(M)$, while $\CCX(\CM)\subsetneq \hat \CCX(\CM)$ for $\CV_1(M)$. We also use repeatedly the algebra of functions $\CC^\infty(M)=\CC^\infty(\CM_0)$.

We follow the historical development of generalized geometry and introduce the torsion tensor first. 
\begin{definition}\label{def:ext_torsion}
 Let $\CM$ be a pre-N$Q$-manifold. Given an extended connection $\nabla$, we define the \uline{extended torsion tensor} $\CT : \otimes^3 \hat \CCX(\CM) \rightarrow \CC^\infty(M)$ for $X,Y,Z \in \hat \CCX(\CM)$ by  
 \begin{equation}
   \label{extorsion}
  \begin{aligned}
    \CT(X,Y,Z):=\, 3&\Bigl((-1)^{n|X|}\,\Bigl\{X,\{\nabla_{\pi(Y)},Z\}\Bigr\}\Bigr)_{[X,Y,Z]}\\
    &+\frac{(-1)^{n(|Y|+1)}}{2}\left(\{X,\{QZ,Y\}\}-\{Z,\{QX,Y\}\}\right)~,  \end{aligned}
  \end{equation}
  where $|X|,|Y|$ denote the respective $\NN$-degrees and $\pi$ is the obvious projection $\hat \CCX(\CM)\rightarrow \CCX(\CM)$.
\end{definition}
\noindent For ordinary differential geometry, which is captured by $\CM=\CV_1(M)$, $\CT(-,-,-)$ reduces to the ordinary torsion function $T(X,Y,Z) = \langle X,\nabla_Y Z -\nabla_Z Y - [Y,Z]\rangle$, where $X$ is a one-form, $Y,Z$ are vector fields and $\langle X,Y\rangle := \iota_Y X$. It vanishes for all other combinations of forms and vectors.  In the cases of generalized geometry and double field theory, this extension of the torsion tensor boils down to the Gualtieri torsion \cite{Gualtieri:2007bq} and \cite{Hohm:2012mf}, as we shall show later. In particular, it is a tensor and therefore $\CC^\infty(M)$-linear in all its entries for the cases we consider.

In a similar way, we are able to provide a definition of an extended curvature operator which reduces to the curvature tensors of ordinary Riemannian geometry, generalized geometry and double field theory in the respective cases of $\CV_1(M)$, $\CV_2(M)$ and $\CE_2(M)$. Again it uses the Poisson brackets for $\hat\CCX(\CM)$.
\begin{definition}\label{def:ext_Riemann}
  Let $\CM$ be a pre-N$Q$-manifold. Given an extended connection $\nabla$, the \uline{extended curvature operator} $\CR : \otimes^4 \hat\CCX(\CM) \rightarrow \CC^\infty(M)$ for $X,Y,Z,W \in \hat\CCX(\CM)$ is defined by
  \begin{equation}
    \label{exCurvature}
    \begin{aligned}
      \CR(X,Y,Z,W):=\; &\frac{1}{2}\Bigl(\Bigl\{\bigl\{\{\nabla_X,\nabla_Y\} -\nabla_{\mu_2(X,Y)},Z\bigr\},W\Bigr\} -(-1)^n( Z\leftrightarrow W)   \\
& + \Bigl\{\bigl\{\{\nabla_Z,\nabla_W\}-\nabla_{\{\nabla_Z,W\} - \{\nabla_W,Z\}},X\bigr\},Y\Bigr\} - (-1)^n(X\leftrightarrow Y)\Bigr)\;.
\end{aligned}
  \end{equation}
\end{definition}
In the following, we will discuss these two tensors for the three symplectic pre-N$Q$-manifolds of interest to us. In the case of $\CE_2(M)$, we will make the remarkable observation that tensoriality of the curvature operator is established using the constraints \eqref{eq:restrictions_LM}, underlining the fact that these are the appropriate strong constraints for functions and vector fields.

Let us start with the case of Riemannian geometry, which is captured by the symplectic pre-N$Q$-manifold $\CM=\CV_1(M)$. We need to show that definitions \ref{def:excovder}, \ref{def:ext_torsion} and \ref{def:ext_Riemann} reduce correctly to the expected objects of Riemannian geometry. 

The following Hamiltonian for $X\in \CCX(\CM)$ has the properties of definition \ref{def:excovder} and reproduces the right expressions for the covariant derivative of forms and vector fields:
\begin{equation}
  \nabla_X =\,X^\mu p_\mu -X^\mu\Gamma^\rho{}_{\mu\nu}\zeta_\rho \xi^\nu\;.
\end{equation}
The coefficients $\Gamma^\rho{}_{\mu\nu}$ are indeed the usual Christoffel symbols.
\begin{proposition}
  For $\CM=\CV_1(M)$, let $X \in \CCX^*(\CM)$ and $Y,Z \in \CCX(\CM)$, then  torsion as defined in \eqref{extorsion} reduces to the torsion operator $T(X,Y,Z)=\langle X,\nabla_Y Z - \nabla_Z Y - [Y,Z]\rangle$, where the bracket is the Lie bracket of vector fields. More generally, this is true whenever we take one element of $\CCX^*(\CM)$ and the other two in $\CCX(\CM)$. In all other cases the extended torsion vanishes.
\end{proposition}
\begin{proof}
  The proof is done by simply writing out the Poisson brackets. Let $X\in\CCX^*(\CM)$ and $Y,Z \in \CCX(\CM)$. Then we get
  \begin{equation}
    \begin{aligned}
      3&\Bigl((-1)^{n|X|}\,\Bigl\{X,\{\nabla_Y,Z\}\Bigr\}\Bigr)_{[X,Y,Z]}\\=& \frac{1}{2}\Bigl(-\Bigl\lbrace X,\lbrace \nabla_Y Z\rbrace\Bigr\rbrace + \Bigr\lbrace X,\lbrace \nabla_Z,Y\rbrace\Bigr\rbrace +\Bigl\lbrace Y,\lbrace \nabla_Z,X\rbrace\Bigr\rbrace - \Bigl\lbrace Z,\lbrace \nabla_Y,X\rbrace\Bigr\rbrace \Bigr)\\
      =&\frac{1}{2}\Bigl(-Y^\mu Z^\nu(\partial_\mu X_\nu - \partial_\nu X_\mu) + X_\rho Y^\mu Z^\nu(\Gamma^\rho{}_{\mu \nu} - \Gamma^\rho{}_{\nu \mu}) +\langle X,\nabla_Y Z - \nabla_Z Y\rangle \Bigr)\;.
    \end{aligned}
  \end{equation}
  Similarly, expanding out the second part in the definition of extended torsion, we get
  \begin{equation}
      \frac{1}{2}\Bigl(\Bigl\lbrace X, \lbrace QZ,Y\rbrace \Bigr\rbrace  - \Bigl\lbrace Z,\lbrace QX,Y\rbrace \Bigr\rbrace\Bigr) 
      =\frac{1}{2}\Bigl(-\langle X,[Y,Z]\rangle + Y^\mu Z^\nu(\partial_\mu X_\nu - \partial_\nu X_\mu)\Bigr)\;.
  \end{equation}
  Because torsion components in Riemannian geometry are given by $\Gamma^\rho{}_{\mu \nu} - \Gamma^\rho{}_{\nu\mu}$, the first claim follows. The calculation is similar whenever there is one element in $\CCX^*(\CM)$. If all elements are in $\CCX(\CM)$, all terms are separately zero, as there are only Poisson brackets of vectors and vectors, which vanish. If there are two elements or all three elements in $\CCX^*(\CM)$, one is left with brackets of two forms, which also vanish.  
\end{proof}

We also have the expected statement for the extended curvature operator:
\begin{proposition}
  For $\CM=\CV_1(M)$, let $X,Y,Z \in \CCX(\CM)$ and $W \in \CCX^*(\CM)$. Then the extended curvature \eqref{exCurvature} reduces to the standard curvature:
  \begin{equation}
    \CR(X,Y,Z,W) = \langle W,\nabla_X\nabla_Y Z - \nabla_Y\nabla_X Z - \nabla_{\mu_2(X,Y)} Z \rangle\,.
  \end{equation}
  Furthermore, if $X,Y,Z,W \in \CCX(\CM)$ or if two, three or all of $X,Y,Z,W$ are in $\CCX^*(\CM)$, we have $\CR(X,Y,Z,W)=0$.
\end{proposition}
\begin{proof}
  Similar to the previous proposition, the proof is done by checking the Poisson brackets. For the last statement, we recall that the binary product $\mu_2$ was defined as
  \begin{equation}
    \mu_2(X,Y):= \,\tfrac{1}{2}(\lbrace QX,Y\rbrace - \lbrace QY,X\rbrace)\,.
  \end{equation}
  If $X \in \CCX^*(\CM)$ and $Y\in \CCX(\CM)$, $\mu_2(X,Y)$ is in $\CCX^*(\CM)$, so in this case $\nabla_{\mu_2(X,Y)}=0$. The statement follows by using the latter and checking the different cases.
\end{proof}

So far, we have seen that for the right Vinogradov algebroid $\CM=\CV_1(M)$, we recover the defining objects of ordinary Riemannian geometry. In this sense, we \emph{extend} ordinary Riemannian geometry. In the following, we describe the remaining two cases, yielding essential aspects of torsion and curvature in generalized geometry and double field theory. 

Since generalized geometry is readily obtained by reduction of double field theory, we focus now on the case $\CM=\CE_2(M)$. Generalized vector fields are given by extended vector fields, which we denote by $\CCX(\CM) \ni X= X^M \theta_M$ and as noted above, we here have $\hat\CCX(\CM)=\CCX(\CM)$. The extended covariant derivative with the properties of definition \ref{def:excovder} now reads
\begin{equation}\label{Covder}
  \nabla_X =\, X^Mp_M - \tfrac{1}{2}\,X^M\Gamma_{MNK}\theta^N\theta^K\;.
\end{equation}
Already from the covariant constancy of the bilinear form $\eta$, we see the antisymmetry of the extended connection coefficients:
\begin{equation}
  0=\,\partial_M\eta_{NK} = \,\Gamma_{MNK}+\Gamma_{MKN}\,.
\end{equation}
As was shown in detail in \cite{Hohm:2012mf}, the appropriate notion of torsion is the Gualtieri torsion, which in generalized geometry reads
\begin{equation}
  \textrm{GT}(X,Y,Z)=\,\langle \nabla_X Y - \nabla_Y X - \nu_2(X,Y),Z\rangle + \langle Y,\nabla_Z X\rangle\,,
\end{equation}
where $\nu_2(-,-)$ denotes the Dorfman bracket of generalized geometry, and $X,Y,Z$ are generalized vectors. Furthermore, it was shown in \cite{Hohm:2012mf}, that a Riemann curvature operator can be defined, see also \cite{Jurco:2015xra} for related results. We will now reformulate these objects in our language and investigate their $\CC^\infty(M)$-linearity properties. We begin with the torsion tensors.
\begin{proposition}
  For extended vector fields $X,Y,Z \in \CCX(\CM)$, the extended torsion tensor equals the Gualtieri-torsion:
\begin{equation}
  \CT(X,Y,Z) =\,\textrm{GT}(X,Y,Z)\;.
\end{equation} 
\end{proposition}
The proof is done by explicit calculation. Non-tensorial derivative terms cancel out and the remaining terms are
\begin{equation}
\label{gualtieri2}
\CT(X,Y,Z) =\,X^M Y^N Z^K(\Gamma_{MNK}-\Gamma_{NMK}+\Gamma_{KMN})\;,
\end{equation}
coinciding with the original torsion found by Gualtieri and used in \cite{Hohm:2012mf} in the context of double field theory.

Now we turn to the extended curvature tensor in the case of $\CM=\CE_2(M)$. We discuss the result of \cite{Hohm:2012mf} in terms of the language set up in this work. As a simplification, we restrict ourselves to the case of vanishing extended torsion \eqref{gualtieri2}. First, we observe that the Poisson-bracket of two extended covariant derivatives contains the tensorial part of the standard Riemann tensor, but on the doubled space\footnote{Here we use the abbreviation $[X,Y]^K$ for the Lie-derivative part, i.e.\ $[X,Y]^K = \,X^M\partial_M Y^K - Y^M\partial_M X^K$.}:
\begin{equation}
\label{standardRiem}
\begin{aligned}
\bigl\{\{\nabla_X,&\nabla_Y\},Z\bigr\} = \Bigl([X,Y]^M\partial_M Z^R +[X,Y]^K\,\Gamma_{KL}{}^R\,Z^L\Bigr)\theta_R \\ 
&+X^MY^NZ^K\bigl(\partial_M\Gamma_{NK}{}^R - \partial_N \Gamma_{MK}{}^R + \Gamma_M{}^{RQ}\Gamma_{NQK} - \Gamma_N{}^{RQ}\Gamma_{MQK}\bigr)\theta_R\;.
\end{aligned}
\end{equation}
In standard Riemannian geometry, the Lie-derivative terms are canceled by the covariant derivative with respect to the Lie derivative of the corresponding vector fields. In double field theory, infinitesimal symmetries are determined by the C-bracket $\mu_2$, so we have to replace the Lie derivative by the latter. The resulting expression is not tensorial, and we have to add further terms to correct for this, as done in \eqref{exCurvature}. We arrive at the following result, written in terms of Poisson brackets:
\begin{proposition}
  For $\CM=\CE_2(M)$, the extended curvature operator $\CR$, given by \eqref{exCurvature} is tensorial up to terms that vanish after imposing the constraints \eqref{eq:restrictions_LM}.
  \end{proposition}
\begin{proof}
First, note that for $n=2$ the sign factors in \eqref{exCurvature} drop out. Denoting the standard Riemannian curvature combination by $R_{MNKR}$, i.e.\ 
\begin{equation}
R_{MNKR} = \partial_M\Gamma_{NKR} - \partial_N \Gamma_{MKR} + \Gamma_{MR}{}^{Q}\Gamma_{NQK} - \Gamma_{NR}{}^{Q}\Gamma_{MQK}\;,
\end{equation}
and writing out the expression \eqref{exCurvature} in terms of components, we get
\begin{equation}
\begin{aligned}
{\cal R}(X,Y,&Z,W) =\; X^MY^NZ^KW^R(R_{MNKR}+R_{KRMN}) \\
&+\frac{1}{4}(X_M\partial^K Y^M - Y_M\partial^K X^M)(W_N\partial_K Z^N - Z_N\partial_K W^N) \\
&+\frac{1}{4}(X_M\partial^K Y^M - Y_M\partial^K X^M)\Gamma_{KPQ}(Z^PW^Q - W^P Z^Q) \\
&+\frac{1}{2}\Bigl( Z^N W^K(-\Gamma_{NK}{}^M + \Gamma_{KN}{}^M)(Y_P\partial_M X^P - X_P\partial_M Y^P)\Bigr) \\
&+\frac{1}{2}\Bigl(Z^N W^K(\Gamma_{NK}{}^M - \Gamma_{KN}{}^M)\Gamma_{MPQ}(-X^P Y^Q + Y^PX^Q)\Bigr)\;.
\end{aligned}
\end{equation}
We now use the vanishing of the Gualtieri torsion \eqref{gualtieri2} to simplify the combinations of connection coefficients, which yields
\begin{equation}
\begin{aligned}
{\cal R}(X,Y,Z,W)=\;&X^MY^NZ^KW^R(R_{MNKR}+R_{KRMN} + \Gamma^Q{}_{MN}\Gamma_{QKR})\\
&+\tfrac{1}{4}(X_M\partial^K Y^M - Y_M\partial^K X^M)(W_N\partial_K Z^N - Z_N\partial_K W^N)\;.
\end{aligned}
\end{equation}
From the last expression, we see, that all terms in the first line are manifestly tensorial, whereas the term in the second line is not. However, the terms destroying tensoriality vanish due to the constraints \eqref{eq:restrictions_LM}. For example, the failure of $\CC^\infty(M)$-linearity in $X$ reads as
\begin{equation}
{\cal R}(fX,Y,Z,W)-f{\cal R}(X,Y,Z,W)  = -\tfrac{1}{4}\,X_MY^M\,\partial^Kf(W_N\partial_K Z^N - Z_N\partial_K W^N)\;,
\end{equation}
and the terms of the form $\{\partial^Kf\, \partial_K Z,W\}$ and $\{\partial^Kf \,\partial_K W,Z\}$ vanish due to the strong constraint in the form of the first two lines of \eqref{eq:restrictions_LM}. 
\end{proof}
Finally we comment on the Ricci and scalar curvatures in the case of $\CM=\CE_2(M)$. By definition, the Ricci scalar is the trace of an endomorphism of the doubled tangent bundle, more precisely, we have 
\begin{equation}
\textrm{Ric}(X,Y)=\; \textrm{tr}(Z\mapsto {\cal R}(X,Z)Y)\;.
\end{equation}
In our case, this means that $\textrm{Ric}(X,Y) =\,{\cal R}(\theta^I,X,\theta_I,Y)$. Using the result \eqref{exCurvature}, we see that the derivative terms which are only tensorial up to the constraints \eqref{eq:restrictions_LM} cancel. Furthermore, denoting the standard combination for the Ricci tensor using \eqref{standardRiem} by $\textrm{Ric}_0$, we arrive at
\begin{equation}
\textrm{Ric}(X,Y)=\; \textrm{Ric}_0(X,Y) + \textrm{Ric}_0(Y,X) +X^M Y^N\,\Gamma^{QK}{}_M\Gamma_{QKN}\;.
\end{equation}
Finally, the Ricci-scalar is computed by contracting the Ricci-tensor. We note that this is done by the bilinear form $\eta$, in contrast to the generalized metric $\CH$ (cf.\ e.g.\ \cite{Hohm:2011si}). Thus we obtain the Ricci scalar
\begin{equation}
\label{scal}
{\cal R} = \eta^{MN}\,\textrm{Ric}(E_M,E_N) = 2R_0 + \Gamma^{MNK}\Gamma_{MNK}\;.
\end{equation}
Note that here and in the following, we raise and lower indices with the constant form $\eta$. Contractions with the generalized metric $\CH$ are written out explicitly to emphasize the appearance of the generalized metric as the dynamical field.

To write down actions, it is necessary to express the connection coefficients $\Gamma_{MNK}$ of \eqref{Covder} in terms of the the generalized metric $\CH$, which is the dynamical field of the theory. This is done in a similar way as in \cite{Hohm:2011si}. We review some details and state the result. In contrast to ordinary Riemannian geometry, where vanishing torsion and metricity uniquely determine the Levi-Civita connection via the Koszul formula, one has three different constraints on the connection coefficients in double field theory:
\begin{itemize}
\item $\textrm{GT}=0 \quad \leftrightarrow \quad \Gamma_{KMN}-\Gamma_{MKN} + \Gamma_{NKM} = 0\;,$
\item $\nabla \eta = 0\quad \leftrightarrow \quad \Gamma_{KMN}+\Gamma_{KNM}=0\;,$
\item $\nabla \CH = 0 \quad \leftrightarrow \quad \partial_K \CH_{MN} = \Gamma_{KMP} \CH^P_N + \Gamma_{KNP}\CH^P_M \;.$
\end{itemize}
In \cite{Hohm:2011si}, it was shown that these conditions do not determine the connection coefficients uniquely. However, different choices of connection coefficients give rise to the same DFT actions, so we will use the most common solution to write down actions in our language. Before stating the solution according to \cite{Hohm:2011si}, we remark that in case of a non-trivial dilaton there is a further constraint involving the trace of the connection coefficients:
\begin{itemize}
\item $\Gamma_{NMK}\eta^{NK} =\,-2\partial_M d\,$.
\end{itemize}
It was shown that the four constraints give a solution to the connection coefficients in terms of the generalized metric which is non-unique but different solutions lead to the same Ricci scalar:
\begin{proposition}
\cite{Hohm:2011si}.
In case of vanishing Gualtieri-torsion, $\CH$-metricity, covariant constancy of the form $\eta$ and the dilaton constraint, a possible solution for the connection coefficients is given by
\begin{equation}
\label{connectioncoef}
\begin{aligned}
  \Gamma_{MNK} =&\,\tfrac{1}{2}\CH_{KQ}\partial_M\CH^Q{}_N + \tfrac{1}{2}(\delta_{[N}{}^P\CH_{K]}{}^Q + \CH_{[N}{}^P\delta_{K]}{}^Q)\partial_P\CH_{QM} \\
  &+\frac{2}{D-1}(\eta_{M[N}\delta_{K]}{}^Q + \CH_{M[N}\CH_{K]}{}^Q)(\partial_Q d + \tfrac{1}{4}\CH^{PM}\partial_M\CH_{PQ})\;,
\end{aligned}
\end{equation}
where $D$ is the dimension of the manifold underlying the doubled manifold. 
\end{proposition}
Using these connection coefficients, an action for double field theory was formulated in \cite{Hohm:2011si} using the extended Ricci scalar \eqref{scal}.

\subsection{Comments on a global description}\label{ssec:global}

There is now an obvious procedure for turning our local description of double field theory into a global one. The first step consists of integrating the infinitesimal symmetries described by the Lie 2-algebra structure on $\CE_2(M)$ to finite symmetries. The second step is then to use the resulting finite symmetries to patch together local descriptions.

In principle, it is possible to integrate any Lie 2-algebra, using the various techniques in the literature \cite{Getzler:0404003,Henriques:2006aa,Severa:1506.04898}. These techniques, however, are very involved, and the outcome is often cumbersome in the sense that a categorical equivalence has to be applied to it to be useful.

Fortunately, we do not really require the finite symmetries themselves, but merely their action on extended tensors. This action is simply given by a Lie algebra, cf.\ section \ref{ssec:Symmetries_of_DFT} and equation \eqref{eq:L_infity_action}, which is readily integrated. Moreover, this has been done before in \cite{Hohm:2013bwa,Hohm:2012gk}.

A proposal of how to glue together local descriptions of double field theory with finite symmetries was made in \cite{Berman:2014jba}. In \cite{Papadopoulos:2014mxa}, however, it was shown that this procedure works at most for gerbes with trivial or purely torsion characteristic class. That is, the 3-form curvature $H$ of the $B$-field is (globally) exact and thus does not exhibit any topological non-triviality. This is clearly too restrictive for a general application in string theory. 

From our perspective, this is not very surprising, as non-trivial background fluxes require to work with the twisted symmetries introduced in section \ref{ssec:twisted_extended_symmetries}. It is the Lie algebra of actions of the twisted infinitesimal symmetries that should be integrated and used in the patching of local descriptions of double field theories with non-trivial background fluxes. We hope to report on progress in this direction in future work.

One can, however, adopt a different standpoint, which also explains that a topologically non-trivial 3-form flux is precisely where things break down. As shown in \cite{Bouwknegt:2003vb}, global T-duality with non-trivial $H$-flux induces a change in topology. In particular, the $H$-flux will encode the topology of the T-dual circle fibration. Even worse, performing a T-duality along the fibers of the torus bundle $T^2\embd T^3\rightarrow S^1$ with $H$-flux the volume form, it was observed in \cite{Kachru:2002sk} that a circle disappears under T-duality. The T-dual was then interpreted in \cite{Mathai:2004qq} as a continuous field of stabilized noncommutative tori fibered over $S^1$. It is very hard to imagine that this situation can be captured by gluing together the current local description of double field theory. Inversely, however, there may be some room to turn pre-N$Q$-manifolds into ordinary noncommutative and/or non-associative N$Q$-manifolds.

Finally, let us note that in generalized geometry, one might be tempted to integrate the underlying Courant algebroid $\CV_2(M)$, which is a symplectic Lie 2-algebroid, to a symplectic Lie 2-groupoid, using the methods of \cite{LiBland:2011aa,Severa:1506.04898}. It would be natural to expect that the associated Lie $2$-algebra of infinitesimal symmetries survives this integration, but this is not clear at all. 

Moreover, when trying to apply this integration method to the case of double field theory, we face the problem that the underlying geometric structure is a pre-N$Q$-manifold instead of an N$Q$-manifold.

\section*{Acknowledgements}

We would like to thank David Berman, Ralph Blumenhagen, André Coimbra, Dieter L{\"u}st and Jim Stasheff for discussions. The work of CS was partially supported by the Consolidated Grant ST/L000334/1 from the UK Science and Technology Facilities Council. AD and CS want to thank the Department of Mathematics of Heriot-Watt University and the Institut f\"ur Theoretische Physik in Hannover for hospitality, respectively.

\bibliography{bigone}

\begin{thebibliography}{10}

\bibitem{Giveon:1994fu}
A.~Giveon, M.~Porrati, and E.~Rabinovici,
{\em Target space duality in string theory,}
\href{http://dx.doi.org/10.1016/0370-1573(94)90070-1}{Phys. Rept {\bf 244}
  (1994)~77} [{\tt
  \href{http://www.arxiv.org/abs/hep-th/9401139}{hep-th/9401139}}].

\bibitem{Tseytlin:1990va}
A.~A.~Tseytlin,
{\em {Duality symmetric closed string theory and interacting chiral scalars},}
\href{http://dx.doi.org/10.1016/0550-3213(91)90266-Z}{Nucl. Phys. B {\bf 350}
  (1991) 395}.
%%CITATION = NUPHA,B350,395;%%

\bibitem{Siegel:1993th}
W.~Siegel,
{\em Superspace duality in low-energy superstrings,}
\href{http://dx.doi.org/10.1103/PhysRevD.48.2826}{Phys. Rev. D {\bf 48} (1993)
  2826} [{\tt \href{http://www.arxiv.org/abs/hep-th/9305073}{hep-th/9305073}}].

\bibitem{Siegel:1993xq}
W.~Siegel,
{\em Two-vierbein formalism for string-inspired axionic gravity,}
\href{http://dx.doi.org/10.1103/PhysRevD.47.5453}{Phys. Rev. D {\bf 47} (1993)
  5453} [{\tt \href{http://www.arxiv.org/abs/hep-th/9302036}{hep-th/9302036}}].

\bibitem{Hull:2004in}
C.~M.~Hull,
{\em A geometry for non-geometric string backgrounds,}
\href{http://dx.doi.org/10.1088/1126-6708/2005/10/065}{JHEP {\bf 0510} (2005)
  065} [{\tt \href{http://www.arxiv.org/abs/hep-th/0406102}{hep-th/0406102}}].

\bibitem{Hull:2009sg}
C.~M.~Hull and R.~A.~Reid{--}Edwards,
{\em {Non-geometric backgrounds, doubled geometry and generalised T-duality},}
\href{http://dx.doi.org/10.1088/1126-6708/2009/09/014}{JHEP {\bf 0909} (2009)
  014} [{\tt \href{http://www.arxiv.org/abs/0902.4032}{0902.4032 [hep-th]}}].
%%CITATION = ARXIV:0902.4032;%%

\bibitem{Hull:2009mi}
C.~Hull and B.~Zwiebach,
{\em Double field theory,}
\href{http://dx.doi.org/10.1088/1126-6708/2009/09/099}{JHEP {\bf 0909}
  (2009)~99} [{\tt \href{http://www.arxiv.org/abs/0904.4664}{0904.4664
  [hep-th]}}].

\bibitem{Zwiebach:2011rg}
B.~Zwiebach,
{\em {Double field theory, T-duality, and Courant brackets},}
\href{http://dx.doi.org/10.1007/978-3-642-25947-0_7}{Lect. Notes Phys. {\bf
  851} (2012) 265} [{\tt \href{http://www.arxiv.org/abs/1109.1782}{1109.1782
  [hep-th]}}].
%%CITATION = ARXIV:1109.1782;%%

\bibitem{Berman:2013eva}
D.~S.~Berman and D.~C.~Thompson,
{\em {Duality symmetric string and M-theory},}
\href{http://dx.doi.org/10.1016/j.physrep.2014.11.007}{Phys. Rept. {\bf 566}
  (2014)~1} [{\tt \href{http://www.arxiv.org/abs/1306.2643}{1306.2643
  [hep-th]}}].
%%CITATION = ARXIV:1306.2643;%%

\bibitem{Aldazabal:2013sca}
G.~Aldazabal, D.~Marques, and C.~Nunez,
{\em Double field theory: A pedagogical review,}
\href{http://dx.doi.org/10.1088/0264-9381/30/16/163001}{Class. Quant. Grav.
  {\bf 30} (2013) 163001} [{\tt
  \href{http://www.arxiv.org/abs/1305.1907}{1305.1907 [hep-th]}}].

\bibitem{Hohm:2013bwa}
O.~Hohm, D.~L{\"u}st, and B.~Zwiebach,
{\em {The spacetime of double field theory: Review, remarks, and outlook},}
\href{http://dx.doi.org/10.1002/prop.201300024}{Fortsch. Phys. {\bf 61} (2013)
  926} [{\tt \href{http://www.arxiv.org/abs/1309.2977}{1309.2977 [hep-th]}}].
%%CITATION = ARXIV:1309.2977;%%

\bibitem{Hitchin:2004ut}
N.~Hitchin,
{\em Generalized Calabi--Yau manifolds,}
\href{http://dx.doi.org/10.1093/qjmath/54.3.281}{Quart. J. Math. Oxford Ser.
  {\bf 54} (2003) 281} [{\tt
  \href{http://www.arxiv.org/abs/math.DG/0209099}{math.DG/0209099}}].

\bibitem{Hitchin:2005in}
N.~Hitchin,
{\em {Brackets, forms and invariant functionals},}
{\tt \href{http://www.arxiv.org/abs/math.DG/0508618}{math.DG/0508618}}.
%%CITATION = MATH/0508618;%%

\bibitem{Gualtieri:2003dx}
M.~Gualtieri,
{\em {Generalized complex geometry},} PhD thesis, Oxford (2003)
[{\tt \href{http://www.arxiv.org/abs/math.DG/0401221}{math.DG/0401221}}].
%%CITATION = MATH/0401221;%%

\bibitem{Vaisman:2012ke}
I.~Vaisman,
{\em {On the geometry of double field theory},}
\href{http://dx.doi.org/10.1063/1.3694739}{J. Math. Phys. {\bf 53} (2012)
  033509} [{\tt \href{http://www.arxiv.org/abs/1203.0836}{1203.0836
  [math.DG]}}].
%%CITATION = ARXIV:1203.0836;%%

\bibitem{Vaisman:2012px}
I.~Vaisman,
{\em {Towards a double field theory on para-Hermitian manifolds},}
\href{http://dx.doi.org/10.1063/1.4848777}{J. Math. Phys. {\bf 54} (2013)
  123507} [{\tt \href{http://www.arxiv.org/abs/1209.0152}{1209.0152
  [math.DG]}}].
%%CITATION = ARXIV:1209.0152;%%

\bibitem{Hohm:2012mf}
O.~Hohm and B.~Zwiebach,
{\em {Towards an invariant geometry of double field theory},}
\href{http://dx.doi.org/10.1063/1.4795513}{J. Math. Phys. {\bf 54} (2013)
  032303} [{\tt \href{http://www.arxiv.org/abs/1212.1736}{1212.1736
  [hep-th]}}].
%%CITATION = ARXIV:1212.1736;%%

\bibitem{Cederwall:2014kxa}
M.~Cederwall,
{\em {The geometry behind double geometry},}
\href{http://dx.doi.org/10.1007/JHEP09(2014)070}{JHEP {\bf 1409} (2014) 070}
  [{\tt \href{http://www.arxiv.org/abs/1402.2513}{1402.2513 [hep-th]}}].
%%CITATION = ARXIV:1402.2513;%%

\bibitem{Blumenhagen:2014gva}
R.~Blumenhagen, F.~Hassler, and D.~L{\"u}st,
{\em Double field theory on group manifolds,}
\href{http://dx.doi.org/10.1007/JHEP02(2015)001}{JHEP {\bf 1502} (2015) 001}
  [{\tt \href{http://www.arxiv.org/abs/1410.6374}{1410.6374 [hep-th]}}].

\bibitem{Deser:2014mxa}
A.~Deser and J.~Stasheff,
{\em {Even symplectic supermanifolds and double field theory},}
\href{http://dx.doi.org/10.1007/s00220-015-2443-4}{Commun. Math. Phys. {\bf
  339} (2015) 1003} [{\tt \href{http://www.arxiv.org/abs/1406.3601}{1406.3601
  [math-ph]}}].
%%CITATION = ARXIV:1406.3601;%%

\bibitem{Bakas:2016nxt}
I.~Bakas, D.~L{\"u}st, and E.~Plauschinn,
{\em {Towards a world-sheet description of doubled geometry in string theory},}
\href{http://dx.doi.org/10.1002/prop.201600085}{Fortsch. Phys. {\bf 64} (2016)
  730} [{\tt \href{http://www.arxiv.org/abs/1602.07705}{1602.07705 [hep-th]}}].
%%CITATION = ARXIV:1602.07705;%%

\bibitem{Aschieri:2015roa}
P.~Aschieri and R.~J.~Szabo,
{\em {Triproducts, nonassociative star products and geometry of R-flux string
  compactifications},}
\href{http://dx.doi.org/10.1088/1742-6596/634/1/012004}{J. Phys. Conf. Ser.
  {\bf 634} (2015) 012004} [{\tt
  \href{http://www.arxiv.org/abs/1504.03915}{1504.03915 [hep-th]}}].
%%CITATION = ARXIV:1504.03915;%%

\bibitem{Gawedzki:1987ak}
K.~Gawedzki,
{\em {Topological actions in two-dimensional quantum field theories},}
Nonperturbative quantum field theory (Carg\`ese, 1987), 101--141, NATO Adv.
  Sci. Inst. Ser. B Phys., 185, Plenum, New York, 1988.

\bibitem{Freed:1999vc}
D.~S.~Freed and E.~Witten,
{\em Anomalies in string theory with D-branes,}
Asian J. Math {\bf 3} (1999) 819 [{\tt
  \href{http://www.arxiv.org/abs/hep-th/9907189}{hep-th/9907189}}].
%%CITATION = HEP-TH/9907189;%%

\bibitem{Lada:1992wc}
T.~Lada and J.~Stasheff,
{\em {Introduction to sh Lie algebras for physicists},}
\href{http://dx.doi.org/10.1007/BF00671791}{Int. J. Theor. Phys. {\bf 32}
  (1993) 1087} [{\tt
  \href{http://www.arxiv.org/abs/hep-th/9209099}{hep-th/9209099}}].
%%CITATION = HEP-TH/9209099;%%

\bibitem{Roytenberg:1998vn}
D.~Roytenberg and A.~Weinstein,
{\em {Courant algebroids and strongly homotopy Lie algebras},}
\href{http://dx.doi.org/10.1023/A:1007452512084}{Lett. Math. Phys. {\bf 46}
  (1998)~81} [{\tt
  \href{http://www.arxiv.org/abs/math.QA/9802118}{math.QA/9802118}}].
%%CITATION = MATH/9802118;%%

\bibitem{Roytenberg:1999aa}
D.~Roytenberg,
{\em Courant algebroids, derived brackets and even symplectic supermanifolds,}
  PhD thesis, UC Berkeley (1999)
[{\tt \href{http://www.arxiv.org/abs/math.DG/9910078}{math.DG/9910078}}].

\bibitem{Kosmann-Schwarzbach:0312524}
Y.~Kosmann-Schwarzbach,
{\em Derived brackets,}
\href{http://dx.doi.org/10.1007/s11005-004-0608-8}{Lett. Math. Phys. {\bf 69}
  (2004)~61} [{\tt
  \href{http://www.arxiv.org/abs/math.DG/0312524}{math.DG/0312524}}].

\bibitem{Blumenhagen:2016vpb}
R.~Blumenhagen and M.~Fuchs,
{\em {Towards a theory of nonassociative gravity},}
\href{http://dx.doi.org/10.1007/JHEP07(2016)019}{JHEP {\bf 1607} (2016) 019}
  [{\tt \href{http://www.arxiv.org/abs/1604.03253}{1604.03253 [hep-th]}}].
%%CITATION = ARXIV:1604.03253;%%

\bibitem{Berman:2014jba}
D.~S.~Berman, M.~Cederwall, and M.~J.~Perry,
{\em {Global aspects of double geometry},}
\href{http://dx.doi.org/10.1007/JHEP09(2014)066}{JHEP {\bf 1409} (2014) 066}
  [{\tt \href{http://www.arxiv.org/abs/1401.1311}{1401.1311 [hep-th]}}].
%%CITATION = ARXIV:1401.1311;%%

\bibitem{Papadopoulos:2014mxa}
G.~Papadopoulos,
{\em {Seeking the balance: Patching double and exceptional field theories},}
\href{http://dx.doi.org/10.1007/JHEP10(2014)089}{JHEP {\bf 1410} (2014) 089}
  [{\tt \href{http://www.arxiv.org/abs/1402.2586}{1402.2586 [hep-th]}}].
%%CITATION = ARXIV:1402.2586;%%

\bibitem{Gualtieri:2007bq}
M.~Gualtieri,
{\em {Branes on Poisson varieties},}
{\tt \href{http://www.arxiv.org/abs/0710.2719}{0710.2719 [math.DG]}}.
%%CITATION = ARXIV:0710.2719;%%

\bibitem{Jeon:1011.1324}
I.~Jeon, K.~Lee, and J.-H.~Park,
{\em Differential geometry with a projection: Application to double field
  theory,}
\href{http://dx.doi.org/10.1007/JHEP04(2011)014}{JHEP {\bf 1104} (2011)~14}
  [{\tt \href{http://www.arxiv.org/abs/1011.1324}{1011.1324 [hep-th]}}].

\bibitem{Jeon:2011cn}
I.~Jeon, K.~Lee, and J.-H.~Park,
{\em {Stringy differential geometry, beyond Riemann},}
\href{http://dx.doi.org/10.1103/PhysRevD.84.044022}{Phys. Rev. D {\bf 84}
  (2011) 044022} [{\tt \href{http://www.arxiv.org/abs/1105.6294}{1105.6294
  [hep-th]}}].
%%CITATION = ARXIV:1105.6294;%%

\bibitem{Hohm:2011si}
O.~Hohm and B.~Zwiebach,
{\em {On the Riemann tensor in double field theory},}
\href{http://dx.doi.org/10.1007/JHEP05(2012)126}{JHEP {\bf 1205} (2012) 126}
  [{\tt \href{http://www.arxiv.org/abs/1112.5296}{1112.5296 [hep-th]}}].
%%CITATION = ARXIV:1112.5296;%%

\bibitem{Jurco:2015xra}
B.~Jurco and J.~Vysoky,
{\em {Leibniz algebroids, generalized Bismut connections and Einstein–Hilbert
  actions},}
\href{http://dx.doi.org/10.1016/j.geomphys.2015.06.017}{J. Geom. Phys. {\bf 97}
  (2015)~25} [{\tt \href{http://www.arxiv.org/abs/1503.03069}{1503.03069
  [hep-th]}}].
%%CITATION = ARXIV:1503.03069;%%

\bibitem{Deser:2017aa}
A.~Deser, M.~Heller, and C.~Saemann,
{\em Extended Riemannian geometry III: Local exceptional field theory,}
in preparation.

\bibitem{Grana:2008yw}
M.~Graña, R.~Minasian, M.~Petrini, and D.~Waldram,
{\em T-duality, generalized geometry and non-geometric backgrounds,}
\href{http://dx.doi.org/10.1088/1126-6708/2009/04/075}{JHEP {\bf 0904} (2009)
  075} [{\tt \href{http://www.arxiv.org/abs/0807.4527}{0807.4527 [hep-th]}}].

\bibitem{Shapere:1988zv}
A.~D.~Shapere and F.~Wilczek,
{\em {Selfdual models with theta terms},}
\href{http://dx.doi.org/10.1016/0550-3213(89)90016-3}{Nucl. Phys. B {\bf 320}
  (1989) 669}.
%%CITATION = NUPHA,B320,669;%%

\bibitem{Giveon:1988tt}
A.~Giveon, E.~Rabinovici, and G.~Veneziano,
{\em {Duality in string background space},}
\href{http://dx.doi.org/10.1016/0550-3213(89)90489-6}{Nucl. Phys. B {\bf 322}
  (1989) 167}.
%%CITATION = NUPHA,B322,167;%%

\bibitem{Maharana:1992my}
J.~Maharana and J.~H.~Schwarz,
{\em Noncompact symmetries in string theory,}
\href{http://dx.doi.org/10.1016/0550-3213(93)90387-5}{Nucl.Phys B {\bf 390}
  (1993)~3} [{\tt
  \href{http://www.arxiv.org/abs/hep-th/9207016}{hep-th/9207016}}].

\bibitem{Narain:1985jj}
K.~S.~Narain,
{\em {New heterotic string theories in uncompactified dimensions $<10$},}
\href{http://dx.doi.org/10.1016/0370-2693(86)90682-9}{Phys. Lett. B {\bf 169}
  (1986)~41}.
%%CITATION = PHLTA,B169,41;%%

\bibitem{Hohm:2010pp}
O.~Hohm, C.~Hull, and B.~Zwiebach,
{\em Generalized metric formulation of double field theory,}
\href{http://dx.doi.org/10.1007/JHEP08(2010)008}{JHEP {\bf 1008} (2010)~8}
  [{\tt \href{http://www.arxiv.org/abs/1006.4823}{1006.4823 [hep-th]}}].

\bibitem{Sheng:1103.5920}
Y.~Sheng and C.~Zhu,
{\em Higher extensions of Lie algebroids,}
\href{http://dx.doi.org/10.1142/S0219199716500346}{Comm. Cont. Math. ~ {\bf }
  (2016) 1650034} [{\tt \href{http://www.arxiv.org/abs/1103.5920}{1103.5920
  [math-ph]}}].

\bibitem{JSTOR:1998201}
M.~Batchelor,
{\em The structure of supermanifolds,}
\href{http://dx.doi.org/10.2307/1998201}{Trans. Am. Math. Soc. {\bf 253} (1979)
  329}.

\bibitem{Bonavolonta:2012fh}
G.~Bonavolonta and N.~Poncin,
{\em {On the category of Lie n-algebroids},}
\href{http://dx.doi.org/10.1016/j.geomphys.2013.05.004}{J Geom. Phys. {\bf 73}
  (2013) 70–90} [{\tt \href{http://www.arxiv.org/abs/1207.3590}{1207.3590
  [math.DG]}}].
%%CITATION = ARXIV:1207.3590;%%

\bibitem{Roytenberg:0203110}
D.~Roytenberg,
{\em On the structure of graded symplectic supermanifolds and Courant
  algebroids,}
in: ``Quantization, Poisson Brackets and Beyond,'' ed.\ Theodore Voronov,
  Contemp. Math., Vol. 315, Amer. Math. Soc., Providence, RI, 2002
[{\tt \href{http://www.arxiv.org/abs/math.SG/0203110}{math.SG/0203110}}].

\bibitem{Roytenberg:0712.3461}
D.~Roytenberg,
{\em On weak Lie 2-algebras,}
in: ``XXVI Workshop on Geometrical Methods in Physics 2007,'' ed.\ Piotr
  Kielanowski et al., AIP Conference Proceedings volume 956, American Institute
  of Physics, Melville, NY
[{\tt \href{http://www.arxiv.org/abs/0712.3461}{0712.3461 [math.QA]}}].

\bibitem{Baez:2003aa}
J.~Baez and A.~S.~Crans,
{\em Higher-dimensional algebra VI: Lie 2-algebras,}
\href{http://tac.mta.ca/tac/volumes/12/15/12-15.pdf}{Th. App. Cat. {\bf 12}
  (2004) 492} [{\tt
  \href{http://www.arxiv.org/abs/math.QA/0307263}{math.QA/0307263}}].

\bibitem{Mehta:2012ppa}
R.~Mehta and M.~Zambon,
{\em $L_\infty$-algebra actions,}
\href{http://dx.doi.org/10.1016/j.difgeo.2012.07.006}{Diff. Geo. App. {\bf 30}
  (2012) 576} [{\tt \href{http://www.arxiv.org/abs/1202.2607}{1202.2607
  [math.DG]}}].

\bibitem{Fiorenza:0601312}
D.~Fiorenza and M.~Manetti,
{\em $L_\infty$ structures on mapping cones,}
\href{http://dx.doi.org/10.2140/ant.2007.1.301}{Alg. Numb. Th. {\bf 1} (2007)
  301} [{\tt
  \href{http://www.arxiv.org/abs/math.QA/0601312}{math.QA/0601312}}].

\bibitem{Getzler:1010.5859}
E.~Getzler,
{\em Higher derived brackets,}
{\tt \href{http://www.arxiv.org/abs/1010.5859}{1010.5859 [math-ph]}}.

\bibitem{Ritter:2015ffa}
P.~Ritter and C.~Saemann,
{\em {Automorphisms of strong homotopy Lie algebras of local observables},}
{\tt \href{http://www.arxiv.org/abs/1507.00972}{1507.00972 [hep-th]}}.
%%CITATION = ARXIV:1507.00972;%%

\bibitem{Baez:2008bu}
J.~C.~Baez, A.~E.~Hoffnung, and C.~L.~Rogers,
{\em {Categorified symplectic geometry and the classical string},}
\href{http://dx.doi.org/10.1007/s00220-009-0951-9}{Commun. Math. Phys. {\bf
  293} (2010) 701} [{\tt \href{http://www.arxiv.org/abs/0808.0246}{0808.0246
  [math-ph]}}].
%%CITATION = ARXIV:0808.0246;%%

\bibitem{MR1074539}
A.~M.~Vinogradov,
{\em The union of the {S}chouten and {N}ijenhuis brackets, cohomology, and
  superdifferential operators,}
Mat. Zametki {\bf 47} (1990) 138.

\bibitem{Gruetzmann:2014ica}
M.~Gruetzmann and T.~Strobl,
{\em {General Yang--Mills type gauge theories for p-form gauge fields: From
  physics-based ideas to a mathematical framework OR From Bianchi identities to
  twisted Courant algebroids},}
\href{http://dx.doi.org/10.1142/S0219887815500097}{Int. J. Geom. Meth. Mod.
  Phys. {\bf 12} (2014) 1550009} [{\tt
  \href{http://www.arxiv.org/abs/1407.6759}{1407.6759 [hep-th]}}].
%%CITATION = ARXIV:1407.6759;%%

\bibitem{Severa:2001qm}
P.~Severa and A.~Weinstein,
{\em {Poisson geometry with a 3-form background},}
\href{http://dx.doi.org/10.1143/PTPS.144.145}{Prog. Theor. Phys. Suppl. {\bf
  144} (2001) 145} [{\tt
  \href{http://www.arxiv.org/abs/math.SG/0107133}{math.SG/0107133}}].
%%CITATION = MATH/0107133;%%

\bibitem{Severa:1998ab}
P.~Severa,
{\em Letter No. 5 to Alan Weinstein (1998),}
available at \href{http://sophia.dtp.fmph.uniba.sk/~severa/letters/}{\ttfamily
  http://sophia.dtp.fmph.uniba.sk/$\sim$severa/letters/}.

\bibitem{Bressler:2002ur}
P.~Bressler and A.~Chervov,
{\em {Courant algebroids},}
\href{http://dx.doi.org/10.1007/s10958-005-0251-7}{J. Math. Sci. {\bf 128}
  (2005) 3030} [{\tt
  \href{http://www.arxiv.org/abs/hep-th/0212195}{hep-th/0212195}}].
%%CITATION = HEP-TH/0212195;%%

\bibitem{Roytenberg:0112152}
D.~Roytenberg,
{\em Quasi-Lie bialgebroids and twisted Poisson manifolds,}
\href{http://dx.doi.org/10.1023/A:1020708131005}{Lett. Math. Phys. {\bf 61}
  (2002) 123} [{\tt
  \href{http://www.arxiv.org/abs/math.QA/0112152}{math.QA/0112152}}].

\bibitem{Murray:9407015}
M.~K.~Murray,
{\em Bundle gerbes,}
\href{http://dx.doi.org/10.1112/jlms/54.2.403}{J. Lond. Math. Soc. {\bf 54}
  (1996) 403} [{\tt
  \href{http://www.arxiv.org/abs/dg-ga/9407015}{dg-ga/9407015}}].

\bibitem{Murray:2007ps}
M.~K.~Murray,
{\em {An introduction to bundle gerbes},}
in: {\em The many facets of geometry: A tribute to Nigel Hitchin,} eds.\ O.\
  Garcia--Prada, J.-P.\ Bourguignon and S.\ Salamon, Oxford University Press,
  Oxford, 2010
[{\tt \href{http://www.arxiv.org/abs/0712.1651}{0712.1651 [math.DG]}}].

\bibitem{Rogers:2010sc}
C.~L.~Rogers,
{\em 2-plectic geometry, Courant algebroids, and categorified prequantization,}
\href{http://dx.doi.org/10.4310/JSG.2013.v11.n1.a4}{J. Symp. Geom. {\bf 11}
  (2013)~53} [{\tt \href{http://www.arxiv.org/abs/1009.2975}{1009.2975
  [math-ph]}}].

\bibitem{Fiorenza:1304.6292}
D.~Fiorenza, C.~L.~Rogers, and U.~Schreiber,
{\em $L_\infty$-algebras of local observables from higher prequantum bundles,}
\href{http://dx.doi.org/10.4310/HHA.2014.v16.n2.a6}{Homol. Homot. App. {\bf 16}
  (2014) 107} [{\tt \href{http://www.arxiv.org/abs/1304.6292}{1304.6292
  [math-ph]}}].

\bibitem{kosmann1996poisson}
Y.~Kosmann-Schwarzbach,
{\em From Poisson algebras to Gerstenhaber algebras,}
\href{http://archive.numdam.org/article/AIF_1996__46_5_1243_0.pdf}{Ann. Inst.
  Fourier {\bf 46} (1996) 1243}.

\bibitem{Hull:2007zu}
C.~M.~Hull,
{\em {Generalised geometry for M-theory},}
\href{http://dx.doi.org/10.1088/1126-6708/2007/07/079}{JHEP {\bf 0707} (2007)
  079} [{\tt \href{http://www.arxiv.org/abs/hep-th/0701203}{hep-th/0701203}}].
%%CITATION = HEP-TH/0701203;%%

\bibitem{Berman:2011pe}
D.~S.~Berman, H.~Godazgar, and M.~J.~Perry,
{\em {$SO(5,5)$ duality in M-theory and generalized geometry},}
\href{http://dx.doi.org/10.1016/j.physletb.2011.04.046}{Phys. Lett. B {\bf 700}
  (2011)~65} [{\tt \href{http://www.arxiv.org/abs/1103.5733}{1103.5733
  [hep-th]}}].
%%CITATION = ARXIV:1103.5733;%%

\bibitem{Hohm:2010jy}
O.~Hohm, C.~Hull, and B.~Zwiebach,
{\em Background independent action for double field theory,}
\href{http://dx.doi.org/10.1007/JHEP07(2010)016}{JHEP {\bf 1007} (2010)~16}
  [{\tt \href{http://www.arxiv.org/abs/1003.5027}{1003.5027 [hep-th]}}].

\bibitem{Bouwknegt:2011vn}
P.~Bouwknegt and B.~Jurco,
{\em {AKSZ construction of topological open $p$-brane action and Nambu
  brackets},}
\href{http://dx.doi.org/10.1142/S0129055X13300045}{Rev. Math. Phys. {\bf 25}
  (2013) 1330004} [{\tt \href{http://www.arxiv.org/abs/1110.0134}{1110.0134
  [math-ph]}}].
%%CITATION = ARXIV:1110.0134;%%

\bibitem{Getzler:0404003}
E.~Getzler,
{\em Lie theory for nilpotent $L_\infty$-algebras,}
\href{http://dx.doi.org/10.4007/annals.2009.170.271}{Ann. Math. {\bf 170}
  (2009) 271} [{\tt
  \href{http://www.arxiv.org/abs/math.AT/0404003}{math.AT/0404003}}].

\bibitem{Henriques:2006aa}
A.~Henriques,
{\em Integrating $L_\infty$-algebras,}
\href{http://dx.doi.org/10.1112/S0010437X07003405}{Comp. Math. {\bf 144} (2008)
  1017} [{\tt
  \href{http://www.arxiv.org/abs/math.CT/0603563}{math.CT/0603563}}].

\bibitem{Severa:1506.04898}
P.~Severa and M.~Siran,
{\em Integration of differential graded manifolds,}
{\tt \href{http://www.arxiv.org/abs/1506.04898}{1506.04898}}.

\bibitem{Hohm:2012gk}
O.~Hohm and B.~Zwiebach,
{\em {Large gauge transformations in double field theory},}
\href{http://dx.doi.org/10.1007/JHEP02(2013)075}{JHEP {\bf 1302} (2013) 075}
  [{\tt \href{http://www.arxiv.org/abs/1207.4198}{1207.4198 [hep-th]}}].
%%CITATION = ARXIV:1207.4198;%%

\bibitem{Bouwknegt:2003vb}
P.~Bouwknegt, J.~Evslin, and V.~Mathai,
{\em T-Duality: Topology change from H-flux,}
\href{http://dx.doi.org/10.1007/s00220-004-1115-6}{Commun. Math. Phys. {\bf
  249} (2004) 383} [{\tt
  \href{http://www.arxiv.org/abs/hep-th/0306062}{hep-th/0306062}}].

\bibitem{Kachru:2002sk}
S.~Kachru, M.~B.~Schulz, P.~K.~Tripathy, and S.~P.~Trivedi,
{\em New supersymmetric string compactifications,}
\href{http://dx.doi.org/10.1088/1126-6708/2003/03/061}{JHEP {\bf 0303}
  (2003)~61} [{\tt
  \href{http://www.arxiv.org/abs/hep-th/0211182}{hep-th/0211182}}].

\bibitem{Mathai:2004qq}
V.~Mathai and J.~Rosenberg,
{\em T-duality for torus bundles via noncommutative topology,}
\href{http://dx.doi.org/10.1007/s00220-004-1159-7}{Commun. Math. Phys {\bf 253}
  (2004) 705} [{\tt
  \href{http://www.arxiv.org/abs/hep-th/0401168}{hep-th/0401168}}].

\bibitem{LiBland:2011aa}
D.~Li-Bland and P.~Severa,
{\em Integration of exact Courant algebroids,}
\href{http://dx.doi.org/10.3934/era.2012.19.58}{ERAMS {\bf 19} (2012)~58} [{\tt
  \href{http://www.arxiv.org/abs/1101.3996}{1101.3996 [math.DG]}}].

\end{thebibliography}

\bibliographystyle{latexeu}

\end{document}